\documentclass[a4paper,DIV12,11pt]{article}

\usepackage[a4paper, top=20mm, bottom =20mm, textwidth=165mm]{geometry}

\usepackage{amsthm,amsmath,amssymb,amsfonts,dsfont}
\usepackage{enumitem}

\usepackage{caption}  

\counterwithin{figure}{section} 

\usepackage[dvipsnames]{xcolor}

\usepackage{epsfig}
\usepackage{url,hyperref}
\usepackage{tikz}
\usetikzlibrary{calc,patterns,arrows,positioning,cd,calc}
\usetikzlibrary{decorations.pathreplacing,decorations.markings,shapes}
\usetikzlibrary{matrix}
\tikzset{
	on each segment/.style={
		decorate,
		decoration={
			show path construction,
			moveto code={},
			lineto code={
				\path [#1]
				(\tikzinputsegmentfirst) -- (\tikzinputsegmentlast);
			},
			curveto code={
				\path [#1] (\tikzinputsegmentfirst)
				.. controls
				(\tikzinputsegmentsupporta) and (\tikzinputsegmentsupportb)
				..
				(\tikzinputsegmentlast);
			},
			closepath code={
				\path [#1]
				(\tikzinputsegmentfirst) -- (\tikzinputsegmentlast);
			},
		},
	},
	mid arrow/.style={postaction={decorate,decoration={
				markings,
				mark=at position .5 with {\arrow[#1]{stealth}}
	}}},
}

\makeatletter \@addtoreset{equation}{section}

\makeatother

\makeatletter
\def\Maketitle{{\def\newpage{}\maketitle}}
\def\Appendix{\appendix
  \def\@seccntbformat##1{Appendix~\csname the##1\endcsname.~~}}
\makeatother

\tikzset{
	styleArrow/.style={postaction={decorate},decoration={markings,mark=at position 0.7 with {\arrow{Stealth}}}},
	blackCircle/.style={fill=black,thick,radius=0.1,inner sep=0},
	mid arrow/.style={postaction={decorate,decoration={
				markings,
				mark=at position .5 with {\arrow[#1]{stealth}}
	}}},
	whiteCircle/.style={fill=white,thick,radius=0.1,inner sep=0}
}

\tikzset{styleNode/.style={
		circle,draw,inner sep=1
	}
}

\tikzset{styleNodeFr/.style={
		rectangle,draw,inner sep=2
	}
}

\tikzset{
	special arrow/.style={
		decoration={
			markings,
			mark=at position #1 with {\arrow{Stealth}}
		},
		postaction={decorate}
	}
}

\newtheorem{Proposition}{Proposition}[section]
\newtheorem{Lemma}[Proposition]{Lemma}
\newtheorem*{Lemma*}{Lemma}
\newtheorem{Conjecture}[Proposition]{Conjecture}
\newtheorem{Theorem}[Proposition]{Theorem}

\newtheorem*{Theorem*}{Theorem}
\newtheorem{Corollary}[Proposition]{Corollary}

\theoremstyle{definition}
\newtheorem{Definition}[Proposition]{Definition}
\newtheorem{Notation}[Proposition]{Notation}

\newtheorem{Remark}[Proposition]{Remark}
\newtheorem{Example}[Proposition]{Example}

\def \sgn{\operatorname{sgn}}
\def \wt{\operatorname{wt}}

\title{Cluster Reductions, Mutations, and $q$-Painlev\'e Equations}

\author{Mikhail Bershtein \and Pavlo Gavrylenko \and Andrei Marshakov \and Mykola Semenyakin}

\date{}

\begin{document}


\Maketitle

\begin{flushright}
	\emph{In memory of Igor Krichever}
\end{flushright}
\vspace{0.5em}

\begin{abstract}
	We propose an extension of the Goncharov-Kenyon class of cluster integrable systems by their Hamiltonian reductions. This extension allows us to fill in the gap in cluster construction of the $q$-difference Painlev\'e equations, showing that all of them can be obtained as deatonomizations 
	of the reduced Goncharov-Kenyon systems.
	
	Conjecturally, the isomorphisms of reduced Goncharov-Kenyon integrable systems are given by mutations in another, dual in some sense, cluster structure.  These are the polynomial mutations of the spectral curve equations and polygon mutations of the corresponding decorated Newton polygons.
	In the Painlev\'e case the initial and dual cluster structures are isomorphic. It leads to self-duality between the spectral curve equation and the Painlev\'e Hamiltonian, and also extends the symmetry from affine to elliptic Weyl group.
\end{abstract}

\tableofcontents

\section{Introduction}

There are two ways to think about the subject and results of the paper. 

On one hand, we study the relation between cluster varieties and \(q\)-difference Painlev\'e equations, previously studied in \cite{Okubo:2015bilinear,Okubo:2017coprimeness,Bershtein:2018cluster,Mizuno,Suzuki:2020}. In particular, in \cite{Bershtein:2018cluster} the affine Weyl groups, which are symmetry groups of the \(q\)-Painlev\'e equations,  were realized  via mutations and permutations of cluster variables, i.e. as subgroups in the cluster mapping class groups. 
Furthermore, all \(q\)-Painlev\'e equations, except two mostly generic with \(E_7\) and \(E_8\) symmetries,  were obtained as deatonomizations of the Goncharov-Kenyon (GK) cluster integrable systems \cite{Goncharov:2013}.  
In this paper we fill this gap, namely we show that all \(q\)-Painlev\'e equations (including those two) can be obtained as deatonomizations of \emph{Hamiltonian reductions} of GK integrable systems. We also observe remarkable \emph{self-duality} between spectral curve and Hamiltonian for these (i.e. corresponding to Painlev\'e) integrable systems. Using this self-duality, we extend affine Weyl groups to \emph{elliptic} (or \emph{double affine}) Weyl groups \cite{Saito:1997} acting on dynamical variables and spectral parameters.

On the other hand, we initiate the study of \emph{cluster} Hamiltonian reductions. The cluster Poisson structure is a covering of a Poisson variety by charts with the Darboux-like coordinates and rational transition maps. Cluster coordinates simplify computations and allow to think about the problems of different nature in a unified framework. The Hamiltonian reduction is a powerful method to construct Poisson varieties and integrable systems on them. It is natural to ask whether Hamiltonian reduction of a cluster Poisson variety has natural cluster structure, and we conjecture that there exists a corresponding class of Hamiltonian reductions of the Goncharov-Kenyon systems. 

The original GK integrable systems are labeled by integral convex polygons~\(N\) while the reduced GK integrable systems are labeled by decorated polygons, where for any side of integral length \(n\) we assign a partition of \(n\).
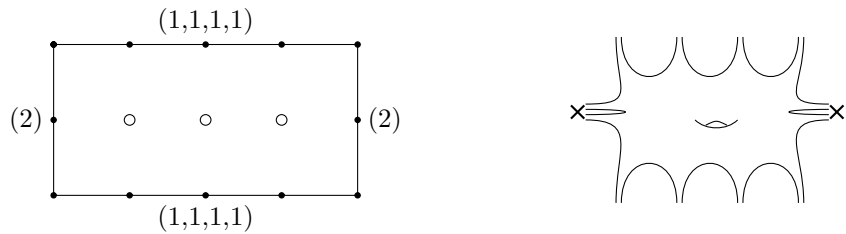
\begin{figure}[h]
	\begin{center}
		\begin{tikzpicture}[font=\small]
			\begin{scope}[scale=1]
				
				\draw[fill] (-2,1) circle (1pt) -- (-2,0) circle (1pt) -- (-2,-1) circle (1pt) -- (-1,-1) circle (1pt) -- (0,-1) circle (1pt) -- (1,-1) circle (1pt) -- (2,-1) circle (1pt) -- (2,0) circle (1pt) 
				-- (2,1) circle (1pt) -- (1,1) circle (1pt)  -- (0,1) circle (1pt) -- (-1,1) circle (1pt) -- (-2,1) circle (1pt);
				\draw (0,0) circle (2pt);
				\draw (1,0) circle (2pt);
				\draw (-1,0) circle (2pt);
				
				\node[above] at (0,1) {(1,1,1,1)};
				\node[below] at (0,-1) {(1,1,1,1)};
				\node[right] at (2,0) {(2)};
				\node[left] at (-2,0) {(2)};
				
			\end{scope}
			
			\begin{scope}[scale=1,shift={(5,-1.1)},font = \small]
				\def\xs{0.8}
				\def\ys{1.1}
				\def\eps{0.07}		
				\def\curveheight{0.7}  
				
				\foreach \i in {0,...,2}
				{
					\draw (0.5*\xs+\i*\xs+\eps,0) .. controls (0.5*\xs+\i*\xs+\eps,\curveheight) and (1.5*\xs+\i*\xs,\curveheight) .. (1.5*\xs+\i*\xs,0);
					\draw (0.5*\xs+\i*\xs+\eps,2*\ys) .. controls (0.5*\xs+\i*\xs+\eps,2*\ys-\curveheight) and (1.5*\xs+\i*\xs,2*\ys-\curveheight) .. (1.5*\xs+\i*\xs,2*\ys);					
				}
				\draw (0,\ys) .. controls (\curveheight,\ys) and ((0.5*\xs,\curveheight) .. ((0.5*\xs,0);		
				\draw (0,\ys+\eps) .. controls (\curveheight,\ys+\eps) and (\curveheight,\ys+2*\eps) .. (0,\ys+2*\eps);
				\draw (0,\ys+3*\eps) .. controls (\curveheight,\ys+3*\eps) and (0.5*\xs,2*\ys-\curveheight) .. (0.5*\xs,2*\ys);
				
				\draw (4*\xs,\ys) .. controls (4*\xs-\curveheight,\ys) and ((3.5*\xs+\eps,\curveheight) .. ((3.5*\xs+\eps,0);		
				\draw (4*\xs,\ys+\eps) .. controls (4*\xs-\curveheight,\ys+\eps) and (4*\xs-\curveheight,\ys+2*\eps) .. (4*\xs,\ys+2*\eps);
				\draw (4*\xs,\ys+3*\eps) .. controls (4*\xs-\curveheight,\ys+3*\eps) and (3.5*\xs+\eps,2*\ys-\curveheight) .. (3.5*\xs+\eps,2*\ys);
				
				\newcommand{\cross}[2]{
					\draw[thick] (#1) -- ++(#2*\eps, #2*\eps);   
					\draw[thick] (#1) -- ++(-#2*\eps, -#2*\eps); 
					\draw[thick] (#1) -- ++(#2*\eps, -#2*\eps);  
					\draw[thick] (#1) -- ++(-#2*\eps, #2*\eps);  
				}
				\cross{-1.5*\eps,\ys+1.5*\eps}{1.25};
				\cross{4*\xs+1.5*\eps,\ys+1.5*\eps}{1.25};
				
				\draw[smooth] (1.8*\xs,\ys) .. controls (1.8*\xs+2*\eps,\ys-2*\eps) and (2.5*\xs-2*\eps,\ys-2*\eps) .. (2.5*\xs,\ys);
				\draw[smooth] (1.8*\xs+2*\eps,\ys-1*\eps) .. controls (1.8*\xs+4*\eps,\ys+0*\eps) and (2.5*\xs-4*\eps,\ys+0*\eps) .. (2.5*\xs-2*\eps,\ys-1*\eps);		
			\end{scope}				
		\end{tikzpicture}
%
%
	\end{center}
    	\caption{On the left decorated polygon, on the right   corresponding spectral curve.}
\end{figure}
The spectral curve of the corresponding integrable system has multicross singularities at infinity in the toric embedding. Conjecturally, the dimension of the phase space \(\mathcal{X}_{\mathrm{red}}\) of integrable system and the rank of Poisson bracket are 
\begin{gather}
		\dim (\mathcal{X}_{\mathrm{red}})=2 \operatorname{Area}(N)-1-\sum\nolimits_{h \in \text{parts of all partitions}} (h^2-1), \label{eq:xred intro}\\
		\operatorname{rank}\{\cdot,\cdot\}=2I-\sum\nolimits_{h \in \text{parts of all partitions}}h(h-1), \label{eq:P rank intro}
\end{gather}
where the summation runs over all parts of all partition in the decoration and \(I\) denotes number of integral points inside the polygon \(N\). Moreover, we conjecture that reduced GK systems possess natural cluster structures.

The decoration of \(N\) can be motivated by the analogy with moduli spaces of framed $G$-local systems on punctured surfaces. Assume for simplicity that \(G=PGL_n(\mathbb{C})\) and the surface is a sphere with \(k\) punctures. Then the local systems are parametrized by \(k\)-tuples of matrices \(M_1,\dots,M_k \in PGL_n (\mathbb{C})\) up to conjugation with the constraint \(M_1\cdot \dots\cdot M_k=1\). The moduli spaces of local systems are labeled by conjugacy classes of matrices \(M_i\). If these matrices are semisimple, then to each class we can assign a partition of \(n\). In analogy to GK systems, the punctures correspond to the sides of the polygon \(N\), so we have partitions assigned to sides. In case of generic monodromies, the cluster Poisson structure on the moduli space of framed local system was constructed in \cite{Fock:2006moduli}. It is expected that moduli spaces corresponding to non-generic punctures can be obtained from the \cite{Fock:2006moduli} moduli spaces via \emph{Hamiltonian reduction}.

In this paper, we mostly discuss reductions performed along one side of \(N\) in GK setting. This corresponds to one special puncture in the setting of local systems. Even in this one-sided case we mainly omit proofs; they will appear in a separate publication in greater generality. The multi-sided reduction is much less understood; for example, it is unclear for which decorations of given Newton polygon the moduli space $\mathcal{X}_{\mathrm{red}}$ is nonempty. Perhaps this can be viewed as analog of Deligne-Simpson problem \cite{Simpson:1992products} \cite{Crawley:2004indecomposable}.

The step from GK integrable systems to reduced GK integrable systems allows to fill a mentioned above gap in the Painlev\'e theory. Painlev\'e equations are second-order equations on one variable, therefore, the rank of Poisson bracket for the corresponding integrable systems should be equal to~2. Assuming that there are no reductions (i.e. all \(h=1\)) and taking into account formula \eqref{eq:P rank intro} we get \(I=1\). Hence Painlev\'e equations which can be obtained from GK integrable systems without reduction boil down to reflexive polygons \cite{Bershtein:2018cluster}. 

In this paper we consider more general class of polygons which was studied in \cite{Kasprzyk:2017minimality} under the name Fano polygons without remainders. Its relevance to the Painlev\'e theory was noted by Mizuno \cite{Mizuno}. For each such polygon one can naturally assign decoration and realize symmetry group of the corresponding \(q\)-Painlev\'e equation via cluster mutations. 

It  was very important for our work that connection between cluster mutations and polygons in \cite{Mizuno} is different from the one in GK integrable systems. Namely, the cluster mutations used in the paper \cite{Mizuno} \emph{mutate the polygon} contrary to mutations in \cite{Goncharov:2013} which preserve the polygon.  We claim that latter class of mutations should be viewed as cluster mutations of \(\mathcal{X}_{\mathrm{red}}\), while former class provides isomorphisms between \(\mathcal{X}_{\mathrm{red}}\) and \(\widetilde{\mathcal{X}}_{\mathrm{red}}\) that correspond to mutations of the decorated polygons. 

These two classes of mutations can be realized geometrically in terms of consistent bipartite graph \(\Gamma\) on a torus, which is the main combinatorial ingredient of the GK integrable system. Mutations used in \cite{Goncharov:2013} are transformations assigned to 4-gon faces (so-called spider moves), see Fig.~\ref{fi:Example} left. Convex polygon \(N\) labeling the GK integrable integrable system is a Newton polygon of a dimer partition function of \(\Gamma\). The polygon \(N\) is invariant under the face mutations. On the other hand, it was shown in \cite{Higashitani:2022} that there is another class of transformations of bipartite graphs that mutate polygon~\(N\). These transformations are assigned to zigzag paths, which are dual to faces in some sense (see Sec. \ref{ssec:PoI}).
We will call such transformations \emph{zigzag mutations}; they were studied recently in \cite{Higashitani:2022,Franco:2023twin,Franco:2023quiver,Cremonesi:2023zig}. The example of mutation for zigzag of length 4 is given on Fig.~\ref{fi:Example} right.

\begin{figure}[h]
	\begin{center}
	\includegraphics[scale=1]{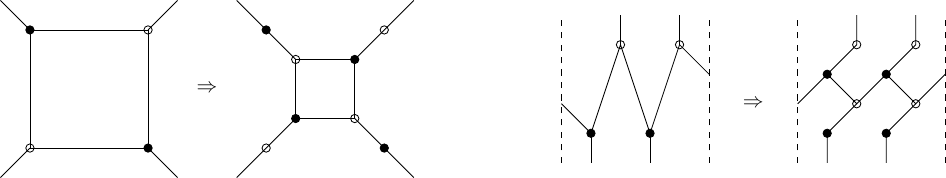}
\end{center}
	\caption{ \label{fi:Example} On the left face mutation, on the right zigzag mutation, drawn on the cylinder} 
\end{figure}

Hence, the results of \cite{Mizuno} suggest that there are two cluster descriptions of Painleve dynamics: one in terms of face mutations and another in terms of zigzag mutations. This is the main reason for the results about the Painlevé equations with which we started the introduction: self-duality and enhancement of affine Weyl groups to elliptic Weyl groups.  

Note that, while cluster mutations can be assigned to any face, only mutations corresponding to \emph{4-gon faces} have description in terms of bipartite graph on a torus. In fact, the length 4 restriction is also important for zigzag mutations. Combinatorially, in terms of the graph \(\Gamma\), zigzag mutation can be defined for any length of zigzag path~\cite{Higashitani:2022}, but if one promotes it to transformation of the cluster variables, then for the length greater than 4 this map is defined only on the submanifold, which is not Poisson and hence not cluster. At this point the Hamiltonian reductions become essential, namely, our claim is that for the length greater than 4 the zigzag mutation should be viewed as an isomorphism between cluster Hamiltonian reductions \(\mathcal{X}_{\mathrm{red}}\) and \(\widetilde{\mathcal{X}}_{\mathrm{red}}\). 

On the side note, one can dualize mutations of zigzags with the length greater than 4 and (at least combinatorially) define mutations of faces with more than 4 sides. These face mutations change the genus of surface, pushing us beyond the class of graphs on tori.



\paragraph{Plan of the paper} In Section \ref{sec:GK} we recall necessary details about dimer models, Goncharov-Kenyon integrable systems, and cluster varieties. In section \ref{sec:zigzag reductions} we discuss reductions and zigzag mutations. We announce here the main statements about the geometry of zigzag mutations and reductions of cluster integrable systems, though their proofs are postponed to the paper in progress~\cite{Inprog}. We believe this section is useful to give the right context for the next one, even though most of the statements are not formally used there. 

In Section \ref{sec:Painleve} we apply this to the study of \(q\)-Painlev\'e equations. Here all results are presented with complete proofs, whereby some of them require case-by-case considerations; the details about them are provided in subsequent sections \ref{sec:examples}, \ref{sec:E7}, \ref{sec:E8}. Sections \ref{sec:E7} and \ref{sec:E8} also provide us with non-trivial examples of reduced GK systems, i.e. nontrivial examples of the results and conjectures stated in Section~\ref{sec:zigzag reductions}.

\paragraph{Acknowledgment} We are grateful to V.~Fock, D.~Rachenkov, Y.~Saito, G.~Schrader, A.~Shapiro for useful discussions.  The results of the paper were reported at conferences and seminars in Moscow (May-June 2022), Trieste (June 2022), Sendai/Zoom (November 2022), Japan/Zoom (March 2023),  Paris/Zoom (May 2023), Beijing (July-August 2023), Sheffield (October 2023), Angers (April 2024), Saint Petersburg/Zoom (September 2024), Chernogolovka/Zoom (October 2024), Montreal (October 2024), Perimeter (October 2024); we are grateful to the organizers and participants for their interest and remarks.

The work of M.B. is partially supported by the European Research Council under the European Union’s Horizon 2020 research and innovation programme under grant agreement No 948885. M.B. is very grateful to Kavli IPMU  and especially Y.~Fukuda, M.~Kapranov, K.~Kurabayashi, H.~Nakajima, A.~Okounkov, T.~Shiga, K.~Vovk, for the hospitality during 2022-2023 years. 

\pagebreak

\section{Dimer models and Goncharov-Kenyon integrable systems} \label{sec:GK}

\subsection{Consistent bipartite graphs}

Let \(\Gamma\) be a bipartite graph on a real torus \(\Sigma=\mathbb{T}^2\) with vertices colored in black and white. Connected components of \(\Sigma \setminus \Gamma\) are called faces, we assume that all faces are contractible. Equivalently one can consider its preimage \(\widehat{\Gamma}\) on the universal cover of \(\Sigma\), which is the periodic bipartite graph on a plane, typical examples of \(\widehat{\Gamma}\) are square and hexagonal grids.


If there is a vertex \(v \in V(\Gamma)\) of valency 2 that is connected with two different vertices \(v_1\) and \(v_2\) one can delete the vertex \(v\) together with the edges \(vv_1\) and \(v v_2\), gluing \(v_1\) and \(v_2\). This operation is called \emph{contraction}, the inverse operation is called uncontraction. We call graphs related by the sequence of contractions and uncontractions to be equivalent.

Let \(V(\Gamma), E(\Gamma)\), and \(F(\Gamma)\) denote the sets of vertices, edges, and faces of  \(\Gamma\) correspondingly. For any face \(f \in F(\Gamma)\) one can consider the path  which goes counterclockwise (equivalently the path which turns left at any vertex) around it, we denote this path by the same letter \(f\). 

Another important set of paths on \(\Gamma\) consists of zigzags, we denote this set by \(Z(\Gamma)\). The zigzag \(\zeta\in Z(\Gamma)\) is a path, which turns \emph{right at black vertices} and \emph{left at white vertices}. It is usually convenient to draw a zigzag as a path that goes through the middles of the edges.
\begin{figure}[h]
    \centering
    \includegraphics[]{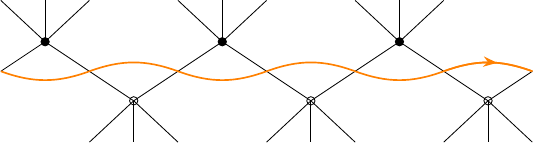}
	\caption{Zigzag path}
\end{figure}
Since graph \(\Gamma\) is finite any zigzag \(\zeta\) should be closed. Similarly one can define zigzags \(\hat{\zeta}\in Z(\hat{\Gamma})\) on the graph \(\hat{\Gamma}\).

\begin{Definition}\label{def:consistancy}
	The bipartite graph on a torus is called \emph{consistent} if it satisfies the following conditions
	\begin{enumerate}[label=(\alph*)]
		\item \label{it:1} Any zigzag \(\zeta\) represents nontrivial homology class \([\zeta]\neq 0 \in H_1(\Sigma)\)
		
		\item \label{it:2} There is no parallel \emph{bigons} on the universal cover, namely any two zigzags \(\hat{\zeta}_1, \hat{\zeta}_1,\) do not have pair of intersections, such that both paths go in the same direction from one intersection to the other.  
			
		\item \label{it:3} Any zigzag \(\hat{\zeta}\) on the universal cover does not have self-intersections. 
	\end{enumerate}
\end{Definition}

Here, we follow the terminology of \cite{Ishii:2011note}, essentially this is equivalent to the \emph{minimal bipartite graphs} in \cite{Goncharov:2013}. Actually the notion of intersection of two zigzag paths is a bit subtle in case of vertices of valency 2, see, e.g. \cite[Def. 3.4]{Ishii:2011note}. However, it follows from consistency conditions that any two-valent vertex is connected with two \emph{different} vertices and therefore can be 
contracted~\footnote{Let us sketch the proof of this property. Assume that black vertex \(b\) has valency 2 and both edges \(e_1, e_2\) connects \(b\) and white vertex \(w\). Consider closed path \(\gamma\) of length 2 which consist of edges \(e_1\) and \(e_2\). If \(\gamma\) is a face then we have homologically trivial zigzag. If the path \(\gamma\) is homologically trivial, but not contractible, then we have zigzag with self-intersection. Finally, if the path \(\gamma\) represents a nontrivial cycle in homology, than any zigzag has zero intersection with \(\gamma\). Hence for any zigzag its homology class is parallel or anti-parallel to \([\gamma]\). If there is a vertex of \(\Gamma\) of valency at least three, then there are two parallel zigzags that intersect at this point, that contradicts consistency. Finally, if all vertices of \(\Gamma\) have valency 2, then the connected components of \(\Sigma \setminus \Gamma\) are not contactable.}.
After a sequence of such contractions, one can get rid of such vertices and impose conditions \ref{it:2} and \ref{it:3} above.

It follows from the consistency conditions that any zigzag has no self-intersections on \(\Sigma\). Therefore, its homology class \([\zeta]\in H_1(\Sigma)\) is primitive. We call two zigzags \(\zeta_1,\zeta_2\) parallel if \([\zeta_1]=[\zeta_2]\). More generally we will say that two vectors \(u_1,u_2 \in \mathbb{R}^2\) are parallel if \(u_1=ku_2\) with \( k > 0\). 

It also follows from the consistency condition that parallel zigzags do not intersect in \(\Sigma\). Hence, for any primitive vector \(u \in \mathbb{Z}^2 \) there is a natural cyclic order on the set of zigzags parallel to \(u\).

\subsection{Dimer models}

The dimer model is a pair of bipartite graph \(\Gamma\) and weight function \([\wt]\in H^1(\Gamma,\mathbb{C}^*)\). 

Alternatively, the weight function is a map \(E(\Gamma) \rightarrow \mathbb{C}^*\), \(e \mapsto \wt(e)\) defined up to gauge transformations \(g \colon V(\Gamma) \rightarrow \mathbb{C}^*\), acting by \(\wt(e) \mapsto g(b)\wt(e)g(w)^{-1} \), where \(e=bw\) and \(b\) is black vertex and \(w\) is white vertex. Then for any path \(\gamma=(e_1,\dots, e_k)\), with \(e_j\in E(\Gamma)\) we define its weight to be \(\wt(\gamma)=\prod_{j=1}^k \wt(e_j)^{\pm 1}\), where sign is plus if the path goes through an edge from black to white vertex, and minus in the opposite case.  Clearly 
\(\wt(\gamma)\) is gauge invariant for any closed path \(\gamma\).

\begin{Notation} 
We will use several sets of variables in what follows:
	\begin{itemize}		
		\item  For any face \(f_i \in F(\Gamma)\) we define \emph{face variable} \(x_i=\wt(f_i)\) (recall that path \(f_i\) goes counterclockwise). Let \(\mathbf{x}\) denotes tuple \((x_1,\dots, x_{|F(\Gamma)|})\).
		
		\item For any zigzag path \(\zeta_j\) we define \emph{zigzag variable} \(z_j=(-1)^{|\zeta|/2-1}\sgn_K(\zeta)\wt(\zeta_j)\), where $|\zeta|$ is the length of $\zeta$ and $\sgn_K$ is Kasteleyn sign defined in Def.~\ref{def:Kasteleyn}. Let \(\mathbf{z}\) denotes tuple \((z_1,\dots, z_{|Z(\Gamma)|})\).
		
		\item For given symplectic basis \([A],[B]\in H_1(\Sigma)\) and given paths \(A\) and \(B\) which represent these cycles one can define \emph{spectral parameters} \(\lambda = \wt(A)\) and \(\mu = \wt(B)\). 
	\end{itemize}
\end{Notation}

Note that definition of \((\lambda, \mu)\) depends on choice of the paths \((A,B)\). If one modifies particular representatives \((\lambda,\mu)\) get multiplied by some monomials in face variables \(\mathbf{x}\)
. Under the change of symplectic basis \([A],[B]\in H_1(\Sigma)\) the spectral parameters are transformed as \(\lambda \mapsto \lambda^a\mu^b, \mu \mapsto \lambda^c\mu^d\), where \(\begin{pmatrix}
a & b \\ c &d
\end{pmatrix}\in SL(2,\mathbb{Z})\). 

Clearly, boundaries of the faces \(\{f_i\}\) and the paths \(A,B\) generate \(H_1(\Gamma)\). Therefore, for any closed path~\(\gamma\) its weight \(\wt(\gamma)\) is a monomial in face variables  \(\mathbf{x}\) and spectral parameters \((\lambda,\mu)\).

\begin{Definition}
	\emph{Dimer cover} (equivalent notion is \emph{perfect matching}) on graph \(\Gamma\) is a subset \(D \subset E(\Gamma)\) such that for any vertex \(v \in V(\Gamma)\) there exists unique edge \(e \in D\) incident to \(v\).
\end{Definition}
 It was shown in \cite[Lemma 3.11]{Goncharov:2013}, \cite[Prop. 7.1]{Ishii:2009:2} that consistent bipartite graphs possess dimer covers. For any two dimer covers \(D,D'\) their difference  \(D-D'\) is a union of cycles, so the weight \(\wt(D-D')=\wt(D) \wt(D')^{-1}\) is a monomial in \(\mathbf{x}, \lambda,\mu\). 

\begin{Definition} \label{def:Kasteleyn}
	The \emph{Kasteleyn sign} is a map \(\sgn_K \colon E(\Gamma)\rightarrow \{\pm1\} \) such that for any face 
	\begin{equation}
		\sgn_K(f)=\prod\nolimits_{e\in f} \sgn_K(e)=(-1)^{|f|/2-1}.
	\end{equation}	
\end{Definition}
Here and below by \(|\gamma|\) we denote the length of the path \(\gamma\). The Kasteleyn sign exists if and only if the number of vertices \(V(\Gamma)\) is even (see e.g. \cite[Th. 3.1]{Cimasoni:2007dimers}). Since for any consistent bipartite graph there exists dimer cover this condition is satisfied. Two Kasteleyn signs are called equivalent if one can be obtained from another by sequence of transformations which reverse  orientations of all edges adjacent to a vertex. For a graph on surface of genus \(g\) there are \(2^{2g}\) classes of equivalent Kasteleyn signs (see \cite[Th. 3.2]{Cimasoni:2007dimers}), so in our case there will be 4 equivalence classes.

The Kasteleyn operator \(\mathcal{K}\) maps from the vector space \(V_B\) with a basis labeled by black vertices to the vector space \(V_W\) with a basis labeled by white vertices. Its matrix elements are given by 
\begin{equation}
	\mathcal{K}_{wb}=\sum_{e \text{ connects } b \text{ and } w} \sgn_K(e) \wt(e).
\end{equation}
The dimer partition function is defined as a determinant
\begin{equation}
	\label{eq:ZKast}
	\mathcal{Z}(\mathbf{x}|\lambda,\mu) \sim \det \mathcal{K}.
\end{equation}
Here and below \(\sim\) stands for equality up to a monomial factor, in what follows the dimer partition function is always considered up to such normalization.
It is easy to see that the summands in \(\det \mathcal{K}\) correspond to dimer covers. Hence for any chosen dimer cover \(D_0\) the normalized determinant \(\wt(D_0)^{-1} \det \mathcal{K}\) is a Laurent polynomial in \(\mathbf{x}, \lambda,\mu\).

%

\begin{Definition}
	The quadratic form on \(H_{1}(\Sigma)\) is a map \(q\colon H_{1}(\Sigma) \rightarrow \mathbb{Z}/2\mathbb{Z}\) such that 
	\begin{equation}
		q(x+y)=q(x)+q(y)+x\cdot y,
	\end{equation}
	where \(x\cdot y\) denotes intersection pairing.	
\end{Definition}
For a surface of genus \(g\) there are \(2^{2g}\) quadratic forms, so there are 4 of them in our case.

\begin{Theorem}[{\cite{Kasteleyn:1963}, \cite[Th. 5.1]{Cimasoni:2007dimers}}] 
	Let \(D_0\) be a dimer cover and \(\sgn_K\) be a Kasteleyn sign. Then 
	\begin{equation} \label{eq:Z = sum D}
		\mathcal{Z}(\mathbf{x}|\lambda,\mu) \sim  \sum_{\alpha \in H_1(\Sigma)} (-1)^{q_{K,D_0}(\alpha)}\sum_{D,\, [D-D_0]=\alpha} \wt(D-D_0), 
	\end{equation}
	where \(q_{K,D_0}\) is a quadratic form. 
	
	Moreover,  the map \(K \mapsto q_{K,D_0}\) is a one to one correspondence between the set of equivalence classes of Kasteleyn signs and set of quadratic forms, this correspondence depends on choice of \(D_0\).
\end{Theorem}

If \([D-D_0]=a[A]+b[B] \in H_1(\Sigma)\) then \(\wt(D-D_0)=\lambda^a\mu^b m_x\) where \(m_x\) is a monomial in face variables. Thus, we can write the partition function as 
\begin{equation}
	\label{eq:Zcurve}
	\mathcal{Z}(\mathbf{x}|\lambda,\mu)= \sum_{(a,b)\in \mathbb{Z}^2} \lambda^a \mu^b \mathcal{Z}_{a,b}(\mathbf{x}).
\end{equation}
The \emph{Newton polygon} \(N\) of \(\mathcal{Z}(\lambda,\mu)\) is a convex hull of \((a,b)\in \mathbb{Z}^2\) such that \(\mathcal{Z}_{a,b}\neq 0\). Change of symplectic basis in \(H_1(\Sigma)\) leads to \(SL(2,\mathbb{Z})\) transformation of \(N\). Moreover, since the partition function is defined up to a monomial factor, the actual freedom in the definition of \(N\) is a group of affine transformations \(SA(2,\mathbb{Z})=SL(2,\mathbb{Z})\ltimes \mathbb{Z}^2\). The Newton polygon, or more precisely its orbit under \(SA(2,\mathbb{Z})\)-action is an important invariant of a consistent dimer model.

By consistent dimer models we call a pair: consistent bipartite graph \(\Gamma\) and \([\wt]\in H^1(\Gamma,\mathbb{C}^*	)\). Note that if two graphs \(\Gamma, \Gamma'\) are related by contraction of 2-valent vertices, there is a natural identification of \(H_1(\Gamma)\) and \(H_1(\Gamma')\), so one can identify the corresponding dimer models. 

There is a deep relation between zigzag paths and Newton polygon for consistent dimer models. Let \(E\) be a oriented side of \(N\), where the boundary of \(N\) is oriented counterclockwise, note that the vector \(E\) is not necessary primitive. Let us define 
\begin{equation}
	\mathcal{Z}(\mathbf{x}|\lambda,\mu)|_E=\sum\nolimits_{(a,b)\in E} \lambda^a \mu^b \mathcal{Z}_{a,b}(\mathbf{x}).
\end{equation}
 
\begin{Theorem}[{\cite[Sec. 5.2]{George2022inverse}, \cite[Cor. 4.27, Prop. 4.35]{Broomhead:2010}}]  \label{th:zigzags boundary}
	For any zigzag \(\zeta\) there is a side~\(E\) such that~\(E\) is parallel to \([\zeta]\). On the other hand for any side \(E\) we have 
	\begin{equation}\label{eq:Z|E prod}
		\mathcal{Z}(\mathbf{x}|\lambda,\mu)|_E \sim \prod_{ \zeta_j\in Z(\Gamma),\; [\zeta_j] \text{ parallel to } E} (1+z_j ).
	\end{equation}
\end{Theorem}
In particular this theorem means that there exist a one-to-one correspondence between the set of homology classes of zigzag paths on a consistent dimer model \(\Gamma\) and the set of primitive side segments of the Newton polygon \(N\). Also since \(\sim\) stands for monomial factor this theorem implies that if \((a,b)\) is a vertex of \(N\) then \(\mathcal{Z}_{a,b}\) is a monomial, i.e. there exists only one dimer configuration corresponding to \((a,b)\).

This theorem can be also stated in more geometric way. 

\begin{Definition} \label{def:spectral curve}
	The \emph{spectral curve} is a compactification \(\overline{\mathcal{C}}\) of the curve \(\mathcal{C}\) defined by an equation \(\mathcal{C}= \{(\lambda,\mu)|\mathcal{Z}(\lambda,\mu)=0\} \subset \mathbb{C}^*\times \mathbb{C}^*\).
\end{Definition}

The compactification \(\overline{\mathcal{C}}\) can be embedded into toric surface assigned to polygon \(N\) (see e.g \cite[Sec. 2.6]{George2022inverse}).
The intersection of \(\overline{\mathcal{C}}\) with divisor corresponding to the side \(E\) corresponds to roots of \(\mathcal{Z}(\mathbf{x}|\lambda,\mu)|_E\). Hence the theorem says that the number of such points is equal to integer length \(|E|_\mathbb{Z}\) of side \(E\) and these points are in one to one correspondence with zigzags parallel to \(E\).

Let \(I=I(N)\) denotes the number of integer points inside of \(N\) and \(B=B(N)\) denotes the number of integer points on the boundary of \(N\). Then the  number of points at infinity \(|\overline{\mathcal{C}}\setminus \mathcal{C}|=B\), and by previous theorem it is equal to the number of zigzags \(|Z(\Gamma)|=B\). Also, standard results say that genus of \(\overline{\mathcal{C}}\) is equal to \(g(\overline{C})=I\) in case of generic values of face variables~\(\mathbf{x}\) (see e.g. \cite[Theorem 1]{Khovanskii:1978}).

\subsection{Poisson structure. Integrability}
\label{ssec:PoI}

One can thicken graph \(\Gamma\) to make a ribbon graph, topologically this ribbon graph is a surface \(\Sigma=\mathbb{T}^2\) with \(F(\Gamma)\) holes. For a ribbon graph the cyclic order of the edges at each vertex is fixed. Let us define a \emph{dual bipartite ribbon graph} \(\Gamma^D\) by reversing the cyclic order at all black vertices \cite{Feng:2008dimer}, \cite{Goncharov:2013}. This dual bipartite ribbon graph is topologically a \emph{dual surface} \(\Sigma^D\) with holes, and it is easy to see that these holes correspond to zigzags on \(\Gamma\). This can be viewed as another useful way to think about definition of the dual surface \(\Sigma^D\), namely one just glues a disk to the graph \(\Gamma\) along each zigzag path, and these discs glued along \(\Gamma\) form a dual surface.

\begin{Theorem}[{\cite[Theorem 3.1]{Gulotta:2008properly}, \cite[Prop 3.15]{Goncharov:2013}}]  \label{th: number of faces}
	For consistent dimer model the number of faces \(|F(\Gamma)|\) is equal to \(2\operatorname{Area}(N)\).
\end{Theorem}

This theorem allows to compute the genus of the dual surface \(\Sigma^D\), indeed 
\begin{equation}
	2-2g(\Sigma^D)= |V(\Gamma)|-|E(\Gamma)|+|Z(\Gamma)|= |V(\Gamma)|-|E(\Gamma)|+B+|F(\Gamma)|-2\operatorname{Area}(N)=2-2I
\end{equation}
where we have used the Euler formula and Pick's theorem \(\operatorname{Area}(N)=I+B/2-1\). Hence topologically dual surface \(\Sigma^D\) is homeomorphic to spectral curve \(\overline{\mathcal{C}}\) since both them are of genus \(I\). Moreover, (the thickened) ribbon graph \(\Gamma^D\) is homeomorphic to \(\mathcal{C}\) since both of them have genus \(I\) and \(|B|=|Z(\Gamma)|\) punctures.

For any closed curve \(\gamma\) one can consider \(\wt(\gamma)\) as a function from \( H^1(\Gamma,\mathbb{C}^*)\) to \(\mathbb{C}^* \). In particular, we will consider the variables \(\mathbf{x},\mathbf{z},\lambda,\mu\) as such functions. 
Let us define the Poisson bracket between any two such functions as
\begin{equation}\label{eq:Poisson SigmaD}
	\{ \wt(\gamma_1),\wt(\gamma_2)\}= (\gamma_1\cdot \gamma_2)_{\Sigma^D} \; \wt(\gamma_1)\wt(\gamma_2),
\end{equation}
where \((\gamma_1\cdot \gamma_2)_{\Sigma^D}\) denotes intersection number on a surface \(\Sigma^D\).

The Poisson center for this bracket is generated by weights of zigzag paths \(z_j\), since zigzags of \(\Gamma\) are contractible on \(\Sigma^D\). From now on we will always assume that \(A,B\) paths are chosen to be zigzags or formal combinations of zigzags with rational coefficients (in particular this is necessary for Theorem \ref{th:Integrability}). Hence \(\lambda,\mu\) become Casimir functions (i.e. belong to the Poisson center), while other Casimir functions can be chosen to be trivial in \(H^1(\Sigma)\).

The definition of Poisson bracket for face variables \(\mathbf{x}\) can be restated in terms of (face) \emph{quivers}. Let us define quiver \(\mathcal{Q}\) to be a dual graph to \(\Gamma \subset \Sigma\), namely the vertices of \(\mathcal{Q}\) correspond to faces of \(\Gamma\) and edges of \(\mathcal{Q}\) correspond to edges of \(\Gamma\). The orientation of quiver \(\mathcal{Q}\) is chosen such that edges go clockwise around black vertices of \(\Gamma\) and counterclockwise around white vertices of \(\Gamma\). 

\begin{figure}[h]
		\begin{center}
	\begin{tikzpicture}[font= \small]
		\def\xs{1}
		\def\ys{0.8}

			\draw ($(-\xs,-\ys)$) circle (2.5pt);
			\draw ($(\xs,\ys)$) circle (2.5pt);
			\draw[fill] ($(\xs,-\ys)$) circle  (2.5pt);
			\draw[fill] ($(-\xs,\ys)$) circle  (2.5pt);

			\draw[] ($(-\xs,-\ys)$) -- ($(\xs,-\ys)$) -- 	($(\xs,\ys)$) -- ($(-\xs,\ys)$) -- cycle;
			
			\draw[] ($(-\xs,-\ys)$) -- ++ (-1*\xs,-0.4*\ys);
			\draw[] ($(\xs,-\ys)$) -- ++ (1*\xs,-0.4*\ys);
			\draw[] ($(-\xs,\ys)$) -- ++ (-1*\xs,0.4*\ys);
			\draw[] ($(\xs,\ys)$) -- ++ (1*\xs,0.4*\ys);

			\node[circle,draw,inner sep=2, blue] (x) at (0,0) {};
			\node[circle,draw,inner sep=2, blue] (x3) at (2*\xs,0) {};
			\node[circle,draw,inner sep=2, blue] (x1) at (-2*\xs,0) {};
			\node[circle,draw,inner sep=2, blue] (x2) at (0,1.5*\ys) {};
			\node[circle,draw,inner sep=2, blue] (x4) at (0,-1.5*\ys) {};
			
			\draw[special arrow=0.8, blue] (x) to (x3); 
			\draw[special arrow=0.8, blue] (x) to (x1); 
			\draw[special arrow=0.8, blue] (x2) to (x);
			\draw[special arrow=0.8, blue] (x4) to (x);
			
			\draw[special arrow=0.8, blue] (x1) to[bend left=40] (x2); 
			\draw[special arrow=0.8, blue] (x1) to[bend right=40] (x4); 
			\draw[special arrow=0.8, blue] (x3) to[bend left=40] (x4); 
			\draw[special arrow=0.8, blue] (x3) to[bend right=40] (x2); 
	\end{tikzpicture}	
	\end{center}	
	\caption{A piece of bipartite graph \(\Gamma\) and a corresponding piece of the quiver  \(\mathcal{Q}\) (in blue).}
\end{figure}
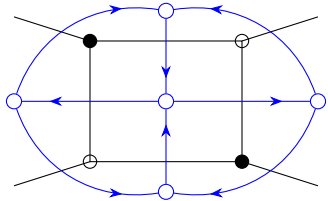
    
Let \(b\) denote adjacency matrix of quiver \(\mathcal{Q}\), namely for any two faces \(f_i\) and \(f_j\) let \(b_{ij}\) denotes number of edges from \(f_i\) to \(f_j\) minus number of edges from \(f_j\) to \(f_i\). Then it follows from the  formula \eqref{eq:Poisson SigmaD} that 
\begin{equation}\label{eq:Poisson cluster}
	\{x_{i},x_{j}\}= b_{ij}x_{i}x_{j}. 
\end{equation}
We can always remove oriented cycles of length 1 and 2 from the quiver, since  such operation preserves matrix \(b\), so we will always assume that there are no such cycles. Note that contraction of vertices of valency 2 in graph \(\Gamma\) mentioned above results to removing of cycle of length 2 in quiver~\(\mathcal{Q}\).

Theorem \ref{th:zigzags boundary} has an important corollary in terms of Poisson structure. Namely, since all zigzag's weights \(\{z_i\}\) are Casimir functions, there exists normalization of \(\mathcal{Z}(\mathbf{x}|\lambda,\mu)\) such that all \(\mathcal{Z}_{a,b}(\mathbf{x})\) for \((a,b) \in (\text{boundary of } N)\) are Casimir functions
\footnote{Note that zigzag variables \(\mathbf{z}\) cannot be expressed only through \(\mathcal{Z}_{a,b}(\mathbf{x})\) for \((a,b) \in (\text{boundary of } N)\) since latter do not depend on spectral parameters. Geometrically, such \(Z_{a,b}(\mathbf{x})\) are the Casimir functions on the cluster variety \(\mathcal{X}\) defined in the next section, while zigzag variables \(\mathbf{z}\) are not defined there.
}.
For example one choose normalization such that \(\mathcal{Z}_{a,b}=1\) for some vertex of \(N\).

On the other hand, it follows from the formula \eqref{eq:Poisson SigmaD} that the rank of the Poisson bracket is equal to \(2g(\Sigma^D)=2I\). The following theorem states that the set of \(I\) functions \(\mathcal{Z}_{a,b}(\mathbf{x})\) for \((a,b) \in (\text{interior of } N)\) defines an integrable system. We will call them the \emph{Goncharov-Kenyon integrable systems}. 

\begin{Theorem}[{\cite[Theorem 3.7]{Goncharov:2013}}]  \label{th:Integrability}
	The functions \(\mathcal{Z}_{a,b}(\mathbf{x})\) for \((a,b) \in (\text{interior of } N)\) Poisson commute and are algebraically independent.
\end{Theorem}

\begin{Remark}  \label{rem:q=1} 
	Several remarks are in order.	
	\begin{itemize}
		\item The functions \(\mathcal{Z}_{a,b}\) depend only on face variables \(\mathbf{x}\), so integrability statement concerns the Poisson bracket defined by formula \eqref{eq:Poisson cluster}. 
	
		\item The face variables satisfy constraint \(\prod_{f_i \in F(\Gamma)} x_i=1 \). 
	
		\item For any consistent dimer model the zigzag variables satisfy \(\prod_{\zeta_j \in Z(\Gamma)} z_j=1\). In order to see this note that each edge $e$ of \(\Gamma\) belongs to exactly two zigzags, and, moreover, these two zigzags go through \(e\) in opposite directions. Hence \( \prod \wt(\zeta_j)=1\) and \(\prod \sgn_K{\zeta_j}=1\). Furthermore 
		\begin{equation}
			\prod (-1)^{|\zeta_j|/2-1}=(-1)^{E(\Gamma)-Z(\Gamma)}=(-1)^{V(\Gamma)-2+2I}=1.
		\end{equation}
 		Here we used Euler formula for \(\Sigma^D\) and parity of \(V(\Gamma)\) which follows from the existence of a dimer configuration.
	\end{itemize}
\end{Remark}

\begin{Remark}  \label{rem:FM} 
Alternative approach to Goncharov-Kenyon integrable systems comes from the Poisson-Lie groups \cite{Fock:2016}. In this approach the phase space of the system is a double Bruhat cell in the quotient of coextended loop group by torus \(\widehat{PGL}(M)/H\). Such cells are parametrized by the elements of extended double affine Weyl group \(W^{\mathrm{ae}} (A_{M-1} \times A_{M-1})\). For such element one can construct Lax matrix in \(L(\lambda)\in \widehat{PGL}(M)\) as certain product of elementary matrices 
\begin{equation}
	 E_i=\exp(e_i),\quad F_i=\exp(f_i)\quad H_i(x)=\exp(x h^i),\quad 0\leq i \leq M-1,
\end{equation}
where \(e_i,f_i\) are simple root generators in Lie algebra \(\widehat{\mathfrak{sl}}_M\) and \(h^i\) are Cartan elements such that \([h^i,e_j]=\delta_{i,j}e_j\), \([h^i,f_j]=-\delta_{i,j}f_j\). Then the spectral curve equation \eqref{eq:Zcurve} arises as a Lax equation 
\begin{equation}
	\label{ZLax}
	\mathcal{Z}(\mathbf{x}|\lambda,\mu) \sim \det (L(\lambda)+\mu)=0.
\end{equation}
Note that established in \cite{Fock:2016} correspondence between the elements \(W^{\mathrm{ae}} (A_{M-1} \times A_{M-1})\) and polygons \(N\) up to \(SA(2,\mathbb{Z})\) action is not one to one. Even the number \(M\) is not uniquely determined by polygon~\(N\).

%

\end{Remark}

\subsection{Cluster mutations and face mutations}
\label{ssec:cluster}

In order to study further properties of dimer models it is convenient to recall cluster structure on them \cite{Goncharov:2013}. For the reference about \(\mathcal{X}\)-cluster varieties see for example \cite{Fock:2006cluster}.

\begin{Definition}
	Let \(n\) be a positive integer. \emph{Cluster seed} is a pair \(\mathsf{s}=(b,\mathbf{x})\), where \(b\) is \(n\times n\), skew-symmetric integral matrix and \(\mathbf{x}=(x_1,\dots,x_n)\) is a tuple of \(n\) variables corresponding to rows (or columns) of matrix \(b\).
	
	\emph{Cluster chart} is an algebraic torus \(\mathcal{X}_{\mathsf{s}}=(\mathbb{C}^*)^n\) such that \(\mathbf{x}\) are coordinate functions on it. The \emph{cluster Poisson bracket} is defined by the formula \(\{x_i,x_j\}=b_{ij}x_ix_j\).
\end{Definition}


For any consistent bipartite graph on a torus we can assign a seed given by matrix \(b\) which is adjacency matrix of quiver \(\mathcal{Q}\) and variables \(x_i\) corresponding to vertices of the quiver (equivalently faces of \(\Gamma\)). We denote the corresponding cluster chart by \(\mathcal{X}_\Gamma\).

For a consistent dimer model one can assign a point on the chart \(\mathcal{X}_\Gamma\) taking \(x_i=\wt(f_i)\). Note that the formula for cluster Poisson bracket for this seed is given by~\eqref{eq:Poisson cluster}. Such points constitute a subvariety defined by equation \(q=1\), where \(q=\prod_{f_i \in F(\Gamma)} x_i\). It is easy to see that \(q\) is a Casimir function for the cluster Poisson bracket. Indeed, in combinatorial language this means that for any vertex of \(\mathcal{Q}\) number of ingoing edges is equal to number of outgoing edges. This follows from the bipartite property of \(\Gamma\).

The integrability Theorem \ref{th:Integrability} works only on the Poisson subvariety given by \(q=1\). However, from the cluster point of view it is natural to consider arbitrary \(q\).  This leads for \(q\neq 1\) to deatonomization of cluster integrable systems e.g. \(q\)-difference Painlev\'e equations \cite{Bershtein:2018cluster}.

\begin{Definition} Mutation in a vertex \(k\) is a transformation of seeds \(\mu_k\colon \mathsf{s}=(b,\mathbf{x}) \rightarrow \tilde{\mathsf{s}}= (\tilde{b},\tilde{\mathbf{x}}) \) such that
\begin{equation}\label{eq:mutation rule}
	\tilde{b}_{ij}=
	\begin{cases} 		
		-b_{ij}, \quad \text{ if } i=k  \text{ or } j=k  
		\\ 
		b_{ij}+\frac{b_{ik}|b_{kj}|-b_{jk}|b_{ki}|}{2},\ \text{ otherwise }
	\end{cases}	
	\quad 
	\tilde{x}_i=\begin{cases} 
		x_k^{-1} \quad \text{ if } k=i 
		\\  
		x_i(1+x_k^{\operatorname{sgn}b_{ik}})^{b_{ik}} \quad \text{ if } k\neq i 
	\end{cases}.
\end{equation}
\end{Definition}

Cluster charts are connected by mutations. It is straightforward to check that mutation is a Poisson map and involution. By \(\mathcal{X}\) we denote the Poisson variety obtained by gluing of all cluster charts \(\mathcal{X}_\mathsf{s}\)  related to a given one by sequences of mutations. Such Poisson varieties are called \(\mathcal{X}\)-cluster varieties.

By \emph{cluster modular group} \(G_{\mathcal{Q}}\) we call the group of birational transformations of \(\mathcal{X}_\mathsf{s}\), generated by sequences of mutations (and permutations of vertices), which preserves the quiver \(\mathcal{Q}\). The group \(G_{\mathcal{Q}}\) depends not on seed \(\mathsf{s}\) but rather on mutational class of seed. In other words mutation \(\mathsf{s} \to \tilde{\mathsf{s}}\) gives a natural isomorphism of the corresponding cluster modular groups \(G_{\mathcal{Q}}\) and \(G_{\tilde{\mathcal{Q}}}\).

It is natural to ask for the interpretation of the mutations in terms of combinatorics of dimer models. If the face \(f_k\) is 4-gon, then the mutation in variable \(x_k\) corresponds to the move of dimer model depicted in Fig.~\ref{fi:face4}.

\begin{figure}[h]
    \centering
    \includegraphics[]{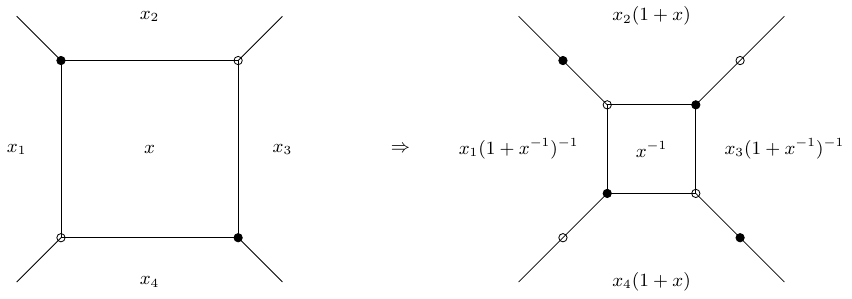}
	\caption{\label{fi:face4}
		4-gon face mutation (spider move)
		}
\end{figure}

The following proposition is standard

\begin{Proposition}\label{prop:4gon mutation}
	(a) The 4-gon mutation maps consistent dimer model to consistent dimer model
	
	(b) The 4-gon mutation  move preserves partition function \(\mathcal{Z}(\mathbf{x}|\lambda,\mu)\).
\end{Proposition}
	
\begin{Theorem}
	Any two consistent bipartite graphs on a torus with the same Newton polygon \(N\) are connected by a sequence of  4-gon face mutations and contractions/uncontructions moves.
\end{Theorem}	

This fact was stated in \cite[Th. 2.5]{Goncharov:2013}, but it looks like only a weaker statement was proven in loc. cit. See also discussion in \cite[pp. 396--397]{Bocklandt:2016dimer}. On the other hand it looks like this claim follows from more general result of \cite{Galashin:2022move}.
	
\begin{Remark}
	There is no so transparent interpretation of the mutation in the non 4-gon faces. See \cite{Casals:2022microlocal} for the approach through weaves.  We will discuss another possible interpretation below.		
\end{Remark}

There is another class of transformations introduced in \cite{Inoue:2016toric}. Assume that two parallel zigzags \(\zeta_1,\zeta_2\) are such that the graph between them has a form of hexagonal grid, see Fig.~\ref{fi:zigzag permutation}.
\begin{figure}[h]
	\centering
    \includegraphics[]{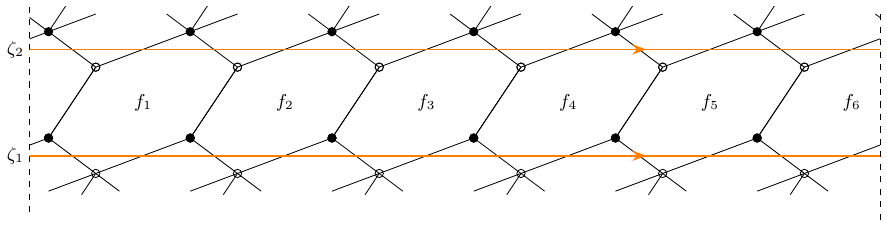}
	\caption{\label{fi:zigzag permutation} Two zigzags \(\zeta_1,\zeta_2\) in bipartite graph in which zigzag transposition can be applied.}	
\end{figure}

The faces \(f_1,\dots f_{l}\) (see Fig.~\ref{fi:zigzag permutation} for \(l=6\)) have natural cyclic order, hence we label them by \(j\in \mathbb{Z}/l \mathbb{Z}\), and denote mutation of seed at the vertex corresponding to \(f_j\) by \(\mu_j\).
\begin{Proposition}[{\cite[Th. 3.1]{Inoue:2019cluster}, \cite[Th. 7.7]{Goncharov:2018donaldson} \cite[Th. 3.6]{Masuda:2023birational}}]
	For \(j\in \mathbb{Z}/l \mathbb{Z}\) define
	\begin{equation}\label{eq:permutation of zigzags}
		R= \big(\mu_j \circ \mu_{j-1}\circ \dots \circ \mu_{j+3} \circ \mu_{j+2} \big)\circ (j+1,j+2)\circ \big( \mu_{j+2} \circ \mu_{j+3}\circ \dots \circ  \mu_{j-1} \circ\mu_j \big) 
	\end{equation}
	\begin{enumerate}[label=(\alph*)]
		\item Transformation \(R\) does not depend on the choice of \(j\).
		
		\item Transformation \(R\) preserves the quiver \(\mathcal{Q}\). 
	\end{enumerate}
\end{Proposition}
\begin{Proposition}[{\cite{Inoue:2016toric},\cite{George:2022discrete}}]	 \label{prop:zigzag transposition}
	The transformation \(R\) can be extended to the transformtion of the dimer model such that 
	\begin{enumerate}[label=(\alph*)]	
		\item The bipartite graph \(\Gamma\) is preserved.
		
		\item The weights of face variables are transformed as cluster variables for transformation \eqref{eq:permutation of zigzags}.
		
		\item The weights of two zigzags \(\zeta_1\) and \(\zeta_2\) are swapped, while all other zigzag's weights are preserved.
		
		\item The partition function \(\mathcal{Z}(\lambda,\mu)\) is preserved.
	\end{enumerate}
\end{Proposition}
Such transformation was called \emph{geometric $R$ matrix} in \cite{Inoue:2016toric}, we also call it \emph{transposition of zigzags} \(\zeta_1,\zeta_2\)

The group generated by 4-gon face mutations was studied in \cite{George:2019cluster}, while the group generated by 4-gon face mutations and permutations of zigzags was studied in \cite{George:2022discrete}. It can be proven that for any two consecutive parallel zigzags \(\zeta_1,\zeta_2\) there is actually a sequence of 4-gon mutations, after which mutated zigzags \(\zeta_1,\zeta_2\) form a picture from Fig.~\ref{fi:zigzag permutation} and can be permuted (see \cite{George:2022discrete}, and also Lemma~\ref{lem:zigzag patch} below).

\section{Zigzag mutations and reductions} \label{sec:zigzag reductions}

This section is organized as follows. We start with the simplest length-4 zigzags mutation of bipartite graph and promote them to the transformation of the dimer model. This transformation cause mutation of the dimer partition function and corresponding Newton polygon. We study generalization of this results to zigzags of greater length. The corresponding transformation of weight function is defined on the subvariety of \(\mathcal{X}\). These subvarieties are not Poisson, but the transformation above naturally acts between their ``quotients'', that are Hamiltonian reductions. At the end of this section we discuss general (almost entirely conjectural) definitions and properties of reduced Goncharov-Kenyon systems labeled by decorated Newton polygons, and their zigzag mutations.


\subsection{Length 4 zigzags}

Let \(\zeta\) be a zigzag path on \(\Gamma\) of length \(4\). Note that on the dual surface \(\zeta\) is a boundary of 4-gon face.  Therefore, it is natural to consider dual to the face mutation, defined on Fig.~\ref{fi:face4}.  Using uncontractions one can make all vertices in \(\zeta\) to be 3-valent. Denote the corresponding zigzag variable by \(z=-\sgn_K(\zeta) \wt(\zeta) \).
\begin{Definition} \label{def:zigzag 4 mutatation}
	The zigzag mutation  of the bipartite graph along \(\zeta\) is defined as on Fig.~\ref{fi:zigzag4}
 	\begin{figure}[h]
		\begin{center}
			\includegraphics[scale=0.8]{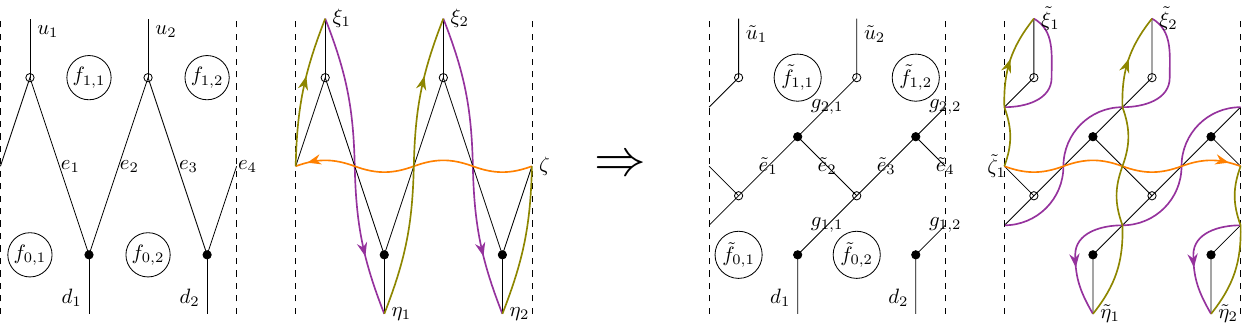}
		\end{center}		
		\caption{\label{fi:zigzag4}
			Length-4 zigzag mutation: bipartite graph, zigzags, edge and face labels before mutation (on the left) and after mutation (on the right), the face labels are encircled. 
			}
	\end{figure}

	The zigzag mutation of the dimer model is supplied with the following transformation of the edge weights \(\wt \mapsto{\widetilde{\wt}}\) and the Kasteleyn signs \(\sgn_K \mapsto \sgn_{\tilde{K}}\):
	\begin{subequations}\label{eq:weights zigzag4}
		\begin{align}
			&\widetilde{\wt}(\tilde{e}_i)=\wt(e_i)\wt(e_2)^{-1}\wt(e_4)^{-1},\; i=1,3; \qquad \widetilde{\wt}(\tilde{e}_i)=\wt(e_{i})^{-1},\; i=2,4;
			\\
			&\sgn_{\tilde{K}}(\tilde{e}_i)=\sgn_K(e_i)\sgn(e_2)\sgn(e_4),\; i=1,3; \qquad \sgn_{\tilde{K}}(\tilde{e}_i)=-\sgn_K(e_{i}),\;  i=2,4;
			\\
			&\widetilde{\wt}(g_{i,j})=1,\; 	\sgn_{\tilde{K}}(g_{i,j})=1,\; i=1,2,\; j=1,2;\\
			& \widetilde{\wt}(\tilde{u}_i)=\wt(u_i)(1+z)^{-1}, \; i=1,2.			 				
		\end{align}
	\end{subequations}
	All other weights and Kasteleyn signs remain intact.
\end{Definition}

\begin{Remark}Several remarks are in order:
		\begin{itemize}
%
		\item Such transformations of bipartite graphs as well as more general zigzag mutations were defined in \cite{Higashitani:2022}. 
		
		\item The definition of transformed weights and Kasteleyn signs is motivated by the transformation property of partition function given below. Note that the edge weights and Kasteleyn signs are defined up to gauge transformations and formula~\eqref{eq:weights zigzag4} corresponds to a particular choice.  
	\end{itemize}
\end{Remark}

Assume for simplicity that zigzag \(\zeta\) is horizontal and the spectral parameter \(\lambda\) is defined as \(\lambda^{-1}=z\).  Let us also assume, that the weights of the edges \(u_1,u_2\) (see Fig.~\ref{fi:zigzag4}) are proportional to another spectral parameter~\(\mu^{-1}\), while all other edge weights are \(\mu\)-independent. Then for any dimer configuration \(D\) its weight \(\wt(D)\) is proportional to \(\mu^{|D\cap \zeta|-2}\). Since \(|\zeta|=4\) the intersection \(D\cap \zeta\) can consist of \(2\), \(1\), or \(0\) edges and in appropriate normalization the contribution of \(D\) into partition function~\eqref{eq:Z = sum D} gives terms, proportional to \(\mu^{-1}\), \(1\) or \(\mu\) correspondingly.

In this normalization the partition function has form
\begin{equation}
	\label{eq:Z zigzag length 4}
	\mathcal{Z}(\mathbf{x}|\lambda,\mu)=\mu P_1(\lambda)+ P_0(\lambda)+\mu^{-1}P_{-1}(\lambda).
\end{equation}
Due to Theorem~\ref{th:zigzags boundary}, \((1+\lambda^{-1})\) divides \(P_1(\lambda)\).

\begin{Proposition}  \label{prop:zigzag4}
	
	(a) Length-4 zigzag mutation transforms partition function \eqref{eq:Z zigzag length 4} into
	\begin{equation}\label{eq:Z zigzag4}
		\tilde{\mathcal{Z}}(\mathbf{x}|\lambda,\nu)=\nu \tilde{P}_1(\lambda)+ P_0(\lambda)+\nu^{-1}\tilde{P}_{-1}(\lambda)
	\end{equation}
	with \(\nu=\mu(1+\lambda^{-1})\) and \(\tilde{P}_1(\lambda)=P_1(\lambda)(1+\lambda^{-1})^{-1}\), \(\tilde{P}_{-1}(\lambda)=P_{-1}(\lambda)(1+\lambda^{-1})\). In other words \(\tilde{\mathcal{Z}}(\mathbf{x}|\lambda,\nu)=\mathcal{Z}(\mathbf{x}|\lambda,\mu)\).
	
	(b) Length-4 zigzag mutation maps consistent dimer model to a consistent dimer model.	
\end{Proposition}

The formula \eqref{eq:Z zigzag4} is a particular case of a polynomial mutation, see Definition \ref{def:mut polyn} below. It is convenient to introduce notation for the product of the weight and Kasteleyn sign 
\begin{equation}
	\wt_K(e)=\sgn_K(e)\wt(e),\qquad \widetilde{\wt}_{\tilde{K}}(e)=\sgn_{\tilde{K}}(e)\widetilde{\wt}(e).
\end{equation}

\begin{proof}
	(a) The partition function~\eqref{eq:Z = sum D} can be decomposed according to the intersection \(D\cap \zeta\). 
	Hence we get
	\begin{equation}
		\mathcal{Z}(\lambda,\mu)=\mu \big(\wt_K(e_1)\wt_K(e_3){-}\wt_K(e_2)\wt_K(e_4)\big) Z_5(\lambda)
		+ \sum\nolimits_{i=1}^4 \wt_K(e_i)  Z_i(\lambda) 		
		+ \mu^{-1} Z_0(\lambda).
	\end{equation}
	Here \(Z_0(\lambda)\) corresponds to dimer configurations with \(D\cap \zeta=\varnothing\), and the contributions \(Z_i(\lambda)\), \(1\leq i \leq 4\) correspond to the dimer configurations, such that \(D\cap \zeta=\{e_i\}\). The term \(Z_5(\lambda)\), corresponding to dimer configurations with intersection \(|D\cap \zeta|=2\), is proportional to
	\begin{equation}
		\wt_K(e_1)\wt_K(e_3)-\wt_K(e_2)\wt_K(e_4)
		 =-(1+\lambda^{-1})\wt_K(e_2)\wt_K(e_4).
	\end{equation}
		Denote the spectral parameter for dimer graph after mutation by \(\nu=\mu(1+\lambda^{-1})\). Due to our conventions, the weights \(\widetilde{\wt}(\tilde{u}_1), \widetilde{\wt}(\tilde{u}_2)\) are proportional to \(\nu^{-1}\), while other weights are \(\nu\)-independent. Hence the partition function for the zigzag mutated dimer model acquires the form
	\begin{multline}
		\widetilde{\mathcal{Z}}(\lambda,\nu)= \nu \Big(\prod\nolimits_{i,j=1}^2 \widetilde{\wt}_{\tilde{K}}(g_{i,j}) \Big) Z_5(\lambda)+ \Big( -\widetilde{\wt}_{\tilde{K}}(g_{2,2})\widetilde{\wt}_{\tilde{K}}(g_{1,1})\widetilde{\wt}_{\tilde{K}}(\tilde{e}_1)Z_1(\lambda)
		\\
		+\widetilde{\wt}_{\tilde{K}}(g_{2,1})\widetilde{\wt}_{\tilde{K}}(g_{1,1})\widetilde{\wt}_{\tilde{K}}(\tilde{e}_4)Z_2(\lambda)	-\widetilde{\wt}_{\tilde{K}}(g_{2,1})\widetilde{\wt}_{\tilde{K}}(g_{1,2})\widetilde{\wt}_{\tilde{K}}(\tilde{e}_3)Z_3(\lambda)
		\\ -\widetilde{\wt}_{\tilde{K}}(g_{2,2})\widetilde{\wt}_{\tilde{K}}(g_{1,2})\widetilde{\wt}_{\tilde{K}}(\tilde{e}_2)Z_4(\lambda)	\Big) 
		+ \nu^{-1} \Big(\widetilde{\wt}_{\tilde{K}}(\tilde{e}_1)\widetilde{\wt}_{\tilde{K}}(\tilde{e}_3)-\widetilde{\wt}_{\tilde{K}}(\tilde{e}_2)\widetilde{\wt}_{\tilde{K}}(\tilde{e}_4)\Big)Z_0(\lambda).		
	\end{multline}
	Using formulas \eqref{eq:weights zigzag4} one concludes, that 
	\begin{equation}
		\mathcal{Z}(\lambda,\mu)=-\wt_K(e_2)\wt_K(e_4) \widetilde{\mathcal{Z}}(\lambda,\nu).
	\end{equation}
	In other words the partition function is preserved up to monomial factor and transformation of spectral variables \(\mu\mapsto \nu=\mu(1+\lambda^{-1})\).
		
	(b) Straightforward from the picture of zigzags on the Fig.~\ref{fi:zigzag4}.
\end{proof}

For two vectors \(v_1=(x_1,y_1)\) and \(v_2=(x_2,y_2)\), let \(\det(v_1,v_2)=x_1y_2-x_2y_1\) be the oriented area of parallelogram with sides \(v_1,v_2\).

\begin{Definition} \label{def:zigzag quiver}
	\emph{Zigzag} (or \emph{dual}) quiver \(\mathcal{Q}^D\) is a quiver with vertices corresponding to zigzags and number of arrows from vertex \(\zeta_i\) to vertex \(\zeta_j\) equal to \(b_{ij}^D=\det([\zeta_i],[\zeta_j])\).	
\end{Definition}
In terms of Newton polygons the vertices of \(\mathcal{Q}^D\) correspond to the primitive segments on the boundary of \(N\). In other words, for any side \(E\) of \(N\) there are \(|E|\) vertices. For vertices \(\zeta_i,\zeta_j\) corresponding to sides \(E,E'\) correspondingly, we have 
\(b_{ij}^D=\det(E,E')/(|E||E'|)\).

It follows from the construction of the dual surface $\Sigma^D$ that the zigzag quiver \(\mathcal{Q}^D\) is a face quiver for dual graph \(\Gamma^D\subset \Sigma^D\).

\begin{Proposition} \label{prop:zigzag4 quivers}
	(a) Length-4 zigzag mutation preserves (face) quiver \(\mathcal{Q}\) and face variables.
	
	(b) Length 4 zigzag mutation is a mutation of the (zigzag) quiver \(\mathcal{Q}^D\) and zigzag variables.	 
\end{Proposition}
\begin{proof}
	(a) The only faces which require attention are \(f_1,f_2,f_3,f_4\) in Fig.\ref{fi:zigzag4}. It is straightforward to check that the numbers of arrows between them are preserved, and their face weights \(x_1,x_2,x_3,x_4\) are preserved under the transformation~\eqref{eq:weights zigzag4}.
	
	(b) The only zigzags which require attention are \(\zeta, \eta_1,\eta_2,\xi_1,\xi_2\) in Fig.\ref{fi:zigzag4}. It is straightforward to check that homology classes of this zigzags are transformed as 
	\begin{equation}
		[\tilde{\zeta}]=-[\zeta],\quad  [\tilde{\eta}_i]=[\eta_i],\;\; [\tilde{\xi}_i]=[\xi_i]+[\zeta],\;\; i=1,2.
	\end{equation}
	Using this formulas one can see that transformation of adjacency matrix \(b^D\) agrees with the mutation rule \eqref{eq:mutation rule}. Let us denote zigzag variables by 
	\begin{equation}
		z=(-1)^{|\zeta|/2-1}\wt_K(\zeta),\quad k_i=(-1)^{|\xi_i|/2-1}\wt_K(\xi_i),\;\;  h_i=(-1)^{|\eta_i|/2-1}\wt_K(\eta_i)\;\; i=1,2.
	\end{equation}
	and similarly for \(\tilde{z}, \tilde{k}_i, \tilde{h}_i\). Then, the transformation of zigzag variables reads	
	\begin{equation}
		\tilde{z}=z^{-1};\quad
		\tilde{k}_i=k_i(1+z^{-1})^{-1},\;\;  
		\tilde{h}_i=h_i(1+z),\;\; i=1,2,		
	\end{equation}
	which is also in agreement with the mutation rule~\eqref{eq:mutation rule}.
\end{proof}

\begin{Example} 
	In the Figs.~\ref{fi:ex square} and \ref{fi:ex triangle} we presented two consistent dimer models with the corresponding Newton polygons, face quivers and zigzag quivers. It is straightforward to check that corresponding bipartite graphs are related by a zigzag mutation along \(\zeta_1\).\footnote{To be more precise, first  uncontractions, second zigzag mutation, third contractions, and fourth \(SL(2,\mathbb{Z})\) transformation.} In agreement with Proposition~\ref{prop:zigzag4 quivers} the face quivers in Figs. \ref{fi:ex square}, \ref{fi:ex triangle} coincide while the zigzag quivers are related by mutation in vertex \(z_1\).
	\begin{figure}[h]
		\begin{center}
			\begin{tikzpicture}
			\node at (0,0) {\includegraphics{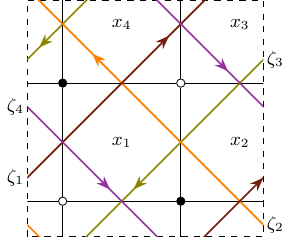}};
	

			\begin{scope}[shift={(4.5,0)}, scale=1, font = \small]
				
				\def\xs{1}
				\def\ys{1}
				
			\node[inner sep=0pt, outer sep=0pt, fill=black, circle, minimum size=3pt] at (0,-\ys) (A) {};
			\node[inner sep=0pt, outer sep=0pt, fill=black, circle, minimum size=3pt] at (\xs,0) (B) {};
			\node[inner sep=0pt, outer sep=0pt, fill=black, circle, minimum size=3pt] at (0,\ys) (C) {};
			\node[inner sep=0pt, outer sep=0pt, fill=black, circle, minimum size=3pt] at (-\xs,0) (D) {};
			
			\draw (0,0) circle (2pt);

				\draw[Brown, thick] (A) -- (B);
				\draw[orange, thick] (B) -- (C);
				\draw[olive, thick] (C) -- (D);
				\draw[Purple, thick] (D) -- (A);			
			\end{scope}
	
			
			\begin{scope}[shift={(7.5,0)},scale=1.5, font = \small]
				\def\xs{1}
				\def\ys{1}

				\node[styleNode] (x1) at(-0.5*\xs,-0.5*\ys){$x_1$};
				\node[styleNode] (x2) at(0.5*\xs,-0.5*\ys){$x_2$};
				\node[styleNode] (x3) at(0.5*\xs,0.5*\ys){$x_3$};
				\node[styleNode] (x4) at(-0.5*\xs,0.5*\ys){$x_4$};
				
				\draw[styleArrow](x1) to[bend left=15] (x2);
				\draw[styleArrow](x1) to[bend right=15] (x2);
				\draw[styleArrow](x2) to[bend left=15] (x3);
				\draw[styleArrow](x2) to[bend right=15] (x3);
				\draw[styleArrow](x3) to[bend left=15] (x4);
				\draw[styleArrow](x3) to[bend right=15] (x4);	\draw[styleArrow](x4) to[bend left=15] (x1);
				\draw[styleArrow](x4) to[bend right=15] (x1);
				
			\end{scope}
	
			
			\begin{scope}[shift={(10.5,0)},scale=1.5, font = \small]
				\def\xs{1}
				\def\ys{1}

				\node[styleNode] (z1) at(-0.5*\xs,-0.5*\ys){$z_1$};
				\node[styleNode] (z2) at(0.5*\xs,-0.5*\ys){$z_2$};
				\node[styleNode] (z3) at(0.5*\xs,0.5*\ys){$z_3$};
				\node[styleNode] (z4) at(-0.5*\xs,0.5*\ys){$z_4$};	
				
				\draw[styleArrow](z1) to[bend left=15] (z2);
				\draw[styleArrow](z1) to[bend right=15]  (z2);		
				\draw[styleArrow](z2) to[bend left=15] (z3);
				\draw[styleArrow](z2) to[bend right=15] (z3);
				\draw[styleArrow](z3) to[bend left=15] (z4);
				\draw[styleArrow](z3) to[bend right=15]  (z4);
				\draw[styleArrow](z4) to[bend left=15] (z1);
				\draw[styleArrow](z4) to[bend right=15] (z1);
				
			\end{scope}
			
			\end{tikzpicture}
		\end{center}
		\caption{\label{fi:ex square} Square: bipartite graph, Newton polygon, face and zigzag quivers}		
	\end{figure}
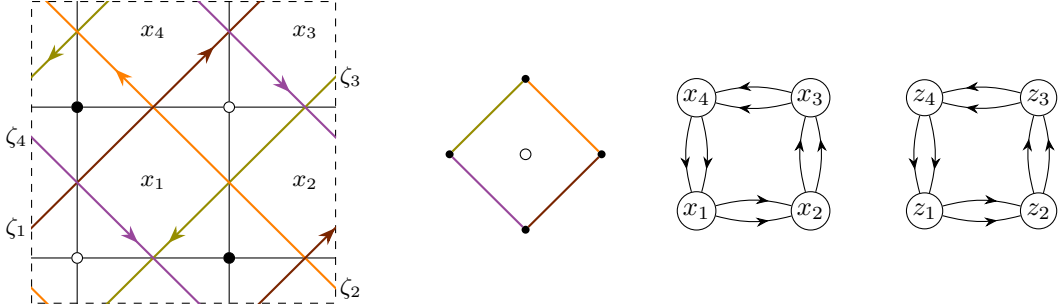

	\begin{figure}[h]
	\begin{center}	
		\begin{tikzpicture}
		
			\node at (0,0) {\includegraphics{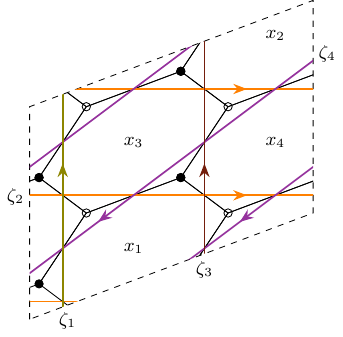}};


			\begin{scope}[shift={(4.5,0)}, scale=1, font = \small]
				
				\def\xs{1}
				\def\ys{1}
				
				\node[inner sep=0pt, outer sep=0pt, fill=black,  circle, minimum size=3pt] at (\xs,-1*\ys) (A) {};
				\node[inner sep=0pt, outer sep=0pt, fill=black,  circle, minimum size=3pt] at (\xs,1*\ys) (B) {};
				\node[inner sep=0pt, outer sep=0pt, fill=black,  circle, minimum size=3pt] at (-\xs,0*\ys) (C) {};
				
				\node[inner sep=0pt, outer sep=0pt, fill=black,  circle, minimum size=3pt] at (\xs,0) (D) {};

				\draw (0,0) circle (2pt);

				\draw[olive, thick] (A) -- (D);
				\draw[Brown, thick] (D) -- (B);
				\draw[Purple, thick] (B) -- (C);
				\draw[orange, thick] (C) -- (A);
				
			\end{scope}

			
%
%
%
%
%
%
%
%
	
			
			\begin{scope}[shift={(7.5,0)},scale=1.5, font = \small]
				\def\xs{1}
				\def\ys{1}

				\node[styleNode] (x1) at(-0.5*\xs,-0.5*\ys){$x_1$};
				\node[styleNode] (x2) at(0.5*\xs,-0.5*\ys){$x_2$};
				\node[styleNode] (x3) at(0.5*\xs,0.5*\ys){$x_3$};
				\node[styleNode] (x4) at(-0.5*\xs,0.5*\ys){$x_4$};			
				
				\draw[styleArrow](x1) to[bend left=15] (x2);
				\draw[styleArrow](x1) to[bend right=15] (x2); 			
				\draw[styleArrow](x2) to[bend left=15] (x3);
				\draw[styleArrow](x2) to[bend right=15] (x3);
				\draw[styleArrow](x3) to[bend left=15] (x4);
				\draw[styleArrow](x3) to[bend right=15] (x4); 			
				\draw[styleArrow](x4) to[bend left=15] (x1);
				\draw[styleArrow](x4) to[bend right=15] (x1);
				
			\end{scope}

			
			\begin{scope}[shift={(10.5,0)},scale=1.5, font = \small]
				\def\xs{1}
				\def\ys{1}

				\node[styleNode, ellipse, minimum height=0.5cm] (z1) at(0,-0.5*\ys){$z_1, z_3$};
				\node[styleNode] (z4) at(-0.5*\xs,0.5*\ys){$z_4$};
				\node[styleNode] (z2) at(0.5*\xs,0.5*\ys){$z_2$};
				
				\draw[styleArrow](z1) to[bend left=15] (z4);
				\draw[styleArrow](z1) to[bend right=10] (z4);	
				\draw[styleArrow](z4) to[bend left=15] (z2);
				\draw[styleArrow](z4) to[bend left=35] (z2);
				\draw[styleArrow](z4) to[bend right=15] (z2);
				\draw[styleArrow](z4) to[bend right=35] (z2);
				\draw[styleArrow](z2) to[bend left=15] (z1);
				\draw[styleArrow](z2) to[bend right=10] (z1);
				
			\end{scope}
		
		\end{tikzpicture}
		\end{center}	
		\caption{\label{fi:ex triangle} Triangle: bipartite graph, Newton polygon, face and zigzag quivers}		
	\end{figure}
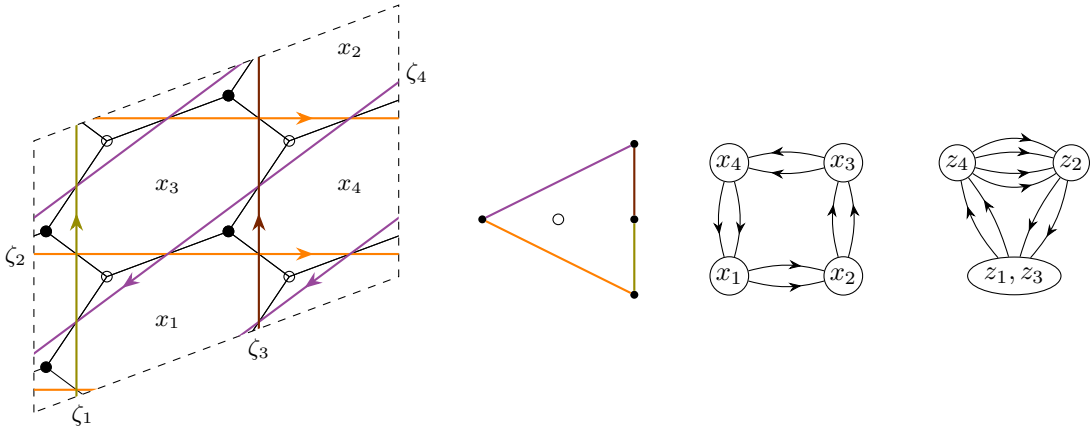
	Moreover, one can compute Hamiltonians of the Goncharov-Kenyon integrable systems in both cases. In both cases there is only one  integral points inside Newton polygon. Hence, by Theorem~\ref{th:Integrability} the classical integrable system consists of a single Hamiltonian. It is straightforward to compute this Hamiltonian and find, that in both cases it can be normalized as 
	\begin{equation}
		\label{TodaHam}
		H=x_1^{-1/2}x_2^{-1/2}(1+x_1+x_1x_2 +x_1x_2x_3)
		.
	\end{equation}
	The equality of Hamiltonians for both models follows also from the transformation of partition function proven in Proposition~\ref{prop:zigzag4}. This is the Hamiltonian of closed or affine relativistic Toda with two particles \cite{Ruijsenaars:1990Toda}, \cite{Marshakov:2013} being the simplest example of Goncharov-Kenyon integrable system.
\end{Example}
This example teaches us that the correspondence between GK integrable systems and ($SA(2,\mathbb{Z})$ orbits in the set of) Newton polygons is not one to one: different Newton polygons can correspond to the same integrable systems. We will argue below that different Newton polygons, corresponding to equivalent integrable systems,  are related by polygon mutations.

For instance, consider the case \(I=1\) when the phase space of the integrable system has Poisson rank 2. Combinatorially it means that the Newton polygon \(N\) is reflexive. This case was studied in detail (see \cite{Bershtein:2018cluster} and references therein), and there are 16  ($SA(2,\mathbb{Z})$ non-equivalent) reflexive Newton polygons, but only 8 (mutation non-equivalent) quivers and integrable systems. It is easy to see, that in this case the Newton polygons (and dimer models) corresponding to the same face quiver are related by a composition of mutations along zigzags of length 4.

\subsection{Polynomial and polygon mutations}

There is a standard action of the group \(SA(2,\mathbb{Z})=SL(2,\mathbb{Z})\ltimes \mathbb{Z}^2\) on the Laurent polynomials~\( P(\lambda,\mu)\). For any \(\begin{pmatrix}
	a & b \\ c &d
\end{pmatrix}\in SL(2,\mathbb{Z})\) one can transform the variables as \(\lambda \mapsto \lambda^a\mu^b, \mu \mapsto \lambda^c\mu^d\), and for any \((n,m)\in \mathbb{Z}^2\) we can multiply polynomial by \(\lambda^n\mu^m\). This is consistent with natural \(SA(2,\mathbb{Z})\) action on the Newton polygons.

We also need another type of transformation of Laurent polynomials, which is called mutation.
Let \(P(\lambda,\mu)\) have the form \( 	P(\lambda,\mu)=\sum\nolimits_{k=-h'}^{h} \mu^k P_k(\lambda)\)	and there exists \(C\) such that
\begin{equation} \label{eq:polyn mut cond}
	(1+C\lambda^{-1})^k \text{ divides } P_k(\lambda),\; \text{ for all } k>0. 
\end{equation}
Then the mutation of the polynomial $P$ is defined by  
\begin{equation}\label{eq:PolMut}
	\tilde{P}(\lambda,\nu)=P(\lambda,\mu), \; \text{ where } \mu=\frac{\nu}{1+C\lambda^{-1}}.
\end{equation}
Note that conditions \eqref{eq:polyn mut cond} ensure that \(\tilde{P}\) is a Laurent polynomial. The transformation \eqref{eq:Z zigzag4} is an example of mutation for \(h=h'=1\).

It follows from the condition \eqref{eq:polyn mut cond} that \( P_h(\lambda)\) is not constant. Hence, Newton polygon of \(P\) has side parallel to \((-1,0)\) at the distance \(y=h\) from horizontal axis.

\begin{Definition}[\cite{Galkin:2010mutations,Akhtar:2012}]  \label{def:mut polyn} 
	Let \(P(\lambda,\mu)\) be a Laurent polynomal with Newton polygon \(N\).  Let \(E\) be a side of \(N\) and \(h\in \mathbb{Z}_{>0}\). Let \(g \in SA(2,\mathbb{Z})\) be a transformation which makes \(E\) to be a side parallel to \((-1,0)\) at \(y=h\). Assume that \(g(P)\) satisfies conditions~\eqref{eq:polyn mut cond}. Then the \emph{mutation of Laurent polynomial} \(\mu_{E,h}(P)\) is defined as composition of \(g\), mutation~\eqref{eq:PolMut} and \(g^{-1}\). 
	
	\emph{Mutation of the polygon} is a a corresponding transformation of the Newton polygon \(N\). 
\end{Definition}

Note that we are a bit sloppy in a notation \(\mu_{E,h}\) since the mutation of the polynomial depends not only on \(E\) and \(h\) but also on the choice of root \(-C\) of \(P_h(\lambda)\). 
In order to give more transparent combinatorial meaning of the mutation of polygon we will need a few standard notions.

\begin{Definition}
	Let \(p \in \mathbb{Z}^2\) be an integral point and \(l\) be an integral line. Let \(l_0\) be an integral line going through \(p\) and parallel to \(l\). Let \(n\) be a number of integral lines parallel to \(l\) between \(l\) and \(l_0\). The \emph{integral distance} from \(p\) to \(l\) is \(n+1\).
\end{Definition}

This definition is clearly \(SA(2,\mathbb{Z})\) invariant. Equivalently, by  \(SA(2,\mathbb{Z})\)-transform one can put \(l\) to be horizontal axis \(y=0\). If this transformation sends \(p\) to \(p'=(a,b)\) then the integral distance from \(p\) to \(l\) is \(|b|\). Note that it is natural to include signs and make this distance oriented, but we do not need this here.

\begin{Definition}
	\emph{Integral height} \(h_{E,N}\) of the polygon \(N\) with respect to the side \(E\) is a maximum of integral distances from \(p\) to \(E\) for \(p\in N\). 
\end{Definition}

Equivalently one can make \(SL(2,\mathbb{Z})\) transformation which makes side \(E\) horizontal. If the image of \(N\) is contained in the (minimal) strip \(a\leq y \leq b\), then \(h_{E,N}=b-a\). 

\medskip

A convex polygon is determined up to translation by vectors of its sides \(E_1, \dots, E_{|\text{sides}|}\). Similarly, an integral convex polygon is determined up to translation by the vectors of primitive segments on sides \(e_1,\dots, e_{|\text{segments}|}\) (each side \(E\) contributes \(|E|\) segments parallel to \(E\)). The mutation of polygon \(N\) depends on the choice of the side \(E\) and number \({ h\le h_{E,N}}\). Let \(h'=h_{E,N}-h\)~\footnote{Such notations are in agreement with the formula \( 	P(\lambda,\mu)=\sum\nolimits_{k=-h'}^{h} \mu^k P_k(\lambda)\) above.}. 

\begin{Proposition}[{e.g. \cite[Corr. 3]{Kasprzyk:2017minimality}}] \label{prop:mutatation edges}
	Under the notations above, the mutated polygon \(\widetilde{N}\) is determined by segments \(\tilde{e}_1,\dots, \tilde{e}_{|\widetilde{\text{segments}}|}\) such that
	\begin{itemize}
		\item Number of segments in \(\tilde{e}_i\) \emph{parallel} to \(E\) is equal to  number of such segments \(e_i\) minus \(h\).
		
		\item Number of segments in \(\tilde{e}_i\) \emph{antiparallel} to \(E\) is equal to  number of such segments \(e_i\) plus \(h'\).
		
		\item There is a one to one correspondence between segments \(\tilde{e}_i\) neither \emph{parallel} nor \emph{antiparallel} to \(E\) with analogous segments \(e_i\), 
		given by the formula 
		\begin{equation}\label{eq:mut segments}
			\tilde{e}_i=
			\begin{cases}
				e_i+\frac{\det(E,e_i)}{|E|^{2}}E,\; \text{ if } \det(E,e_i)>0
				\\
				e_i,\; \text{ if } \det(E,e_i)<0
			\end{cases}
		\end{equation}
	\end{itemize}
\end{Proposition}

We give a couple of examples of polygon mutations on Fig.~\ref{fi:mutation polygon}. 
\begin{figure}[h]
	\begin{center}
		\begin{tikzpicture}[scale=1, font = \small]

			\begin{scope}

				\begin{scope}
					\draw[fill] (0,-1) circle (1pt) -- (1,-1) circle (1pt) -- (1,0) circle (1pt) -- (0,1) circle (1pt)  -- (-1,1) circle (1pt) -- (-1,0) circle (1pt) -- (0,-1) circle (1pt) ;
					\draw (0,0) circle (2pt);	
				\end{scope}
				
				\node at (1.5,0) {$\Rightarrow$};
				
				\begin{scope}[shift={(3,0)}]
					\draw[fill] (0,-1) circle (1pt) -- (1,-1) circle (1pt) -- (1,0) circle (1pt) -- (0,1) circle (1pt)  -- (-1,0) circle (1pt) -- (-1,-1) circle (1pt) -- (0,-1) circle (1pt) ;
					\draw (0,0) circle (2pt);	
				\end{scope}
				
			\end{scope}
			
			\begin{scope}[shift={(7.5,0.5)}]
				\begin{scope}
					\draw[fill] (0,1) circle (1pt) -- (-1,1) circle (1pt) -- (1,-2) circle (1pt) -- (0,1)  ;
					\draw (0,0) circle (2pt);	
				\end{scope}
				
				\node at (1.5,-0.5) {$\Rightarrow$};
				
				\begin{scope}[shift={(3,0)}]
					\draw[fill] (0,1) circle (1pt) -- (-1,-2)  circle (1pt) -- (0,-2) circle (1pt) -- (1,-2) circle (1pt)  -- (0,1);
					\draw (0,0) circle (2pt);	
					\draw (0,-1) circle (2pt);	
				\end{scope}
				
			\end{scope}

		\end{tikzpicture}	
	\end{center}
	\caption{\label{fi:mutation polygon} Examples of mutations of polygons: in both cases the mutation is performed along side parallel to vector \((-1,0)\) and \(h=1\).} 
\end{figure}
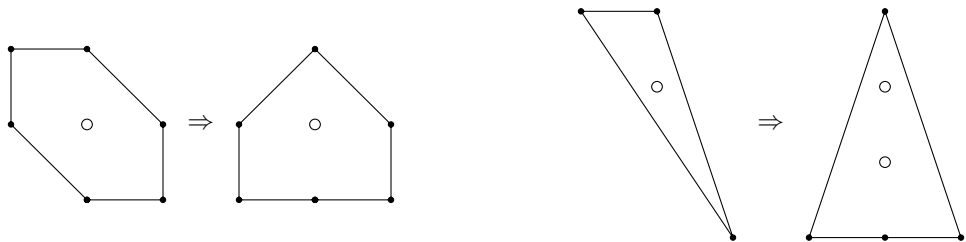

Note that even the fact that \(\sum_i \tilde{e}_i=0\) is not obvious from the description in Proposition \ref{prop:mutatation edges}. It follows from the following formulas for the height
\begin{equation}\label{eq:h in terms of edges}
	h_{E,N}=  \sum_{E' \in \text{sides of N},\; \det(E',E)>0} \frac{|\det(E',E)|}{|E|}   = \frac12 \sum_{E' \in \text{sides of N}} \frac{|\det(E',E)|}{|E|}. 
\end{equation}

\begin{Remark}	
	One can also define mutation by using variable \(\nu=\mu(\lambda+C)\) in the formula \eqref{eq:PolMut}. This is equivalent to the transformation in \eqref{eq:PolMut} up to \(SL(2,\mathbb{Z})\) action.
		In terms of the primitive segments this would correspond to the modification of \eqref{eq:mut segments} to 	$\tilde{e}_i=e_i$,\; if $\det(E,e_i)>0$ and $\tilde{e}_i=e_i-\frac{\det(E,e_i)}{|E|^{2}}E$, if  $\det(E,e_i)<0$. 
	
	In the cluster framework this correspond to distinction between signs in the definition of mutation \(\mu^+\) vs. \(\mu^-\), see e.g. \cite{Mizuno} and references therein.
	Also, jumping ahead, in the setting of zigzag mutations this correspond to zig mutations and zag mutations in \cite{Higashitani:2022}.
\end{Remark}

\subsection{Zigzag mutations 1: bipartite graphs} \label{ssec:zigzag mut 1}

\begin{Proposition} \label{prop:zigzag length}
	Let \(\zeta\) be a zigzag parallel to the side \(E\). 
	Then \(|\zeta|\geq 2h_{E,N}\).
\end{Proposition}
\begin{proof} 
	For any other zigzag \(\zeta'\) its intersection number with \(\zeta\) is equal to \(\det([\zeta'],[\zeta])\).  Hence, the number of common edges of \(\zeta\) and \(\zeta'\) is greater or equal to \(|\det([\zeta'],[\zeta])|\). Since any edge of \(\zeta\) should belong to exactly one another zigzag we have 
	\begin{equation}
		|\zeta|\geq \sum_{\zeta' \in \text{zigzags}} |\det([\zeta'],[\zeta])| = \sum_{E' \in \text{sides of N}} \frac{|\det(E',E)|}{|E|}=2h_{E,N}. 
	\end{equation}
	where we used \eqref{eq:h in terms of edges}.
\end{proof}
	
We call zigzag \(\zeta\) \emph{minimal} if its length equals to \(2h_{E,N}\), where \(E\) is the side of \(N\) parallel to \(\zeta\).  There are no 2-valent vertices on minimal zigzags, since otherwise by constructing these vertices one would decrease the zigzag length.

Let \(\zeta_1, \zeta_2\) be two minimal parallel zigzags, which are consecutive in the sense of cyclic order. Let \(|\zeta_1|=|\zeta_2|=2l\). It is easy to see from the proof of Proposition~\ref{prop:zigzag length}, that each \(\zeta_i\) has common edges with \(2l\) different zigzag's segments. Hence the minimal number of edges in the open strip between \(\zeta_1\) and \(\zeta_2\) is \(l\). Moreover, if there are exactly \(l\) edges between them, they connect white vertices of one zigzag with black vertices of another one (depending on orientation). Hence the strip between \(\zeta_1\) and \(\zeta_2\) is a union of \(l\) hexagons, see Fig.~\ref{fi:zigzag permutation}.

\begin{Definition}  \label{def:patch}
	Let  \(\zeta_1,\dots, \zeta_h\) be a set of parallel zigzags consecutive in the sense of cyclic order. We will say that bipartite graph between them is a \emph{hexagonal patch \(\Pi_{l,h}\) of height \(h\) and length \(2l\)}, \(h< l\) if 
	\begin{itemize}
		\item  zigzags \(\zeta_i\) are minimal of the length \(2l\), \(\forall  1 \leq i \leq h\);
		\item the bipartite graph between \(\zeta_i\) and \(\zeta_{i+1}\) consists of \(l\) edges, \(\forall 1 \leq i \leq h-1\).
	\end{itemize}	
\end{Definition}

As explained above, the bipartite graph between \(\zeta_1\) and \(\zeta_h\) consists of \(h-1\) layers of hexagons, see an example on Fig.~\ref{fi:patch}. 

\begin{figure}[h]
	\begin{center}
		\includegraphics[scale=1]{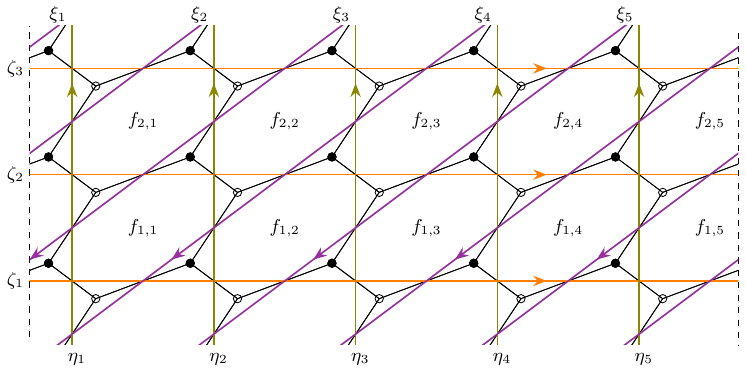}
	\end{center}		
	
	\caption{\label{fi:patch} Hexagonal patch \(\Pi_{5,3}\). It is assumed to be drawn on a cylinder, namely dashed lines on the left and right should are glued.}
\end{figure}

We label the hexagonal faces by \(f_{a,j}\), \(0 \leq a \leq h, 1 \leq j \leq l\) as on Fig.~\ref{fi:patch}. Let \(x_{a,j}=\wt(f_{a,j})\) be the corresponding face variables. We denote the segments of other zigzags~\footnote{Note that it is possible that, say, some $\xi_i$ and $\xi_j$ are segments of the same zigzag which intersects with all \(\zeta_k\) more than once.} intersecting with \(\zeta_1,\dots, \zeta_h\) by \(\xi_1,\dots, \xi_l, \eta_1,\dots, \eta_l \).
We also say, that side \(E\) corresponds to patch \(\Pi_{l,h}\) if it is parallel to zigzags \(\zeta_1,\dots,\zeta_h\).

The following lemma stays that any \(h\) parallel zigzags can be transformed to a patch using 4-gon mutations, contractions/uncontractions and zigzag transpositions \(R\) (for the case of \(h=2\) see~\cite{George:2022discrete}). Note that in this lemma (and eight other lemmas and theorems below) we refer to the paper in preparation \cite{Inprog}, so the reader can also consider these statements as conjectures supported by certain examples and computations given in this paper.

\begin{Lemma}[\cite{Inprog}]  \label{lem:zigzag patch}
	For any \(h\) parallel consecutive zigzags \(\zeta_1, \dots, \zeta_h\) there is a consistent dimer model such that all of them are minimal and the bipartite graph between them is given by hexagonal patch.
\end{Lemma}

There are no 2-valent vertices on zigzags \(\zeta_1,\dots,\zeta_h\) since they are minimal. Performing uncontractions if necessary one can assume that vertices of top and bottom zigzags \(\zeta_1,\zeta_h\) which are connected with exterior part of the graph are 3-valent. Let \(d_{1},\dots, d_l\) denote the edges below \(\zeta_1\) and \(u_{1},\dots, u_{l}\) denote the edges above \(\zeta_h\), all outside the patch. Each of zigzag's segments \(\eta_1,\dots, \eta_l,\xi_1,\dots,\xi_l\) has one of the \(\mathbf{d}\)-edges and one from \(\mathbf{u}\). According to the directions of the segments we can idefine two maps:  \(\eta \colon \mathbf{d} \to \mathbf{u} \) and \(\xi \colon \mathbf{u} \to \mathbf{d} \). It is easy to see that their composition \(\xi\circ\eta \colon \mathbf{d} \to \mathbf{d} \) is shift by \(l-h\) {(for \(h\leq l\))}  along the direction of \(\zeta_1\) (to the right on Fig.~\ref{fi:patch})

\begin{Definition}  \label{def:zigzag mutations}
	Zigzag mutation of consistent bipartite graph along the patch \(\Pi_{l,h}\) (\(h\leq l\)) is a surgery which removes patch \(\Pi_{l,h}\) and glues  upside-down the ``dual'' patch \(\Pi_{l,l-h}\) so that the maps \(\eta \colon \mathbf{d} \to \mathbf{u} \) and \(\xi \colon \mathbf{u} \to \mathbf{d} \) are preserved.
\end{Definition}

In the example with \(l=5\) and \(h=3\) the transformation of bipartite graph is presented on Fig.~\ref{fi:mut l5h3}. The transformation of zigzags for this mutation is depicted on Fig.~\ref{fi:mut zigzag l5h3} (see also examples with \(l=2, h=1\) on Fig.~\ref{fi:zigzag4}, and for \(l=3, h=1\) on Fig.~\ref{fi:zigzag6}).
\begin{figure}[h]
	\begin{center}
		\includegraphics[scale=1]{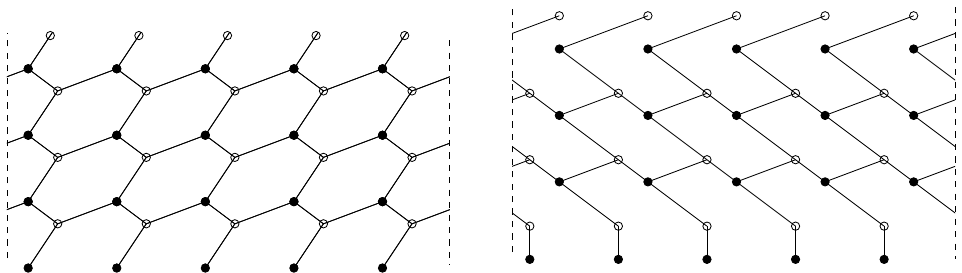}
	\end{center}		

	\caption{\label{fi:mut l5h3}Mutation along the patch \(\Pi_{5,3}\)  (with \(3\) zigzags of length \(10\)): graph transformation.}
\end{figure}
\begin{figure}[h]
	\begin{center}
		\includegraphics[scale=1]{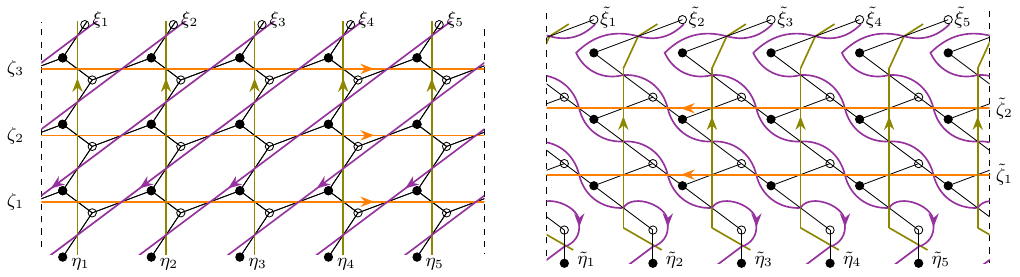}
	\end{center}	
	\caption{\label{fi:mut zigzag l5h3} Zigzag's transformation for the mutation from Fig.~\ref{fi:mut l5h3}.}
\end{figure}

Definition~\ref{def:zigzag mutations} is a special case of the construction from \cite{Higashitani:2022}.

\begin{Lemma} \label{lem:zigzag mut - polyn mut}
	Let zigzag mutation map a consistent bipartite graph \(\Gamma\) to \(\tilde{\Gamma}\). Then \(\tilde{\Gamma}\) is consistent bipartite graph. Moreover, the Newton polygon  corresponding to \(\tilde{\Gamma}\) is given by \(\widetilde{N}= \mu_{E,h}(N)\), where \(N\) is Newton polygon corresponding to \(\Gamma\), and \(E\) is the side corresponding to the patch.
\end{Lemma}
\begin{proof}
	The check of consistency is straightforward. In order to show mutation property of the polygon it is sufficient to see that transformation of homology classes of zigzags coincides with that one from Proposition~\ref{prop:mutatation edges}. This follows from the following transformations of zigzag's segments (see Fig.~\ref{fi:mut zigzag l5h3} for \(l=5\) and \(h=3\)):
	\begin{equation}
		[\tilde{\zeta}_{1}]=\dots=[\tilde{\zeta}_{l-h}]=-[\zeta_1]=\dots=-[\zeta_h]; \qquad [\tilde{\eta}_i]=[\eta_i],\;\;  [\tilde{\xi}_i]=[\xi_i]+[\zeta_1], \;\;\forall 1\leq i \leq l.
	\end{equation}	
	Here \([\xi_i], [\eta_i], [\tilde{\xi}_i], [\tilde{\eta}_i]\) denote classes in homology of the cylinder relative to its boundary.
\end{proof}

\subsection{Zigzag mutations 2: dimer models} \label{ssec:zigzag mut 2}

Now it is natural to ask about the zigzag mutation of the dimer model. In other words, we want to find a transformation of the weights \(\wt \mapsto \widetilde{\wt}\) such that transformation of partition function \(\mathcal{Z}(\lambda,\mu) \mapsto \widetilde{\mathcal{Z}}(\lambda,\mu)\) would be mutation of Laurent polynomials.


Note that mutation of the polynomial requires the conditions \eqref{eq:polyn mut cond}. These conditions determine an ideal in the algebra of functions on the cluster chart. To be more precise, the condition \eqref{eq:polyn mut cond} for \(k=h\) means there exists \(h\) zigzags \(\zeta_1,\dots,\zeta_h\) parallel to \(E\) such that 
\begin{equation}\label{eq:z1=dots=zh}
	z_1=\dots = z_h=C\lambda^{-1},
\end{equation}
see~Theorem \ref{th:zigzags boundary}. This implies \(h-1\) relations in  \(\mathcal{O}(\mathcal{X}_\Gamma) \) given by \( z_1 z_2^{-1}= \dots= z_{h-1} z_h^{-1}= 1\). All these relation are imposed on the Casimir functions. The conditions~\eqref{eq:polyn mut cond} for \(1\leq k<h\) are certain linear relations for the Hamiltonians and Casimir functions of GK integrable system. It is easy to see that totally one has \(h-1 + h(h-1)/2=(h-1)(h+2)/2\) relations. We denote the ideal generated by these relations by \(I_{E,h} \subset \mathcal{O}(\mathcal{X}_\Gamma) \). 

It is convenient to consider a bit larger ideal, to be defined as follows.

\begin{Notation} We use the following notation
	\begin{equation}
		\mathsf{S}(x_1,\dots,x_k)=x_1+x_1 x_2+\ldots+x_1x_2\cdot\ldots\cdot x_k = x_1(1+x_2(1+\ldots x_{k-1}(1+x_k)\ldots )).
	\end{equation}
\end{Notation}

For given strip \(\Pi_{l,h}\) let us introduce (recall that \(x_{a,j}=\wt(f_{a,j})\), see Fig.~\ref{fi:patch})
\begin{equation}\label{eq:C,H}
	C_a= \prod\nolimits_{j=1}^{l} x_{a,j} ,  \;\;
	 H_{a,a+1}= \mathsf{S}(x_{a,1},\dots,x_{a,l-1}), \qquad \text{ for } 1\leq a < h.
\end{equation}
t is easy to see that \(C_a=z_az_{a+1}^{-1}\), in particular \(C_a\) is a Casimir function. The functions \(H_{a,a+1}\) serve as a generators of the algebra of Hamiltonians for Hamiltonian reduction.

\begin{Lemma}[\cite{Inprog}]  \label{lem:Sevostyanov}
	The functions \(H_{a,a+1}\) generate Poisson algebra with algebraic generators  \(H_{a,b}\) \(1\leq a<b \leq h\) and Poisson brackets 
\begin{equation}\label{eq:Sevostyanov}
	\{H_{a,b},H_{c,d}\}=
	\begin{cases}
		H_{a,b}H_{c,d} & \text{if } a=c, b<d \\
		H_{a,b}H_{c,d} & \text{if } a<c, b=d \\
		H_{a,d}        & \text{if } a<b=c<d \\
		H_{a,b}H_{c,d}+H_{a,d}H_{c,b} & \text{if } a<c<b<d \\
		0              & \text{if } a<c<d<b \\
		0              & \text{if } a<b<c<d
	\end{cases}
\end{equation}
\end{Lemma}

This algebra is classical analog of the algebra introduced by Sevostyanov in \cite{Sevostyanov:1999}. 

\begin{Lemma}[\cite{Inprog}] \label{lem:J implies I}
	Assume we have consistent dimer model with patch \(\Pi_{l,h}\). Let \(J_{l,h}\subset \mathcal{O}(\mathcal{X}_{\Gamma})\) be an ideal generated by 
	\begin{equation}
		J_{l,h}=(\{C_a-1, H_{a,a+1}+1, H_{a,b} \mid 1\leq a < h, a+1<b\leq h\}).
		\label{Jdef}
	\end{equation}	
	Then the submanifold \(\mathbf{V}(J_{l,h})\subset \mathcal{X}_\Gamma\) defined by ideal \(J_{l,h}\) is a component of maximal dimension in \(\mathbf{V}(I_{E,h})\subset \mathcal{X}_\Gamma\). In particular, the polynomial mutation conditions \eqref{eq:polyn mut cond} are fulfilled on \(\mathbf{V}(J_{l,h})\).
\end{Lemma}

In particular, both \(\mathbf{V}(J_{l,h})\subset \mathcal{X}_\Gamma\) and \(\mathbf{V}(I_{E,h})\subset \mathcal{X}_\Gamma\) are of the same codimension.
We will see embedding \(\mathbf{V}(J_{l,h})\subset \mathbf{V}(I_{E,h})\) in the examples below (namely in Sections~\ref{sec:E7} and \ref{sec:E8}).

\begin{Theorem}[\cite{Inprog}]  \label{th: mutation ideal}
	Let \(\Gamma\) be a consistent bipartite graph which contains a patch \(\Pi_{l,h}\). Let \(\tilde{\Gamma}\) be a consistent bipartite graph which contains a patch \(\Pi_{l,l-h}\) and obtained from \(\Gamma\) by zigzag mutation along the patch. Then there is a transformation  \(\wt \mapsto \widetilde{\wt}\), \(K \mapsto \tilde{K}\), inducing a map from \(\mathbf{V}(J_{l,h}) \subset \mathcal{X}_\Gamma \) to \(\mathbf{V}(J_{l,l-h}) \subset \mathcal{X}_{\tilde \Gamma} \), such that the corresponding partition functions \(\mathcal{Z}\) and \(\tilde{\mathcal{Z}}\) are connected by mutation of polynomial.
\end{Theorem}

In the Section~\ref{ssec:l=3 h=1} we present an explicit example of such transformations.

\begin{Remark} 
	Ideal \(J_{l,h}\) has another description. Consider patch \(\Pi_{l,h}\). Let \(e_a\) denote the edge of zigzag \(\zeta_i\), that separates faces \(f_{a-1,l}\) and \(f_{a,l}\). Similarly let \(e_1\) and \(e_h\) be edges of zigzags \(\zeta_1\) and \(\zeta_h\) correspondingly that separate faces \(f_{1,l}\) and \(f_{h-1,l}\) and boundary of the patch.

	Remove edges going outside the strip (they were labeled \(d_1,\dots,d_l,u_1,\dots,u_l\) above). Also cut $h$ edges \(e_1,\dots ,e_h\) defined above. Thus the cylinder becomes a strip. Orient remaining edges (and half-edges) along the direction of zigzags \(\zeta_1, \dots , \zeta_h\). See Fig.~\ref{fi:pathch orhgogonal} for an example of \(l=5, h=3\).
    \begin{figure}[h]
        \centering
        \includegraphics[scale=1]{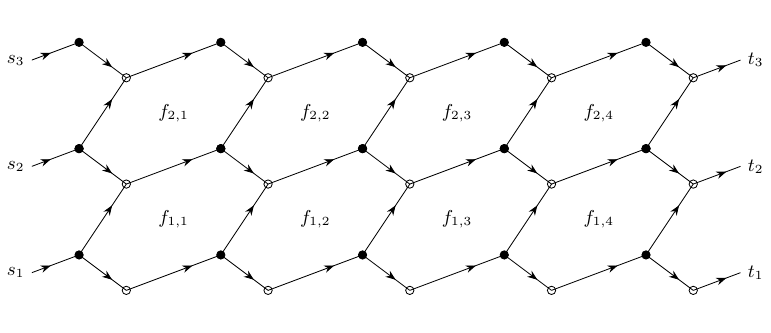}
        \caption{\label{fi:pathch orhgogonal}
        	Strip obtained from the Patch on Fig.~\ref{fi:patch}.}
    \end{figure}
    
	 After such surgery each zigzag \(\zeta_a\) becomes an oriented  segment. Let \(s_a\) denotes its source and \(t_a\) its target. Orient all edges of the patch parallel to zigzags \(\zeta_a\), see Fig.~\ref{fi:pathch orhgogonal}. Let \(L\) be \(h\times h\) parallel transport matrix defined by 
	 \begin{equation}
	 	L_{a,b}=\sum_{p\colon s_b \to t_a} \wt(p),
	 \end{equation}
	 where summation goes over oriented paths from \(s_b\) to \(t_a\) on the patch (after cuttings). The weight \(\wt(p)\) is defined as a product of weights of edges in \(p\). The weights of half-edges is defined such that product of weights of halfs of \(e_a\) is equal to \(\wt(e_a)\).
	 
	 It is easy to see that matrix \(L\) is lower triangular. Moreover, the diagonal elements are equal to the zigzag weights \(L_{a,a}=\wt(\zeta_a)\) and elements on the lower diagonal have the form \(L_{a+1,a}\sim H_{a,a+1}+1\). 
	 
	 The main message of this remark is that it can be shown that the submanifold \(\mathbf{V}(J_{l,h})\) is the submanifold on which the parallel transport matrix is scalar \(L=\lambda \mathds{1}_h\), for some \(\lambda\).
\end{Remark}

\begin{Remark}
	The patch \(\Pi_{l,h}\) has obvious cyclic \(\mathbb{Z}/l\mathbb{Z}\) symmetry . The definition \eqref{Jdef} of the ideal \(J_{l,h}\) naively is not invariant under this symmetry, since the Hamiltonians \(H_{a,a+1}\) defined in~\eqref{eq:C,H} depend on the choice of `initial'' face on each strip. However the ideal \(J_{l,h} \in \mathcal{O}(\mathcal{X}_\Gamma)\) actually does not depend on this choice. Indeed, if we replace \(H_{a,a+1}\), for example, by next telescoping sum $\mathsf{S}(x_{a,2},\dots,x_{a,l})$ we get an element of the same ideal due to relation
	\begin{multline}
		1+\mathsf{S}(x_{a,2},\dots,x_{a,l})=1+x_{a,2}+x_{a,2}x_{a,3}+\dots+x_{a,2}\cdots x_{a,l}
		\\
		=x_{a,1}^{-1}(C_a-1)+x_{a,1}^{-1}(1+H_{a,a+1}) \in J_{l,h}.
	\end{multline}	
\end{Remark}

\subsection{Zigzag mutations 3: Poisson geometry} \label{ssec:zigzag Poisson}

The ideal \(I_{E,h}\) is generated by linear combinations of Hamiltonians and Casimir functions. Hence, it is closed under the Poisson bracket. It follows from the Lemma~\ref{lem:Sevostyanov} that the ideal \(J_{l,h}\) has the same property.
However \(\mathbf{V}(J_{l,h})\) are not Poisson submanifolds for \(h>1\). 

\begin{Definition} \label{def:reduction} 
	Let \(\Gamma\) be a consistent bipartite graph containing patch \(\Pi_{l,h}\). \emph{Hamiltonian reduction with respect to the patch} is a spectrum of algebra 
	\begin{equation}
		(\mathcal{O}(\mathcal{X}_\Gamma)/J_{l,h})^{\{\bullet,J_{l,h}\}}= \Big\{ f+J_{l,h}\Big| \{f,J_{l,h}\}\subset J_{l,h}, f \in \mathcal{O}(\mathcal{X}_\Gamma)\Big\}.	
	\end{equation}
\end{Definition}

To be more precise, we want to perform Hamiltonian reduction of the whole cluster variety \(\mathcal{X}\), but above we defined reduction of only one cluster chart \(\mathcal{X}_\Gamma\) corresponding to bipartite graph \(\Gamma\). In order to have whole reduction \(\mathcal{X}_{\mathrm{red}}\) we should do the same procedure for all charts. 

It appears that there exist more convenient cluster charts (which do not correspond to bipartite graphs) where reduction can be done via monomial map.

\begin{Lemma}[\cite{Inprog}]   \label{lem:patch cluster chart}
    (a) 	Let \(\Gamma\) be a consistent bipartite graph containing a patch \(\Pi_{l,h}\). Let \(\mathsf{s}\) denotes corresponding seed. Then there exists a sequence of cluster mutations \(\boldsymbol{\mu}\colon \mathsf{s} \mapsto \mathsf{t} \) such that the image of the ideal \(J_{l,h}\) is generated by \(\{m_y^{(i)}+1\mid1 \leq i \leq (h-1)(h+2)/2\}\), where  \(\mathbf{y}\) denotes cluster variables in seed $\mathsf{t}$ and \(m_y^{(i)}\) are monomial in \(\mathbf{y}\).

    (b) Certain monomials in \(\mathbf{y}\) are coordinates on the reduction of the 
    chart \(\mathcal{X}_\mathsf{t}\). We will denote these coordinates by \(\mathbf{w}\).
    
    (c) The dimension of the reduction variety is equal to \(\dim \mathcal{X}_{\mathrm{red}} =\dim \mathcal{X} - (h^2-1)\).
    
    (d) The Poisson rank of the reduction variety \(\mathcal{X}_{\mathrm{red}}\) is equal to \(2I-h(h-1)\).
\end{Lemma}

We demonstrate this below by an example in Sect.~\ref{ssec:l=3 h=1}, formulas \eqref{eq:example l3h1 y} and \eqref{ex:wy}.
Far more nontrivial examples of this statement can be found in Sect.~\ref{sec:E7} (see Fig.~\ref{fig:E7quivers} and formulas \eqref{eq:JE7 in y}, \eqref{glu_var}) and Sect.~\ref{sec:E8} (see \eqref{muE8}, \eqref{eq:E8 J in y} and \eqref{wE8}).

It follows from monomiality that Poisson bracket of the coordinates \(\mathbf{w}\) on the reduction \(\mathcal{X}_{\mathrm{red}}\) has cluster form \(\{w_i,w_j\}=b_{ij}^{\mathrm{red}}w_iw_j\) for certain integral skew-symmetric matrix \(b^{\mathrm{red}}\). It is natural to define the corresponding cluster seed by \(\mathsf{t}^{\mathrm{red}}=(b^{\mathrm{red}},\mathbf{w})\). 

\begin{Theorem}[\cite{Inprog}]   \label{th:mut red patch}
	Let \(\Gamma\) be a consistent bipartite graph which contains a patch \(\Pi_{l,h}\). Let \(\tilde{\Gamma}\) be a consistent bipartate graph which contains a patch \(\Pi_{l,l-h}\) and obtained from \(\Gamma\) by zigzag mutation along the patch. Let us assume that weights for \(\Gamma\) and \(\tilde{\Gamma}\) are related as proposed in Theorem \ref{th: mutation ideal}. Let  \((b^{\mathrm{red}},\mathbf{w})\) and \((\tilde{b}^{\mathrm{red}},\tilde{\mathbf{w}})\) denote corresponding seeds after reduction. Then
    \(b^{\mathrm{red}}=\tilde{b}^{\mathrm{red}}\) and \(\mathbf{w}=\tilde{\mathbf{w}}\).
\end{Theorem}

Recall that the ideal \(J_{l,h}\) is generated by \(h-1\) Casimir functions \(C_a\) and \(h(h-1)/2\) functions \(H_{a,b}\). Then it follows from Lemma~\ref{lem:patch cluster chart} (c) that these functions are algebraically independent and, moreover, functions \(H_{a,b}\) are Poisson independent in sense that Poisson commutativity with them gives 
\(h(h-1)/2\) independent constraints. 

\begin{Theorem}[\cite{Inprog}]\label{th:Z,J}
	 Let \(\Gamma\) be a consistent bipartite graph which contains a patch \(\Pi_{l,h}\). Then for any \((a,b)\in N\) we have \(\{\mathcal{Z}_{a,b},J_{l,h}\}\subset J_{l,h}\).
\end{Theorem}

This theorem means that Hamiltonians of Goncharov-Kenyon integrable system descend to well defined functions on the variety \(\mathcal{X}_{\mathrm{red}}\). Moreover, since they Poisson commute on \(\mathcal{X}_\Gamma\) they would commute on \(\mathcal{X}_{\mathrm{red}}\). 
Note that analogous statetement for the ideal \(I_{E,h}\) just reads \(\{\mathcal{Z}_{a,b},I_{E,h}\}=0 \) and follows easily from the fact, that \(I_{E,h}\) is generated by linear combinations of \(Z_{a,b}\)'s.

In other words one can descend dimer partition function \(\mathcal{Z}(\mathbf{x}|\lambda,\mu)\) to the reduction. We denote obtained function by \(\mathcal{Z}_{\mathrm{red}}(\mathbf{w}|\lambda,\mu)\).

\begin{Conjecture}
	The Goncharov-Kenyon Hamiltonians \(\mathcal{Z}_{a,b}(\mathbf{x})\) for \((a,b) \in (\text{interior of } N)\) define integrable system  on the subvariety of \(\mathcal{X}_{\mathrm{red}}\) given by equation \(q=1\).
\end{Conjecture}

Here \(q\) is a Casimir function given by \(q=\prod_{f_i \in F(\Gamma)} x_i\) in the seed \(\mathsf{s}\), but now considered as a function on  \(\mathcal{X}_{\mathrm{red}}\).

Taking into account Lemma~\ref{lem:patch cluster chart} (d) this conjecture means that there are \(I-h(h-1)/2\) algebraically independent Goncharov-Kenyon Hamiltonians on  \(\mathcal{X}_{\mathrm{red}}\). On the other hand, the only relations on these Hamiltonians imposed by the ideal \(I_{E,h}\) are \(h(h-1)/2\) linear relations given in conditions~\eqref{eq:polyn mut cond}. Using Lemma~\ref{lem:J implies I} we see that there are the same linear relations on \(\mathbf{V}(J_{l,h})\). So the conjecture means, that there are no other constraints on Goncharov-Kenyon Hamiltonians on the submanifold \(\mathbf{V}(J_{l,h})\).

We have the similar counting in terms of the genus of spectral curve, see the Definition~\ref{def:spectral curve}. Note that conditions~\eqref{eq:polyn mut cond} imply that closure of open curve \(\mathcal{C}\) has singularity of type \(x^h=y^h\) at the infinity. Its resolution decreases the genus by \(h(h-1)/2\) (this follows e.g. from the formula for genus in terms of Milnor number \cite[Cor 7.1.3]{Wall:2004singular} or from  adjunction formula). Hence we have

\begin{Proposition}
	Let \((\Gamma, \wt)\) be a consistent dimer model. Assume that \(\wt\) satisfies conditions~\eqref{eq:polyn mut cond} and generic with this property. Then the genus of the spectral curve \(\overline{\mathcal{C}}\) is equal to \(g(\overline{\mathcal{C}})=I-{h(h-1)}/{2}\).
\end{Proposition}

\subsection{Example: length 6 zigzags}\label{ssec:l=3 h=1}
Let us illustrate constructions above on the example of the patch with \(l=3\) and \(h=1\). According to the Definition~\ref{def:zigzag mutations} zigzag mutation leads to the transformation of the graphs depicted in Fig.~\ref{fi:zigzag6}.
	
	
	\begin{figure}[h]
		\begin{center}
			\includegraphics[scale=0.8]{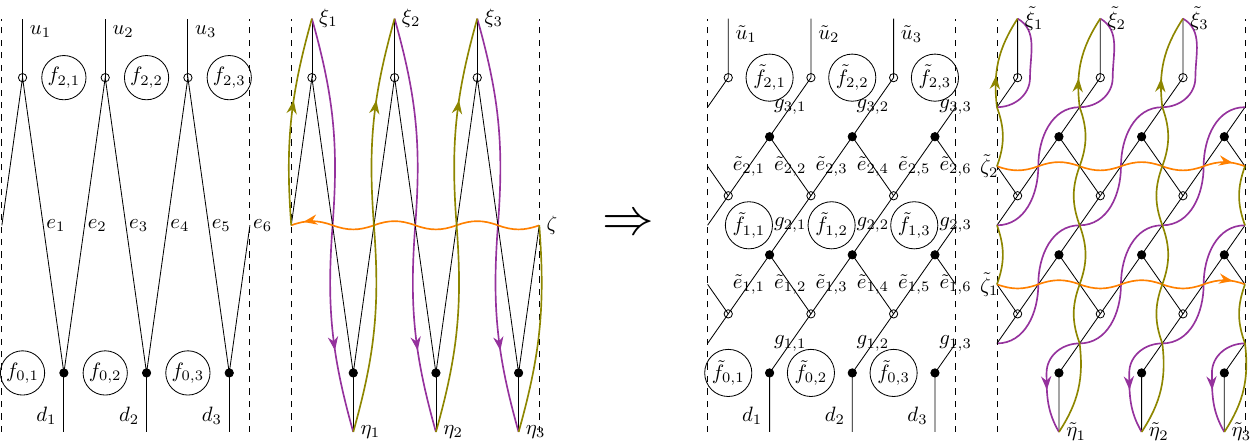}
		\end{center}		
		\caption{\label{fi:zigzag6}
			Length 6 zigzag mutation: transformation of graph, edge weights, and  zigzags.
		}
	\end{figure}
	
	Let us now give transformation of the edge weights and Kasteleyn orientation that ensures mutation of the partition function. For simplicity we will give formulas in a certain gauge. Namely, using gauge transformations we can assume that 
	\begin{equation}
		\sgn_K(e_i)=1,\; \wt(e_i)=1, \qquad 1\leq i \leq 5.
	\end{equation}
	Let us also assume that \(\lambda\) is chosen such that \(\lambda^{-1}=z=\wt(\zeta)\sgn_K(\zeta)=\wt(e_6)^{-1}\sgn_K(e_6)\). Finally, we assume that the weights of the edges \(u_1,u_2,u_3\) are proportional to an inverse of another spectral parameter~\(\mu^{-1}\) and all other edge weights are \(\mu\) independent. 
	
	
	Then transformed weights and Kasteleyn signs are given by 
	\begin{subequations}\label{eq:l3h1 edges}
		\begin{align}
			&\widetilde{\sgn}_{K}(\tilde{e}_{ij})=-\sgn_K(\zeta),\quad \widetilde{\wt}(\tilde{e}_{ij})=-\lambda \qquad \text{for } j=2i,
			\\
			&\widetilde{\sgn}_{K}(\tilde{e}_{ij})=-1,\quad \widetilde{\wt}(\tilde{e}_{ij})=-1 \qquad \text{for $j$ even and } j\neq 2i,
			\\
			&\widetilde{\sgn}_{K}(\tilde{e}_{ij})=1,\quad \widetilde{\wt}(\tilde{e}_{ij})=1 \qquad \text{for $j$ odd}, \quad \widetilde{\sgn}_{K}(\tilde{g}_{ij})=1, 
			\\
			& \widetilde{\wt}(\tilde{u}_i)=\wt(u_i)(1+\lambda)^{-1}, \text{ for } i=1,2,3,
			\\
			&
			\widetilde{\wt}(g_{1,1})= 1,\; \;
			\widetilde{\wt}(g_{1,2})= - 1-\beta^{-1},  \;\;
			\widetilde{\wt}(g_{1,3})= \beta^{-1},
			\\
			&\widetilde{\wt}(g_{2,1})=1, \; \;
			\widetilde{\wt}(g_{2,2})= - \frac{1}{ 1+\beta}, \;\; 
			\widetilde{\wt}(g_{2,3})= - \frac{\beta} {\beta+1},
			\\
			& \widetilde{\wt}(g_{3,1})=1, \;\; \widetilde{\wt}(g_{3,2})=\beta, \;\; \widetilde{\wt}(g_{3,3})= -1- \beta.		 
		\end{align}
	\end{subequations}
	These formulas illustrate Theorem~\ref{th: mutation ideal}.
	
	It is straightforward to check that partition function under this transformation is mutated as in formula~\eqref{eq:PolMut}, where \(\nu=\mu(1+\lambda)\). In particular, the partition function does not depend on~\(\beta\). Note that \(\beta\) is not a gauge transformation since the face variables depend on it, see \eqref{eq:l3h1 faces}.  The parameter \(\beta\) is dual to the single Hamiltonian of reduction (second formula in~\eqref{eq:example l3h1 ideal}) and will disappear after the Hamiltonian reduction below.

	The weights and signs \eqref{eq:l3h1 edges} are defined up to a gauge transformation. Here we used the gauge different from the one in Definition \ref{def:zigzag 4 mutatation}.
	
	The weights of face variables are transformed as 
	\begin{subequations}\label{eq:l3h1 faces}
		\begin{align}
			\tilde{x}_{0,1}&= -{\beta}x_{0,1},&\;		\tilde{x}_{0,2}&=\frac{1+\beta}{\beta}x_{0,2},&\;
			\tilde{x}_{0,3}&=\frac{1}{1+\beta}x_{0,3},   
			\\
			\tilde{x}_{1,1}&= \frac{1+\beta}{-\beta},&\;
			\tilde{x}_{1,2}&= \frac{-1}{1+\beta},&\;
			\tilde{x}_{1,3}&={\beta},
			\\
  			\tilde{x}_{2,1}&= \frac{1}{1+\beta}  x_{2,1},&\;
			\tilde{x}_{2,2}&= -{\beta}x_{2,2},&\;
  			\tilde{x}_{2,3}&= \frac{1+\beta}{\beta}x_{2,3}
  			.
		\end{align}
	\end{subequations}
	where \(x_{i,j}=\wt(f_{i,j}), \tilde{x}_{i,j}=\widetilde{\wt}(f_{i,j})\). The point in  \(\mathcal{X}_{\tilde \Gamma} \) defined by \eqref{eq:l3h1 faces} belongs to \(\mathbf{V}(J_{3,2})\). Explicitly this means (cf. with formula~\eqref{eq:C,H}) that
	\begin{equation}\label{eq:example l3h1 ideal}
		1-\tilde{x}_{1,1}\tilde{x}_{1,2}\tilde{x}_{1,3}=0,\quad 1+\tilde{x}_{1,1}+\tilde{x}_{1,1}\tilde{x}_{1,2}=0.
	\end{equation}

	The face quiver for variables \(\tilde{x}\) is drawn on Fig.~\ref{fi:l3h1 example}. 
	
	\begin{figure}[h]
		\begin{center}
			\begin{tikzpicture}[scale=1.5, font = \small]
				\def\xs{1} 
				\def\ys{1}
				\node[styleNode] (x11) at(0.5*\xs,\ys){$\tilde{x}_{2,1}$};
				\node[styleNode] (x12) at(1.5*\xs,\ys){$\tilde{x}_{2,2}$};
				\node[styleNode] (x13) at(2.5*\xs,\ys){$\tilde{x}_{2,3}$};
				\node[styleNode] (x21) at(0,0){$\tilde{x}_{1,1}$};
				\node[styleNode] (x22) at(\xs,0){$\tilde{x}_{1,2}$};
				\node[styleNode] (x23) at(2*\xs,0){$\tilde{x}_{1,3}$};
				\node[styleNode] (x31) at(0.5*\xs,-\ys) {$\tilde{x}_{0,1}$};
				\node[styleNode] (x32) at(1.5*\xs,-\ys) {$\tilde{x}_{0,2}$};
				\node[styleNode] (x33) at(2.5*\xs,-\ys){$\tilde{x}_{0,3}$};
				
				\draw[styleArrow] (x21) to (x22);
				\draw[styleArrow] (x22) to (x23);
				\draw[styleArrow] (x23) to [bend left=20] (x21);
				
				\draw[styleArrow] (x32) to (x22);
				\draw[styleArrow] (x33) to (x23);
				\draw[styleArrow] (x31) to (x21);
				\draw[styleArrow] (x22) to (x31);
				\draw[styleArrow] (x23) to (x32);
				\draw[styleArrow] (x21) to [bend right=5] (x33);

				\draw[styleArrow] (x22) to (x12);
				\draw[styleArrow] (x23) to (x13);
				\draw[styleArrow] (x21) to (x11);
				\draw[styleArrow] (x13) to (x22);
				\draw[styleArrow] (x11) to [bend right=5] (x23);
				\draw[styleArrow] (x12) to (x21);
				
			\end{tikzpicture}
			\qquad\qquad 
			\begin{tikzpicture}[scale=1.5, font = \small]
				\def\xs{1} 
				\def\ys{1}
				\node[styleNode] (x11) at(0.5*\xs,\ys){$y_{2,1}$};
				\node[styleNode] (x12) at(1.5*\xs,\ys){$y_{2,2}$};
				\node[styleNode] (x13) at(2.5*\xs,\ys){$y_{2,3}$};
				\node[styleNode] (x21) at(0,0){$y_{1,1}$};
				\node[styleNode] (x22) at(\xs,0){$y_{1,2}$};
				\node[styleNode] (x23) at(2*\xs,0){$y_{1,3}$};
				\node[styleNode] (x31) at(0.5*\xs,-\ys) {$y_{0,1}$};
				\node[styleNode] (x32) at(1.5*\xs,-\ys) {$y_{0,2}$};
				\node[styleNode] (x33) at(2.5*\xs,-\ys){$y_{0,3}$};
				
				\draw[styleArrow] (x23) to (x22);
				\draw[styleArrow] (x21) to [bend right=20] (x23);
				
				\draw[styleArrow] (x23) to (x33);
				\draw[styleArrow] (x31) to (x21);
				\draw[styleArrow] (x22) to (x31);
				\draw[styleArrow] (x32) to (x23);

				\draw[styleArrow] (x22) to (x12);
				\draw[styleArrow] (x13) to (x23);
				\draw[styleArrow] (x23) to [bend left=5] (x11);
				\draw[styleArrow] (x12) to (x21);
				
				\draw[styleArrow] (x11) to [bend right=20](x32);
				\draw[styleArrow] (x33) to (x32);
				\draw[styleArrow] (x11) to [bend left=20](x13);
				\draw[styleArrow] (x33) to [bend right=10](x13);
			\end{tikzpicture}			
			\qquad\qquad 
			\begin{tikzpicture}[scale=1.5, font = \small]
				\def\xs{1} 
				\def\ys{1}
				\node[styleNode] (x11) at(0.5*\xs,\ys){$w_{2,1}$};
				\node[styleNode] (x12) at(1.5*\xs,\ys){$w_{2,2}$};
				\node[styleNode] (x13) at(2.5*\xs,\ys){$w_{2,3}$};
				\node[styleNode] (x31) at(0.5*\xs,-\ys) {$w_{0,1}$};
				\node[styleNode] (x32) at(1.5*\xs,-\ys) {$w_{0,2}$};
				\node[styleNode] (x33) at(2.5*\xs,-\ys){$w_{0,3}$};
				
				\draw[styleArrow] (x11) to (x32);
				\draw[styleArrow] (x12) to (x33);
				\draw[styleArrow] (x13) to [bend left=10] (x31);				
				
				\draw[styleArrow] (x31) to (x11);
				\draw[styleArrow] (x32) to (x12);
				\draw[styleArrow] (x33) to (x13);
				
				\draw[styleArrow] (x33) to (x32);
				\draw[styleArrow] (x32) to (x31);
				\draw[styleArrow] (x31) to [bend right=20](x33);
				
				\draw[styleArrow] (x13) to (x12);
				\draw[styleArrow] (x12) to (x11);				
				\draw[styleArrow] (x11) to [bend left=20](x13);				
			\end{tikzpicture}
		\end{center}
		\caption{ \label{fi:l3h1 example} Quiver for the patch after zigzag mutation, same quiver after (face) mutation in \(\tilde{x}_{2,3}\), quiver after reduction.} 
	\end{figure}
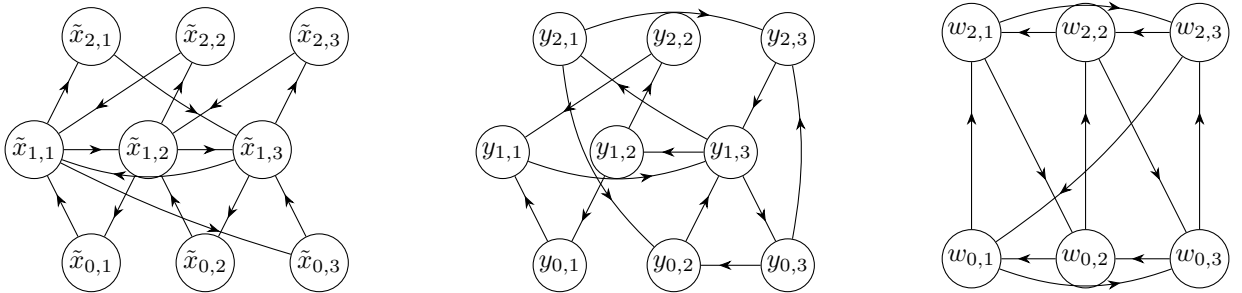
	
	Let us do mutation in the vertex \(\tilde{x}_{1,3}\). The resulting quiver is drawn on Fig.\ref{fi:l3h1 example} in the center. We denote cluster variables after mutation by  \(\mathbf{y}\). The ideal \(J_{3,2}\) in new chart is given by relations 
	\begin{equation}\label{eq:example l3h1 y}
		1+y_{1,1}=0,\quad  1+y_{1,2}=0.
	\end{equation}
	These relations are equivalent to defining relations \eqref{eq:example l3h1 ideal} however now they are given by binomials as in Lemma~\ref{lem:patch cluster chart}. Now we can perform Hamiltonian reduction with respect to the ideal generated by~\eqref{eq:example l3h1 y}. It is easy to see that the algebra of functions which Poisson commute with Hamiltonians~\eqref{eq:example l3h1 y} is generated by 
	\begin{equation}
		\label{ex:wy}
		w_{2,1}=y_{2,1},\;\; w_{2,2}=-y_{2,2}y_{1,3},\;\; w_{2,3}=y_{2,3},\;\; w_{0,1}=-y_{0,1}y_{1,3},\;\; w_{0,2}=y_{0,2},\;\; w_{0,3}=y_{0,3}.
	\end{equation} 
	The Poisson bracket between these functions has the cluster form~\eqref{eq:Poisson cluster}, where the adjacency matrix \(b^{\mathrm{red}}\) corresponds to the quiver drawn on Fig.\ref{fi:l3h1 example} on the right. Note that this quiver coincides with the face quiver for the patch before mutation (see. Fig. \ref{fi:zigzag6} left). Furthermore, using formulas for the face variables~\eqref{eq:l3h1 faces} we get 
	\begin{equation}
		w_{0,1}=x_{0,1},\;\; w_{0,2}=x_{0,2},\;\; w_{0,3}=x_{0,3},\;\; w_{2,1}=x_{2,1},\;\; w_{2,2}=x_{2,2},\;\; w_{2,3}=x_{2,3}.
	\end{equation} 
	Therefore the (cluster) face variables befor zigzag mutation are equal to the cluster variables after mutation and reduction. This illustrates Theorem \ref{th:mut red patch}.
	
	Remark once again that while the face variables \eqref{eq:l3h1 faces} depend on \(\beta\), however after reduction this dependence is gone.

\subsection{Reduced Goncharov-Kenyon integrable systems}
\label{ssec:reduced GK}

Conjecturally the results of previous sections can be extended into more general setting. 

Let us start from the generalization of the reduction construction. Let \(\Gamma,[\wt]\) be a consistent dimer model and \(\mathsf{s}\) be a corresponding seed. In the Definition~\ref{def:reduction} the reduction was performed by the ideal \(J_{l,h}\) assigned to a patch \(\Pi_{l,h}\). The patch itself is assigned to set of \(h\) parallel zigzags which are consecutive in the cyclic order (see Definition~\ref{def:patch}). Any \(h\) zigzags \(\zeta_1,\dots, \zeta_h\) that are ordered (but not necessary consecutively) in cyclic order can be made consecutive using zigzag transpositions defined in Prop. \ref{prop:zigzag transposition}. Then using Lemma~\ref{lem:zigzag patch} we can transform these zigzags into a patch and define ideal \(J_{|\zeta_1|,h}\). Using inverse transformation we can define corresponding ideal \(J_{\zeta_1,\dots, \zeta_h} \subset \mathcal{O}(\mathcal{X}_\Gamma)\).

\begin{Lemma}[\cite{Inprog}]
	The ideal \(J_{\zeta_1,\dots, \zeta_h} \subset \mathcal{O}(\mathcal{X}_\Gamma)\) does not depend on choices in the construction above.
\end{Lemma}

Hence one can perform Hamiltonian reduction with respect to ideal \(J_{\zeta_1,\dots, \zeta_h}\). Furthermore, one can perform Hamiltonian reduction with respect to ideals corresponding several to sets of zigzags. This motivates the following definition.

\begin{Definition} \label{def:dimer decorations}
	\emph{Consistent dimer model with decoration} is the following data 
	\begin{itemize}
		\item Consistent dimer model \((\Gamma, [\wt])\)
		\item Partition of set of zigzags 
		\(
			\bigcup_{E, i} \{\zeta_{E,i,1},\zeta_{E,i,2},\dots, \zeta_{E,i,h_{E,i}}\}, 
		\)	
		where \(E\) runs over sides of Newton polygon \(N\) and for any \(E,i\) zigzags \(\zeta_{E,i,1},\dots, \zeta_{E,i,h_{E,i}}\) are ordered in cyclic order and parallel to \(E\).
	\end{itemize}
\end{Definition}

For a consistent dimer model with decoration denote collection of zigzags \(\zeta_{E,i,1},\dots, \zeta_{E,i,h_{E,i}}\) by  \(\boldsymbol{\zeta}_{E,i}\). As was explained above, for any such collection we have an ideal \(J_{E,i}=J_{\boldsymbol{\zeta}_{E,i}}\). Let \(J\) be a an ideal generated by \(J_{E,i}\) for all \(E,i\).  The following lemma implies that this ideal is closed under the Poisson bracket.
\begin{Lemma}[\cite{Inprog}]
	Let \(E,E'\) denote two (possibly coincident) sides of \(N\). Then for any \(i,i'\) we have
	\(	\left\{J_{E,i},J_{E',i'}\right\}
		\subset 
		\left(J_{E,i},J_{E',i'}
		\right)\).

%
\end{Lemma}

The main conjecture of the paper is that under certain conditions the Hamiltonian reduction of cluster variety $\mathcal{X}$ by ideal \(J\) is a phase space of a certain integrable system. Let us first formulate the meaning of reduction conditions in terms of dimer partition function and spectral curve.

Recall that number of zigzags parallel to the side \(E\) is equal to the \(|E|_\mathbb{Z}\) (integral length of the segment \(E\)). Hence, the partition of set of zigzags used in Definition~\ref{def:dimer decorations} induces partition of \(|E|\) for any side \(E\). This motivates the following definition.

\begin{Definition}
	A  decorated Newton polygon is a pair \((N,\mathbf{H})\), where \(N\) is a convex integral polygon, and \(\mathbf{H}=(H_E\mid E \in \text{sides of } N)\) is a collection of partitions \(H_E=\{h_{E,i}\}\) of \(|E|_\mathbb{Z}\).
\end{Definition}


Recall that the reduction condition $J_{E,i}$ implied that the open curve \(\mathcal{C}\) has singularity of type \(x^{h_{E,i}}=y^{h_{E,i}}\). Resolution of all these singularities would decrease genus and we get 
\begin{equation}\label{eq:genus reduction}
	g(\overline{\mathcal{C}})=I-\sum\nolimits_{E,i}h_{E,i}(h_{E,i}-1)/2.
\end{equation}

\begin{Conjecture} \label{conj:generic}
	Under the certain conditions for decorated Newton polygon there exists integrable system such that 
	\begin{enumerate}[label=(\alph*)]
		\item \label{it:red} The integrable system is reduction of Goncharov-Kenyon integrable system corresponding to Newton polygon \(N\).

		\item The dimension of the phase space is equal to $\dim \mathcal{X}_{N,\mathbf{H}}-1$, where 
		\begin{equation} \label{eq:dim X_red}
			\dim \mathcal{X}_{N,\mathbf{H}}=2 \operatorname{Area}(N)-\sum\nolimits_{E,i} (h_{E,i}^2-1) 
		\end{equation}
		
		\item The rank of the Poisson bracket is equal to 
		\begin{equation} \label{eq:rank X_red}
			\operatorname{rk}\{\cdot,\cdot\}_{\mathcal{X}_{N,\mathbf{H}}} = 2I-\sum_{E,i}h_{E,i}(h_{E,i}-1).
		\end{equation}

		\item The phase space is a subvariety given by equation \(q=1\) in the \(\mathcal{X}\)-cluster variety \(\mathcal{X}_{N,\mathbf{H}}\), where \(q\) is certain Casimir function. 
	\end{enumerate}	
\end{Conjecture}

Similarly to the discussion after Theorem~\ref{th:Z,J} above, the property \ref{it:red} means that we can descend dimer partition function \(\mathcal{Z}_{N}(\lambda,\mu|\mathbf{x})\) to a function \(\mathcal{Z}_{N,\mathbf{H}}(\lambda,\mu|\mathbf{w})\). In latter serves as an equation of the spectral curve of reduced integrable system.

More precisely the reduction mentioned above is expected to include the following analogs of Lemma~\ref{lem:patch cluster chart}.

\begin{Conjecture}
    \label{conj:reductioncord}
	(a) There exists a sequence of cluster mutations \(\boldsymbol{\mu}\colon \mathsf{s} \mapsto \mathsf{t} \) such that the image of the ideal \(J\) is generated by \(\{m_y^{(I)}+1\mid1 \leq I \leq \sum_{E,i} (h_{E,i}-1)(h_{E,i}+2)/2\}\), where  \(\mathbf{y}\) denotes cluster variables in seed $\mathsf{t}$ and \(m_y^{(I)}\) are monomial in \(\mathbf{y}\).
	
	(b) One can define cluster coordinates \(\mathbf{w}\) on the reduction variety \(\mathcal{X}_{\mathrm{red}}\) of the chart \(\mathcal{X}_\mathsf{t}\) to be monomials in \(\mathbf{y}\). 
\end{Conjecture}

Moreover, there should exist different seeds \(\mathsf{t}\) which are convenient for reduction (i.e., the ideal is given by the binomials). These seeds will define different seeds \(\mathsf{s}_{\mathrm{red}}\), which should be related by mutations in the cluster variety \(\mathcal{X}_{\mathrm{red}}\).

\bigskip

Perhaps it is more accurate to study the symplectic leaves. We will call extended decoration an assignment of numbers  \(z_{E,i}\in \mathbb{C}^*\) for all parts \(h_{E,i}\). These numbers would be zigzag variables corresponding to zigzags \(\zeta_{E,i,j}\) (but do not depend on \(j\)). For any coefficients \(r_{E,i}\) such that \(\sum_{E,i} r_{E,i}E/|E|_{\mathbb{Z}}=0\) the number \(\prod \zeta_{E,i}^{r_{E,i}}\)
mean to be Casimir function. Let us denote the corresponding symplectic leaf by \(\mathcal{X}_{N,\mathbf{H},\mathbf{z}}\).  The formula \eqref{eq:rank X_red} gives dimension of \(\mathcal{X}_{N,\mathbf{H},\mathbf{z}}\) for generic values of \(\mathbf{z}\) subject of constraint \(\prod z_{E,i}^{h_{E,i}}=1\). In the analogy with space of local systems, the varieties \(\mathcal{X}_{N,\mathbf{H},\mathbf{z}}\) are analogs of the character varieties.

At this point we do not know precise form the conditions in the Conjecture \ref{conj:generic}. The obvious necessary condition is nonnegativity of the dimension of the phase space. Another quite natural condition is a boundness of the partition \(h_{E,i}\leq h_{E,N}\), this follows from the condition \(h \leq  l\) in the definition of the patch.

Note that the problem whether  \(\mathcal{X}_{N,\mathbf{H},\mathbf{z}}\) is empty can be viewed as an analogue of the Deligne-Simpson problem in case of semi-simple orbits. It is tempting to conjecture that analogously to the solution of the latter one \cite{Crawley:2004indecomposable} the answer could be given in terms of root system of a Kac-Moody Lie algebra. Namely to a decorated Newton polygon one have to assign Dynkin diagram and vector in the root lattice. 

\begin{Example} \label{Ex:decorated polygons} 
	(a) Assume that Newton polygon is a triangle with vertices \((0,0), (n,0), (0,n)\). Let partitions assigned to the sides of this polygons be \(h^{(i)}_1\geq h^{(i)}_1 \ge \dots h^{(i)}_{k_i}\), for \(i=1,2,3\). Then (conjecturally) the corresponding Dynkin diagram has a star shape, with one central vertex and three lines departing from it. The length of these lines are equal to \(k_1-1, k_2-1, k_3-1\). The toot markings corresponding to $i$-th line are given by \(h_1^{(i)},h^{(i)}_1+h^{(i)}_2,\dots, h^{(i)}_1+h^{(i)}_2+\dots,+h^{(i)}_{k_i}=n\) from end to the center. In particular the marking corresponding to the central vertex is equal to \(n\).
	
	In Fig.~\ref{fi:Example Triangle} we depicted example with \(n=6\) and partitions \((1^6), (2^3), (3^2)\). Note the obtained Dynkin diagram is \(E_8^{(1)}\) diagram. Moreover the markings are equal to markings of the imaginary root of \(E_8^{(1)}\). This example correspond to $q$-Painlev\'e equation with \(W^\mathrm{ae}(E_8)\) symmetry, see Section~\ref{sec:E8} and in particular Rem.~\ref{rem:E8 two polygons}.
	\begin{figure}[h]
	\centering
        \centering
        \includegraphics[scale=1]{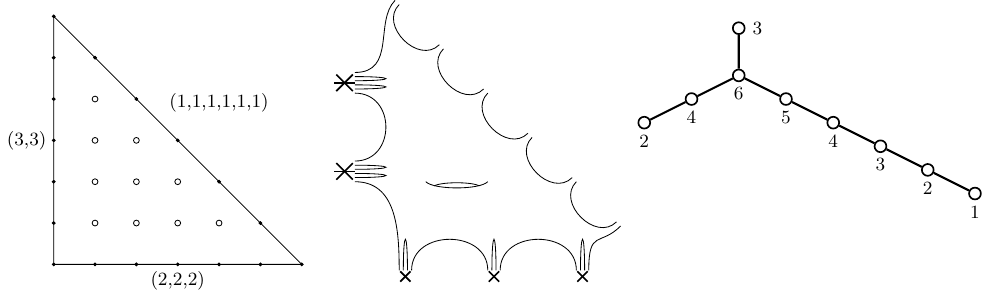}
		\caption{\label{fi:Example Triangle} Decorated Newton polygon, corresponding spectral curve with singular points at infinity, corresponding Dynkin diagram with root markings}
	\end{figure}

	(b) Assume that Newton polygon in rectangular with with vertices \((0,0), (n,0), (n,m), (0,m)\). Assume that partitions corresponding to vertical sides are \(h^{(1)},h^{(3)}\) and partitions corresponding to horizontal sides are \(h^{(2)}, h^{(4)}\). Then (conjecturally) the corresponding Dynkin diagram has an ``H'' shape. It means that we have two central nodes connected by one edge. There are two lines from the first one of the length \(l(h^{(1)})-1, l(h^{(3)})-1\) and two lines from the second one of the of the length \(l(h^{(2)})-1, l(h^{(4)})-1\). The root markings of the central vertices are equal to \(m\) and \(n\) and markings for nodes on the lines are defined as in previous case.
	
	In Fig.~\ref{fi:Example Rectangle} we depicted example with \(n=4,m=2\), and partitions \((2)\) for vertical sides and \((1^4)\) for horizontal sides. We obtain then \(E_7^{(1)}\) Dynkin diagram, and moreover, with the markings of the imaginary root of \(E_7^{(1)}\). This example corresponds to $q$-Painlev\'e equation with \(W^\mathrm{ae}(E_7)\) symmetry, see Section~\ref{sec:E7}. 
	\begin{figure}[h]
		\centering
        \includegraphics[scale=1]{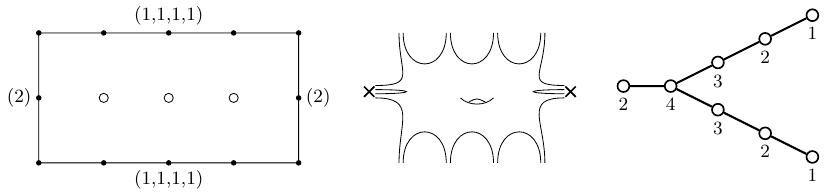}
		\caption{\label{fi:Example Rectangle} \label{Fig:decoratedE7} Decorated Newton polygon, the corresponding spectral curve with singular points at infinity and Dynkin diagram with root markings.}
	\end{figure}
\end{Example}
	An essential class of examples comes from the pointed Painlev\'e polygons. We discuss them in the next section.

\bigskip 

At the end of the section let us return to zigzag mutations. They correspond to the isomorphisms between reductions of the Goncharov-Kenyon integrable systems.

Consider consistent dimer model with decoration \((\Gamma, [\wt], \bigcup \boldsymbol{\zeta}_{E,i})\). For any collection \(\zeta_{E,i}\) using zigzag transpositions and Lemma~\ref{lem:zigzag patch} one can transform this dimer model to a one, where \(\zeta_{E,i}\) confine the \(\Pi_{l,h}\) patch with \(l=h_{E,N}\) and \(h=h_{E,i}\). Then we can perform zigzag mutation along \(\Pi_{l,h}\) as in Definition~\ref{def:zigzag mutations}.  The decoration is changed as follows: the collection \(\boldsymbol{\zeta}_{E,i}\) will be replaced by the collection of \(h_{E,N}-h_{E,i}\) zigzags going into the opposite direction. For any other collection \(\boldsymbol{\zeta}_{E',i'}\) its zigzags after mutation remain parallel and ordered, i.e. form new collection. 

Such zigzag mutation corresponds to mutation of decorated Newton polygons in the sense of definition below (the proof of Lemma~\ref{lem:zigzag mut - polyn mut} works without changes). 

\begin{Definition}
	Let \((N,\mathbf{H})\) be a decorated Newton polygon, \(E\) is a side of \(N\) and \(h\) belongs to the partition \(H_E\). Let \(\bar{H}\) denote partition corresponding to the antiparallel to \(E\) side of \(N\),  if it exists, and \(\bar{H}=\varnothing\) otherwise. Then mutation \((\widetilde{N},\widetilde{\mathbf{H}})=\mu_{E,h}(N,\mathbf{H})\) is defined as \(\widetilde{N}=\mu_{E,h}(N)\) with decoration
	\begin{itemize}		
		\item \(\tilde{H}_{\tilde{E}'}=H_{E'}\), where \(E'\) is a side of \(N\) neither parallel, nor antiparallel to \(E\), and \(\tilde{E}'\) is the corresponding side of \(\widetilde{N}\).
		\item Partition \(\widetilde{H}\), corresponding to the side parallel to \(E\) is equal to \(H_E\setminus h\).
		\item Partition \(\widetilde{H}\), corresponding to side antiparallel to \(E\) is equal to  \(\bar{H} \sqcup (h_{E,N}-h)\).
	\end{itemize}	
\end{Definition}
The definition of zigzag quiver (Definition~\ref{def:zigzag quiver}) for decorated polygons changes as follows:
\begin{Definition}
	Let \((N,\mathbf{H})\) be a decorated Newton polygon. \emph{Zigzag quiver} \(\mathcal{Q}^D\) is defined as 
	\begin{itemize}
		\item vertices correspond to the sides, namely for any side \(E\) there are \(l(H_{E})\) vertices;
		\item number of edges, between vertices corresponding to the sides \(E\) and \(E'\), is equal to \(\dfrac{\det(E,E')}{|E|_{\mathbb{Z}}|E'|_{\mathbb{Z}}}\).
	\end{itemize}
\end{Definition}
Equivalently, the number of edges is equal to the oriented area of parallelogram with sides given by two primitive vectors parallel to \(E\) and \(E'\). In terms of consistent dimer models with decoration the vertices of the quiver \(\mathcal{Q}^D\) are in one two one correspondence with collections of parallel zigzags \(\boldsymbol{\zeta}_{E,i}\) and number of edges between vertices corresponding to \(\boldsymbol{\zeta}_{E,i}\) and \(\boldsymbol{\zeta}_{E',i'}\) is equal to \(\det([\zeta_{E,i,1}],[\zeta_{E',i',1}])\).

\begin{Example}
	The zigzag quivers for the decorated Newton polygons studied in Example~\ref{Ex:decorated polygons} are depicted in Fig.~\ref{fi:quivers decoration}.
	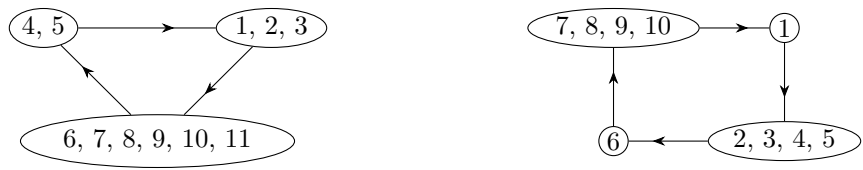
\begin{figure}[h]
		\centering
		\begin{tikzpicture}[font=\small]
			\def\xs{1.5}
			\def\ys{1.5}

			\begin{scope}
				
				\node[styleNode, ellipse, minimum height=0.7cm] (ybottom) at(0,0){6, 7, 8, 9, 10, 11};
				
				\node[styleNode, ellipse] (yright) at(\xs,\ys){1, 2, 3};
				
				\node[styleNode, ellipse] (yleft) at(-\xs,\ys){4, 5};
				
				\draw[styleArrow](yright) to[] (ybottom); 		
				\draw[styleArrow](ybottom) to[] (yleft); 		
				\draw[styleArrow](yleft) to[] (yright); 		
			\end{scope}
			\begin{scope}[shift={(4*\xs,0)}]
				\node[styleNode] (y1) at(0,0){6};
				
				\node[styleNode, ellipse] (y2) at(0,\ys){7, 8, 9, 10};
				
				\node[styleNode] (y3) at(1.5*\xs,\ys){1};

				\node[styleNode, ellipse] (y4) at(1.5*\xs,0){2, 3, 4, 5};
				
				\draw[styleArrow](y1) to (y2); 		
				\draw[styleArrow](y2) to (y3); 		
				\draw[styleArrow](y3) to (y4); 		
				\draw[styleArrow](y4) to (y1);

			\end{scope}
		\end{tikzpicture}
		\caption{\label{fi:quivers decoration} On the left quiver corresponding to the decorated Newton polygon on Fig.~\ref{fi:Example Triangle}, on the right the one for Fig.~\ref{fi:Example Rectangle}}
	\end{figure}
	
	Here and below several labels of vertices inside one ellipse means that that corresponding vertices have the same edges. In terms of dimer models such vertices correspond to parallel zigzags. 
	
	One can show (see Sections \ref{sec:E7}, \ref{sec:E8}) that Weyl groups of the  Dynkin diagrams depicted on Figs.~\ref{fi:Example Triangle} and \ref{fi:Example Rectangle} are embedded into cluster modular groups of these quivers.  
\end{Example}

In the definition of mutation of consistent dimer model we did not discuss mutation of the weights \(\wt\). It is defined only on the submanifolds \(\mathbf{V}(J)\) and given by Theorem~\ref{th: mutation ideal}. The following Theorem and Conjecture are generalizations of Proposition~\ref{prop:zigzag4 quivers} and Theorem~\ref{th:mut red patch}. 

\begin{Theorem}[\cite{Inprog}]
	Under the assumptions above, zigzag mutation of consistent dimer model with decoration gives mutation of cluster seed with adjacency matrix given by \(\mathcal{Q}^D\) and cluster variables are equal to zigzag variables.
\end{Theorem}

\begin{Conjecture}
	Under the assumptions above, zigzag mutation of consistent dimer model with decoration is an isomorphism of cluster varieties \(\mathcal{X}_{N,\mathbf{H}}\) and \(\mathcal{X}_{\widetilde{N},\widetilde{\mathbf{H}}}\). Furthermore, it induces isomorphism of reduced Goncharov-Kenyon integrable systems.
\end{Conjecture}

Similarly to Theorem~\ref{th: mutation ideal} the dimer partition function \(\mathcal{Z}_{N,\mathbf{H}}\) and \(\mathcal{Z}_{\widetilde{N},\widetilde{\mathbf{H}}}\) are expected to be related by polynomial mutation.

In particular, this conjecture implies that dimension and Poisson rank given by formulas~\eqref{eq:dim X_red} and \eqref{eq:rank X_red} correspondingly are preserved under zigzag mutations. It is easy to check this directly.

\subsection{$2n$-gon face mutations}
Recall that duality \(\Sigma \mapsto \Sigma^D\) swaps faces and zigzags. We discussed above zigzag mutations, in particular for zigzags of length greater than 4. It is natural to ask for dual transformation which would correspond to faces of length greater than 4.

The first nontrivial example is given on Fig.~\ref{fi:6-gon mutation}. 
\begin{figure}[h]
    \centering
    \includegraphics{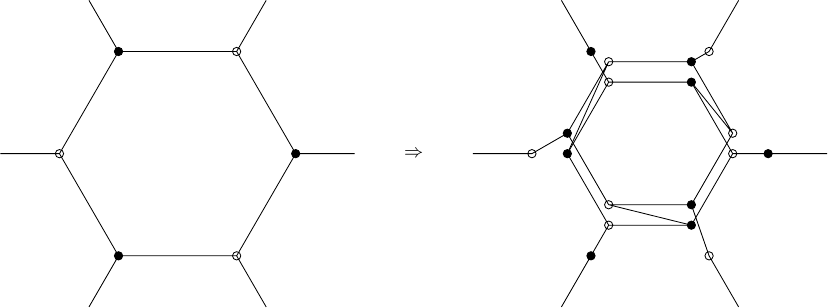}
    \caption{Mutation in 6-gon face}
    \label{fi:6-gon mutation}
\end{figure}
Here after mutation one 6-gon face was replaced by two new 6-gon faces. Moreover the genus of the surface was increased by 1. In other words in this surgery an additional handle is glued. The increasing of the genus by 1 in dual picture corresponds to increasing of the number of integral points inside the Newton polygon by 1 (which is fixed by the Hamiltonian reduction).

We hope to discuss this approach to face mutations elsewhere.

\section{Painlev\'e case} \label{sec:Painleve}

\subsection{\(q\)-difference Painlev\'e equations}

The \(q\)-difference Painleve equations were classified by Sakai \cite{Sakai:2001}. They are classified by their symmetry given by affine extended Weyl groups \(W^{\mathrm{ae}}(E_n)\). \footnote{For the symmetry type $E_1^{(1)}$ the symmetry group is slightly bigger and for $(E_1^{(1)})_{|\alpha^2|=8}$ and $E_2^{(1)}$ is slightly smaller.}
\begin{figure}[h]
	\begin{center}
			\begin{tikzcd}[row sep=scriptsize, column sep=scriptsize, font = \small]
				& & & & & & & (E_1^{(1)})_{|\alpha^2|=8} \arrow[rd] &\\
				{E_8^{(1)}} \arrow[r] & 
				{E_7^{(1)}}  \arrow[r] &	
				{E_6^{(1)}} \arrow[r] &
				{E_5^{(1)}}  \arrow[r]  &	
				{E_4^{(1)}} \arrow[r]  &
				{E_3^{(1)}}  \arrow[r]  &	
				{E_2^{(1)}} \arrow[r]  \arrow[ru] & {E_1^{(1)}}  &
				{E_0^{(1)}}
		\end{tikzcd}
	\end{center}
	\caption{$q$-Painlev\'e equations by symmetry type}
	\label{fi:qPeq}
\end{figure}
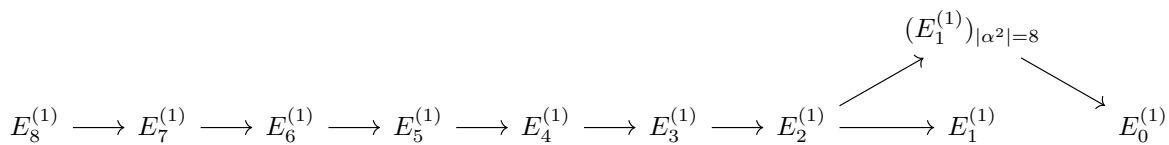

Here and below we use notations $E_5^{(1)}=D_5^{(1)}$, $E_4^{(1)}=A_4^{(1)}$, $E_3^{(1)}=(A_2+A_1)^{(1)}$, $E_2^{(1)}=(A_1+A_1)^{(1)}$, $E_1^{(1)}=A_1^{(1)}$. The word extended here means that external automorphisms acting on the affine Dynkin diagram of \(E_n^{(1)}\) are also added to the group \(W^{\mathrm{ae}}(E_n)\). Hence the group is generated by simple reflections \(s_i,\) \(i=0,\dots,n\) and also external automorphisms. Usually we denote the latter generators by \(\pi\) and write them as permutation of vertices of Dynkin diagram.

In the Sakai's geometric approach (see \cite{Sakai:2001}, \cite{KNY15} for many more details) to each \(q\)-Painlev\'e equation assigned a family of rational surfaces \(\mathcal{S}\). This family is parametrized by a collection of \emph{root variables} \(\mathbf{a}=(a_0,\dots,a_n)\subset (\mathbb{C}^*)^{n+1}\), i.e. for any \(\mathbf{a}\) we have a surface \(\mathcal{S}_{\mathbf{a}}\). The group \(W^{\mathrm{ae}}(E_n)\) acts multiplicatively on the root variables by formula 
\begin{equation}
\label{eq:actiononroot}
	s_i(a_j)=a_ja_i^{-C_{ij}},\; i=0,\dots,n,\qquad \pi(a_i)=a_{\pi(i)},
\end{equation}
where \(C\) denotes Cartan matrix of the root system \(E_8^{(1)}\). The multiplicative shift in difference equations \(q\) is a root variable corresponding to  imaginary root. The group \(W^{\mathrm{ae}}(E_n)\) acts on family \(\mathcal{S}\) such that \(w\colon \mathcal{S}_{\mathbf{a}}\rightarrow \mathcal{S}_{w(\mathbf{a})}\).

The group \(W^{\mathrm{ae}}(E_n)\) has a decomposition \(W^{\mathrm{ae}}(E_n)\simeq W(E_n) \ltimes P\), where \(W(E_n)\) is finite Weyl group and \(P\) is a weight lattice for \(E_n\). Any translation \(t \in P \subset W^{\mathrm{ae}}(E_n)\) introduces certain dynamics on the family \(\mathcal{S}\) which can be called Painlev\'e dynamics.

The case \(q=1\) is usually called autonomous. In this case translations  \( P \subset W^{\mathrm{ae}}(E_n)\) preserve the root variables \(\mathbf{a}\) (but act on \(\mathcal{S}_{\mathbf{a}}\) non-trivially). Moreover, there is a pencil of $P$-invariant elliptic curves on surfaces \(\mathcal{S}_{\mathbf{a}}\). Let \(\lambda_d,\mu_d\) denote some local coordinates on \(\mathcal{S}_{\mathbf{a}}\), then this pencil is determined by equation \(f_d(\mathbf{a}|\lambda_d,\mu_d)=c\). The function \(f_d(\mathbf{a}|\lambda_d,\mu_d)\) can be viewed as a Hamiltonian of the autonomous Painlev\'e equation.

\subsection{Painlev\'e pointed polygons}
\begin{Definition}
	Let \(N\) be an integral polygon such that \((0,0)\in N\). We call \(N\) to be \emph{Painlev\'e pointed polygons} if 
	\begin{enumerate}[label=(\alph*)]
		\item for all vertices \((a,b)\) of \(N\), \(\operatorname{gcd}(a,b)=1\) \label{it:a}
		\item for any side \(E\subset N\) the integral distance from \((0,0)\) to \(E\) devides \(|E|_{\mathbb{Z}}\). \label{it:b}
	\end{enumerate} 
\end{Definition}

Such polygons were introduces in \cite{Kasprzyk:2017minimality} under the name \emph{ Fano polygon without remainders}. It is easy to see that under condition \ref{it:b} the condition \ref{it:a} is equivalent to the fact that coordinates of all vertices of \(N\) are coprime. 

Let \(E\) be a side of \(N\) and let \(h_E\) be an integral distance from \((0,0)\) to \(E\).\footnote{Not to be confused with \(h_{E,N}\) which denotes the height ot the polygon \(N\). Clearly \(h_{E,N}>h_E\).} In this case we call mutation \(\mu_{E,h_E}\) by mutation with respect to origin. It is easy to see that set of Painlev\'e pointed polygons is closed under such mutations.

\begin{Theorem}[{\cite[Theorem 6]{Kasprzyk:2017minimality}}]  \label{th:KNP}
	Using mutations with respect to origin any Painlev\'e pointed polygons can be mutated to exactly one of the pointed polygons drawn on Fig.~\ref{fi:Pain polig}.
	\begin{figure}[h]
		\begin{center}
			\begin{tabular}{c c c c }
				$E_8$ & $E_7$ & $E_6$ &$E_5$
				\\		
				\begin{tikzpicture}[scale=0.7, font = \small]
					\draw[fill] (-4,2) circle (1pt) -- (-1,1) circle  (1pt) -- (2,0) circle (1pt) -- (5,-1) circle (1pt) -- (4,-1) circle (1pt) -- (3,-1) circle (1pt) -- (2,-1) circle (1pt) -- 
					(1,-1) circle (1pt) -- (0,-1) circle (1pt) --  (-1,-1) circle (1pt) -- (-2,-1) circle (1pt) 
					-- (-3,-1) circle (1pt) -- (-4,-1) circle (1pt)  --  (-4,0) circle (1pt) -- (-4,1) circle (1pt) -- (-4,2);
					\draw[fill] (-1,0) circle (2pt);
					\draw (-2,0) circle (2pt);
					\draw (-3,0) circle (2pt);
					\draw (-3,1) circle (2pt);
					
					\draw (0,0)  circle (2pt);
					\draw (1,0) circle (2pt);
					\draw (-2,1) circle (2pt);
					
				\end{tikzpicture}
				&
				\begin{tikzpicture}[scale=0.7, font = \small]
					
					\draw[fill] (-2,1) circle (1pt) -- (-2,0) circle  (1pt) -- (-2,-1) circle (1pt) -- (-1,-1) circle (1pt) -- (0,-1) circle (1pt) -- (1,-1) circle (1pt) -- (2,-1) circle (1pt) -- (2,0) circle (1pt) 
					-- (2,1) circle (1pt) -- (1,1) circle (1pt)  --  (0,1) circle (1pt) -- (-1,1) circle (1pt) -- (-2,1) circle (1pt);
					\draw[fill] (0,0) circle (2pt);
					\draw (1,0) circle (2pt);
					\draw (-1,0) circle (2pt);
					
				\end{tikzpicture}
				&
				\begin{tikzpicture}[scale=0.7, font = \small]
					
					\draw[fill] (-1,-1) circle (1pt) -- (0,-1) circle  (1pt) -- (1,-1) circle (1pt) -- (2,-1) circle (1pt) -- (1,0) circle (1pt) -- (0,1) circle (1pt) -- (-1,2) circle (1pt) -- (-1,1) circle (1pt) -- (-1,0) circle (1pt) -- (-1,-1) circle (1pt) ;
					\draw[fill] (0,0) circle (2pt);	
				\end{tikzpicture}
				&
				\begin{tikzpicture}[scale=0.7, font = \small]
					
					\draw[fill] (0,-1) circle (1pt) -- (1,-1) circle  (1pt) -- (1,0) circle (1pt) -- (1,1) circle (1pt) -- (0,1) circle (1pt)  -- (-1,1) circle (1pt) -- (-1,0) circle (1pt) -- (-1,-1) circle (1pt) -- (0,-1) circle (1pt) ;
					\draw[fill] (0,0) circle (2pt);	
				\end{tikzpicture}			
			\end{tabular} 
			\vspace{0.1cm}
			
			\begin{tabular}{c c c c c c}
				$E_4$ & $E_3$ & $E_2$ &$E_1$ &$E_1$ &$E_0$ \vspace{0.1cm}
				\\ 
				\begin{tikzpicture}[scale=0.7, font = \small]
					
					\draw[fill] (0,-1) circle (1pt) -- (1,-1) circle  (1pt) -- (1,0) circle (1pt) -- (1,1) circle (1pt) -- (0,1) circle (1pt)  -- (-1,1) circle (1pt) -- (-1,0) circle (1pt) -- (0,-1) circle (1pt) ;
					\draw[fill] (0,0) circle (2pt);	
				\end{tikzpicture}			
				&
				\begin{tikzpicture}[scale=0.7, font = \small]
					
					\draw[fill] (0,-1) circle (1pt) -- (1,-1) circle  (1pt) -- (1,0) circle (1pt) -- (0,1) circle (1pt)  -- (-1,1) circle (1pt) -- (-1,0) circle (1pt) -- (0,-1) circle (1pt) ;
					\draw[fill] (0,0) circle (2pt);	
				\end{tikzpicture}			
				&
				\begin{tikzpicture}[scale=0.7, font = \small]
		
					\draw[fill] (0,-1) circle (1pt) -- (1,-1) circle  (1pt) -- (1,0) circle (1pt) -- (0,1) circle (1pt)   -- (-1,0) circle (1pt) -- (0,-1) circle (1pt) ;
					\draw[fill] (0,0) circle (2pt);	
				\end{tikzpicture}			
				&
				\begin{tikzpicture}[scale=0.7, font = \small]
					
					\draw[fill] (0,-1) circle (1pt) -- (1,0) circle  (1pt) -- (0,1) circle (1pt)  --  (-1,0) circle (1pt) -- (0,-1) circle (1pt) ;
					\draw[fill] (0,0) circle (2pt);	
				\end{tikzpicture}
				&
				\begin{tikzpicture}[scale=0.7, font = \small]
					
					\draw[fill] (1,-1) circle (1pt) -- (1,0) circle  (1pt) -- (0,1) circle (1pt)  --  (-1,0) circle (1pt) -- (1,-1) circle (1pt) ;
					\draw[fill] (0,0) circle (2pt);	
				\end{tikzpicture}			
				&
				\begin{tikzpicture}[scale=0.7, font = \small]
					
					\draw[fill] (1,-1) circle (1pt) -- (0,1) circle  (1pt) -- (-1,0) circle (1pt)   -- (1,-1) circle (1pt) ;
					\draw[fill] (0,0) circle (2pt);	
				\end{tikzpicture}
			\end{tabular}
		\end{center}
		\caption{\label{fi:Pain polig} Representatives of the mutation equivalence classes of Painlev\'e pointed polygons}
	\end{figure}
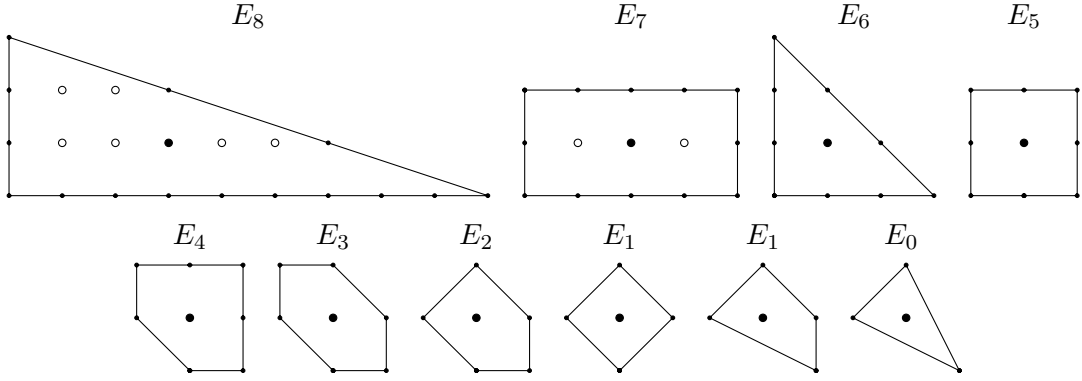
\end{Theorem}

It was observed in \cite{Mizuno} that there is a one to one correspondence between 10 polygons which appear in this classification and Painlev\'e equations in Fig.~\ref{fi:qPeq}. We will revisit this correspondance below. So in the Fig.~\ref{fi:Pain polig} we label polygons (or, better to say, mutation classes of polygons) by the symmetries of the corresponding equations. Here and below, in the figures of pointed polygons we will draw the origin by a filled circle.

To each Painlev\'e pointed polygon we assign a decoration in which the partition assigned to a side \(E\) is \((h_E^{d_E})\), where  \(d_E=|E|_\mathbb{Z}/h_E\). With the given decoration we will define cluster Poisson variety \(\mathcal{X}_{N,\mathbf{H}}\) and reduced Goncharov-Kenyon integrable system on it as in Section \ref{ssec:reduced GK}. Note, however, that we will not rely on results of Section \ref{sec:zigzag reductions} in this section. Indeed, for polygons \(E_0,\dots,E_6\) all \(h_E=1\) hence all decorations are trivial (of the form \((1^{|E|})\)) and no reductions are needed. These polygons have only one integral point inside and are \emph{reflexive}. In these cases \(\mathcal{X}_{N,\mathbf{H}}=\mathcal{X}_{N}\) are standard Goncharov-Kenyon varieites, see Sec.\ref{ssec:cluster}. The reductions for the \(E_7\) and \(E_8\) cases are performed in Sec. \ref{sec:E7} and \ref{sec:E8}.

%
%



By \(G_\mathcal{Q}\) we denote the correponding cluster modular group, where \(\mathcal{Q}\) denotes the cluster quiver (constructed from certain consistent bipartite graph in cases without reduction and in Sec.~\ref{sec:E7} and~\ref{sec:E8} for the other cases).

\begin{Theorem} \label{th:Painleve1} 
	Groups \(G_\mathcal{Q}\) for Painlev\'e pointed polygons give symmetries of \(q\)-difference Painlev\'e equations.
\end{Theorem}
\begin{proof}
	More explicitly the statement means
	\begin{enumerate}[label=(\roman*)]
		\item \label{it:Weyl}
		There is an embedding of the Painlev\'e symmetry group \(W^\mathrm{ae}(E_n)\) into cluster modular group~\(G_\mathcal{Q}\).

		\item \label{it:Casimirs}
		There is a choice of generators \(a_0,\dots,a_n\) of Poisson center of \(\mathcal{X}_{N,\mathbf{H}}\) such that action of \(W^\mathrm{ae}(E_n)\) on them is given by formula~\eqref{eq:actiononroot}.
		
		\item \label{it:surf}
		The symplectic leaves \(\mathcal{X}_{N,\mathbf{H},a}\) specified by the values of Casimirs are two-dimensional. The action of the group \(W^\mathrm{ae}(E_n)\) on family of these surfaces  coincides with action on the Painlev\'e--Sakai family \(\mathcal{S}\).
	\end{enumerate}
	This was checked  case by case in \cite{Bershtein:2018cluster} and also revisited in \cite{Mizuno}. We present details on properties \ref{it:Weyl} and \ref{it:Casimirs} in Sec.~\ref{sec:examples}, \ref{sec:E7}, \ref{sec:E8} below.
\end{proof}

\begin{Remark}
	As was established in \cite{Mizuno} geometrically Theorem~\ref{th:Painleve1} connects Sakai's blowup description of phase space of Painlev\'e equation with blowup description of cluster varieties in \cite{Gross:2013birational}.  
\end{Remark}

The property \ref{it:surf} in particular means that the rank of Poisson bracket \( \operatorname{rk} \{\cdot,\cdot\}_{\mathcal{X}_{N,\mathbf{H}}} =2\). In autonomous case \(q=1\) the corresponding integrable system should consist of one Hamiltonian. Therefore the corresponding spectral curve should have genus 1. This can be seen directly just from the combinatorial definition of Painlev\'e pointed polygons.

\begin{Proposition}
	The genus of spectral curve for Painlev\'e pointed polygon is equal to 1.
\end{Proposition}
\begin{proof}
	Mark \(d_E-1\) points on each side \(E\) that divide side into \(d_E\) segments of length \(h_E\). Connect vertices of \(N\) and new points on sides to the origin. We get decomposition of polygon \(N\) into triangles; for each side \(E\) we get \(d_E\) triangles which we denote by \(\Delta_{E,i}\), \(1\leq i \leq d_E\). 
	
	It follows from the definition that side of \(\Delta_{E,i}\) opposite to the origin has integral length \(h_E\) and integral distance from the origin to this side is also equal to \(h_E\). Hence  \(\operatorname{Area}(\Delta_{E,i})=h_E^2/2\).

	
	Using the formula \eqref{eq:genus reduction} and Pick's theorem we get
	\begin{multline}
		g(C)= \operatorname{Area}(N)-B/2+1-\sum_{E,i} h_E(h_E-1)/2
		\\
		=\sum_{E,i} \Big( \operatorname{Area}(\Delta_{E,i}) -  h_E/2-h_E(h_E-1)/2 \Big)+1=1.
	\end{multline}
\end{proof}

\subsection{Self-duality}

Note that while both papers \cite{Bershtein:2018cluster} and \cite{Mizuno} mentioned above, connect cluster mutations with \(q\)-Painlev\'e equations, the dimer model interpretation of these mutations is different. Namely, mutations in \cite{Bershtein:2018cluster} are basically face 4-gon mutations, while mutations in \cite{Mizuno} are polynomial mutations which are zigzag mutations as was explained in Sec.~\ref{ssec:zigzag mut 1}. The existence of two descriptions of the same \(q\)-Painlev\'e Weyl group is a manifestation of the following self-duality.

\begin{Theorem} \label{th:selfduality}
	(a) Face and zigzag quivers for Painlev\'e pointed polygons are mutation equivalent.

	(b) There exists function $f$ depending on spectral variables \(\lambda,\mu\), root variables \(\mathbf{a}\), and element \(\mathrm{w}\in W(E_n)\) such that the properly normalized partition function of the dimer model has the form 
	\begin{equation}\label{eq:selfduality}
		\mathcal{Z}\left(\mathbf{a}|\lambda_d,\mu_d;\lambda,\mu\right)=f(\mathbf{a}|\lambda,\mu)-f(\mathrm{w}(\mathbf{a})|\lambda_d,\mu_d).
	\end{equation}
\end{Theorem}
Recall that \(\lambda_d,\mu_d\) denotes coordinates on the Sakai's surfaces \(\mathcal{S}_a\), i.e. dynamical variables for the Painlev\'e equation. We will explain who they are.
\begin{proof}
	The zigzag quivers are easy to compute from the  polygons on Fig.\ref{fi:Pain polig}. The face quivers for cases without reduction were computed in \cite{Bershtein:2018cluster} (see also \cite{Hanany:2012brane}), we recall them in Sec.~\ref{sec:examples}. The cluster (reduced face) quivers for \(E_7\) and \(E_8\) cases require reductions and are computed in Sec. \ref{sec:E7} and \ref{sec:E8} below. In all these cases face reduced quivers coincide with zigzag quivers. 
	
	We checked formula \eqref{eq:selfduality} in Sec.~\ref{sec:examples}, \ref{sec:E7}, \ref{sec:E8}. 
	by computing the dimer partition functions for all Painlev\'e pointed polygons cases.
	
	Let us now present generic of arguments for this theorem. While they are not used in the proof above, and also contain the steps we can only show by case by case analysis, we hope that these arguments explain why the theorem is true. We restrict ourselves to the cases without reduction (i.e. exclude \(E_7\) and \(E_8\) cases).
	
	\textbf{Step 1.} Recall that \(\Gamma^D\subset \Sigma^D\) denotes dual dimer model. In the Painlev\'e cases \(g(\Sigma^D)=1\), hence \(\Gamma^D\) is again a dimer model on a torus. We claim that it is consistent dimer model. To see this one need to check the properties \ref{it:1},\ref{it:2},\ref{it:3} in the Definition \ref{def:consistancy}. For the  property \ref{it:1} note that for any face of \(\Gamma\) the corresponding variable is not a Casimir function for the Poisson bracket. Hence the corresponding zigzag on \(\Gamma^D\) is homologically nontrivial. The properties \ref{it:2} and \ref{it:3} are easy to see in cases \(E_3, E_4, E_5, E_6\). Indeed, in these cases there is a choice of \(\Gamma\) in which any two faces have no more than one common edge, hence in the dual graph \(\Gamma^D\) any two zigzags have no more than one common edge. That implies \ref{it:2} and \ref{it:3} for \(E_3, E_4, E_5, E_6\) cases. The check for \(E_0, E_1, E_1, E_2\) cases can be done directly.
		
	\textbf{Step 2.} There exists choice of \(\Gamma\) such that ribbon graphs \(\Gamma\)  and \(\Gamma^D\) are isomorphic. Indeed, graphs \(\Gamma\)  and \(\Gamma^D\) have the same number of vertices and edges, therefore they have the same number of faces. Hence the corresponding Newton polygons \(N, N^D\) are of the same area. For a given area of \(N\) among the reflexive polygons there exists only one with the minimal number of arrows in zigzag quiver. This quiver is a face quiver for \(\Gamma^D\). It is easy te see that there exist only one up to isomorphism (and recoloring of white and black vertices in \(E_3\) case) bipartite graph with this face quiver and number edges equal to this minimal number of arrows. Since number of arrows for \(\Gamma\) and \(\Gamma^D\) are equal we see that these graphs are isomorphic (as ribbon graphs).  Case \(\operatorname{Area}(N)=2\) requires more care since there are two mutation non-equivalent polygons with this area. The menitioned above \(E_3\) case also can be checked directly.
	
	 The face quivers for graphs \(\Gamma\) and \(\Gamma^D\) are isomorphic since the graphs are isomorphic. Therefore, the face and zigzag quiver for \(\Gamma\) are isomorphic.
	
	\textbf{Step 3.} Let us denote by \(\varphi\) an isomorphism between \(\Gamma\) and \(\Gamma^D\). For example, consider curves representing \(A\) and \(B\) cycles on \(\Gamma^D\) and denote weights of \(\varphi^{-1}(A)\) and \(\varphi^{-1}(B)\) by dual spectral variables \(\lambda_d, \mu_d\). Recall that we always choose representatives of \(A\) and \(B\) cycles to be (linear combinations of) zigzags, hence \(\lambda_d, \mu_d\) are monomials in face variables, i.e. local coordinates on \(\mathcal{X}_{N}\). It is also clear from the definition that they are Darboux coordinates on symplectic leafs on \(\mathcal{X}_{N}\).
		
	Recall that the root generators \(a_0,\dots, a_n\) are generators of the the Poisson center of \(\mathcal{X}_N\). The corresponding paths generate kernel of the map \(H_1(\Gamma)\mapsto H_1(\Sigma)\oplus H_1(\Sigma^D)\). Geometrically these paths can be represented either as linear combinations of faces or as a linear combinations of zigzags. Since the description is symmetric with respect to duality we see that \(\varphi(a_i)\) should be monomial in root variables \(\mathbf{a}\). We claim that this monomial transformation is given by formula \eqref{eq:actiononroot} for certain element \(\mathrm{w}\in W(E_n)\) such that \(\varphi(\mathbf{a})=\mathrm{w}(\mathbf{a})\). We prove this via case by case analysis. Note that translation part \(P \in W^{\mathrm{ae}}\) acts trivially on root variables, therefore it is sufficient to take \(\mathrm{w}\in W(E_n)\).
	
	\textbf{Step 4.} Let \( \mathcal{Z}\left(\mathbf{a}|\lambda_d,\mu_d;\lambda,\mu\right)\) denotes partition function of dimer model. The Newton polygon~\(N\) has only one integral point inside and we normalize \( \mathcal{Z}\) such that the Hamiltonian corresponds to constant term of \(\mathcal{Z}\). Recall formula \eqref{eq:Z = sum D} for the partition function. Let us choose Kasteleyn orientation such that $q_{K,D_0}(\alpha)=0$ for \(\alpha=(a,b)\in H_1(\Sigma)\) if and only if \(a,b\) are both even. Since polygon \(N\) is reflexive the only such point \((a,b)\) of \(N\) is the origin. Therefore in the formula~\eqref{eq:Z = sum D} dimer configurations contributing to the Hamiltonian appear with the plus sign while ones corresponding to the boundary appear with the minus sign. It is convenient for us then to change overall sign of \(\mathcal{Z}\).
	
	Let \(f=\mathcal{Z}-\mathcal{Z}_{0,0}\). Since all terms on the boundary are expressed through Casimirs the function \(f\) depends on spectral parameters \(\lambda,\mu\) and root variables \(\mathbf{a}\). The isomorphism \(\varphi\)	induces one to one correspondence between dimer configurations in \(\Gamma\) and \(\Gamma^D\). The contribution of terms corresponding to \(f\) in \(\Gamma^D\) gives \(f(\mathrm{w}(\mathbf{a}),\lambda_d,\mu_d)\). This function does not depend on spectral variables, hence it  contributes to the Hamiltonian. Therefore we have 
	\begin{equation}
		\mathcal{Z}\left(\mathbf{a}|\lambda_d,\mu_d;\lambda,\mu\right)=f(\mathbf{a}|\lambda,\mu)-f(\mathrm{w}(\mathbf{a})|\lambda_d,\mu_d)+\text{(Extra terms)}.
	\end{equation}
	Here \(\text{(Extra terms)}\) correspond to contributions which do not depend neither on spectral variables \(\lambda,\mu\) nor on dynamical variables \(\lambda_d,\mu_d\). We claim that there is no such terms and check this case by case.
\end{proof}

\begin{Remark}
	The formula~\eqref{eq:selfduality} represents duality between \(\lambda_d,\mu_d\) which are Darboux coordinates on symplectic leafs on \(\mathcal{X}\) and spectral parameters \(\lambda,\mu\). It can be explained in terms spectral transform \cite{Kenyon:2006planar}, \cite{George2022inverse}. 
	Recall that the there exists a rational map \(\mathcal{X}_\Gamma \dashrightarrow \{(\overline{\mathcal{C}}, D,\nu)\}\), where \(\overline{\mathcal{C}}\) is spectral curve, \(D\) is divisor on \(\overline{\mathcal{C}}\) of given degree (usually of degree \(g\)) and \(\nu\) is some discrete data (parametrization of points at infinity by zigzag paths). 
	In our case, for a given values of Casismir functions the point on \(\mathcal{X}\) is parametrized by \((\lambda_d,\mu_d)\). 
	On the other hand, since \(g(\overline{\mathcal{C}})=1\) the pairs of spectral curve with given values of Casimirs and point on it are parametrized by pairs \((\lambda,\mu)\). 
	Hence the spectral transform gives a birational map \(\{(\lambda_d,\mu_d)\} \dashrightarrow \{(\lambda,\mu) \}\).
\end{Remark}

\subsection{Partition function and elliptic Weyl group}

Self-duality established in previous section suggests that there are two affine Weyl group acting on the Painlev\'e dimer models: face mutations and zigzag mutations. These groups can be combined into one elliptic (double affine) Weyl group. Namely, we will show that this group acts on the space with coordinates \(\lambda,\mu,\lambda_d,\mu_d, \mathbf{a}\) preserving the partition function \(\mathcal{Z}\). 

Note that for dimer models we have several results that allow to show invariance of the partition function, see Propositions~\ref{prop:4gon mutation}, \ref{prop:zigzag transposition}, \ref{prop:zigzag4} above. But technically we will not use them in this section for two reasons. One of the reasons is that they are not sufficiently strong to find the whole group, in particular (but not only) in cases with reduction. Another reason is that in Painlev\'e cases there exists more elementary approach bases on the specifics of the situation. This can be seen in the following lemma.

\begin{Lemma} \label{lem:Painleve by boundary}
	 Let \(N\) be a Painlev\'e pointed polygon and \(f_{N,\text{red}}(\lambda,\mu)\) be a Laurent polynomial with Newton Polygon \(N\) which satisfies reduction conditions and has vanishing constant term. Then the polynomial \(f_{N,\text{red}}(\lambda,\mu)\) is uniquely determined up to multiplicative constant by roots of restrictions \(f_{N,\text{red}}(\lambda,\mu)|_E\) on the sides \(E\) of \(N\).
\end{Lemma}
\begin{proof}
	This property (polynomial is essentially determined by its value on the boundary of Newton polygon) is preserved under mutation. So it is sufficient to check it for the polygons depicted on Fig.~\ref{fi:Pain polig}. Let us denote coefficient of the polynomial \(f_{N,\text{red}}(\lambda,\mu)\) as \(f_{a,b}\) via 
	\begin{equation}
		f_{N,\text{red}}(\lambda,\mu)=\sum_{(a,b)\in N}f_{a,b}\lambda^a\mu^b.
	\end{equation}
	
	We already assumed that \(f_{0,0}=0\). Values of the \(f_{a,b}\) on the boundary are determined by the roots of restrictions on sides up to overall multiplicative constant. 
	This proves lemma for the cases \(E_0,\dots,E_6\) (cases without reduction) since the corresponding Newton polygons have only one internal point \((0,0)\). 
	
	Let us consider \(E_7\). Due to reduction conditions, the restriction of \(f_{N,\text{red}}(\lambda,\mu)\) on the right side should have root of multiplicity 2, i.e. be proportional to  \((1+c\mu)^2\) for some \(c\). Moreover, due to reduction condition \eqref{eq:polyn mut cond} we have that \((1+c\mu) \) divides \( \big(f_{1,-1}\mu^{-1}+f_{1,0}+f_{1,1} \mu\big)\). This determines \(f_{1,0}\), since the coefficients on the boundary of \(N\) are determined. Similarly, one can show that coefficient \(f_{1,0}\) is determined.
	
	Let us consider \(E_8\). Due to reduction conditions, the restriction of \(f_{N,\text{red}}(\lambda,\mu)\) on the vertical  side should have root of multiplicity 3, i.e. be proportional to  \((1+c\mu^{-1})^3\). Moreover, due to reduction condition \eqref{eq:polyn mut cond} we have that \((1+c\mu^{-1})^2\) divides \( \big(f_{-2,-1}\mu^{-1}+f_{-2,0}+f_{-2,1} \mu\big)\). This determines \(f_{-2,0}\) and \(f_{-2,1}\). Similarly, using reduction along slanted side, one can show that coefficients \(f_{2,0}\) and \(f_{-1,1}\) are determined. Using condition for the vertical side again we have that \((1+c\mu^{-1})^1 \) divides \( \big(f_{-1,-1}\mu^{-1}+f_{-1,0}+f_{-1,1} \mu\big)\). This determines \(f_{-1,0}\), since \(f_{-1,-1}\) corresponds to the boundary of \(N\) and \(f_{-1,1}\) was determined above. Similarly, one can show that coefficient \(f_{1,0}\) is determined.	
\end{proof}

For example in Fig.~\ref{fi:example polygon and polynomial} we present the form of this polynomial in the case of \(E_5\) polygon. We give more examples below, see e.g. formulas~\eqref{eq:E7 Hamiltonian} for \(E_7\) case and \eqref{eq:E8 Hamiltonian} for \(E_8\) case.
\begin{figure}[h]
	\centering
	\begin{tikzpicture}[scale=0.7, font = \small]
			\def\xs{2}
			\def\ys{2}

			\draw[fill] (0,-\ys) circle (1pt) -- (\xs,-\ys) circle  (1pt) -- (\xs,0) circle (1pt) -- (\xs,\ys) circle (1pt) -- (0,\ys) circle (1pt)  -- (-\xs,\ys) circle (1pt) -- (-\xs,0) circle (1pt) -- (-\xs,-\ys) circle (1pt) -- (0,-\ys) circle (1pt) ;
			\draw[fill] (0,0) circle (2pt);

			\node[left] at (-\xs,0.5*\ys) {$(1+z_1)$};
			\node[left] at (-\xs,-0.5*\ys) {$(1+z_2)$};
			\node[below] at (-0.5*\xs,-\ys) {$(1+z_3)$};
			\node[below] at (0.5*\xs,-\ys) {$(1+z_4)$};
			\node[right] at (\xs,-0.5*\ys) {$(1+z_5)$};
			\node[right] at (\xs,0.5*\ys) {$(1+z_6)$};
			\node[above] at (-0.5*\xs,\ys) {$(1+z_8)$};			
			\node[above] at (0.5*\xs,\ys) {$(1+z_7)$};
	
		\begin{scope}[shift={(6+2*\xs,0)}]
			\node at (0,0) {\parbox{8cm}{
				\begin{align*}
					 f_{N,\text{red}}(\mathbf{z})&= (z_3z_4)^{1/4}(z_5z_6)^{1/2}(z_7z_8)^{3/4}
					 \\&\Big( ((1+z_1)(1+z_2)-1) 
					\\ &+((1+z_3)(1+z_4)-1)\prod\nolimits_{i=1}^2z_i 
					\\ &+((1+z_5)(1+z_6)-1)\prod\nolimits_{i=1}^4z_i
					\\ &+((1+z_7)(1+z_8)-1)\prod\nolimits_{i=1}^6z_i \Big).
				\end{align*}			
			}};
		\end{scope}
	\end{tikzpicture}
	\caption{$E_5$ polygon; on the left roots of the restrictions on the boundary, on the right  the corresponding polynomial \(f_{N,\text{red}}(\mathbf{z})\). 
		 \label{fi:example polygon and polynomial} }	
\end{figure}

Denote by \(\mathbf{z}\) the roots of its restriction \(f_{N,\text{red}}(\lambda,\mu)|_E\) on the sides \(E\) of \(N\). C.f. notations for zigzag variables and formula~\eqref{eq:Z|E prod} above. Due to reduction conditions assumed on \(f_{N,\text{red}}\) we have \(f(\lambda,\mu)|_E \sim \prod_{j=1}^{d_E} (1+z_{i_j})^{h_E}\) where \(z_{i_1},\dots z_{i_{d_E}}\) are variables corresponding to \(E\). For variable~\(z\) corresponding to the side \(E\) define \(h(z)=h_{E}\) and define degree by \(\deg(z)=E/|E|\in \mathbb{Z}^2\) to be a primitive vector parallel to \(E\) (c.f. \(z \mapsto [\zeta]\) for dimer models). We have \(\prod_i z_i^{h(z_i)}=1\) and \(\sum_i h(z_i) \wt(z_i)=0\) where product and sum runs over all variables.

By Lemma~\ref{lem:Painleve by boundary} the polynomial \(f_{N,\text{red}}(\lambda,\mu)|_E\) is determined by \(\mathbf{z}\) up to normalization which is monomial in \(f_{N,\text{red}}(\lambda,\mu)|_E\). We will write this polynomial by \(f_{N,\text{red}}(\mathbf{z})\). We know (see Theorem~\ref{th:Painleve1}) that affine Weyl group 
\(W^{\mathrm{ae}}(E_n)\) can be realized via permuations and polynomial mutations (note that polynomial mutations are cluster mutations for seed with reduced zigzag quiver \(\mathcal{Q}^D\) and variables \(\mathbf{z}\)). If element \(u\in W^{\mathrm{ae}}(E_n)\) transform variables \(\mathbf{z}\mapsto \mathbf{\tilde{z}}\) then the transformation of the polynomial \(f_{N,\text{red}}(\mathbf{z})\) would be proportional to \(f_{N,\text{red}}(\mathbf{\tilde{z}})\) by Lemma~\ref{lem:Painleve by boundary}. We claim that in appropriate normalization the transformed \(f\) is equal to \(f\) on transformed variables, not just proportional.
\begin{Theorem}\label{th:partition invariance}
	There exists normalization of \(f_{N,\text{red}}(\lambda,\mu)\) which is preserved by the action of affine Weyl group \(W^{\mathrm{ae}}(E_n)\). 
\end{Theorem}
\begin{proof}
	We present \(f\) and cluster realization in all cases below in Sec.~\ref{sec:examples}, \ref{sec:E7}, \ref{sec:E8}.  The invariance of \(f\) can be checked in all cases directly.
\end{proof}

Let us now give uniform proof of Theorem~\ref{th:partition invariance} that do not require case by case analysis. In fact we prove slightly stronger statement. We will need the following  geometric notion.

\begin{Definition}
	Let \(N\) be a convex polygon with vertices \(A_1,\dots, A_m\) in counterclockwise order. For any point \(O\) let us define \emph{rational barycentric coordinates} to be 
	\begin{equation}\label{eq:barycentric}
		\phi_i(O)= \frac{2\operatorname{Area}(A_{i-1}A_{i}A_{i+1})}{2\operatorname{Area}(A_{i-1}A_{i}O)  \cdot 2\operatorname{Area}(OA_{i}A_{i+1})}.
	\end{equation}
\end{Definition}
Here we assumed periodicity in indices, i.e. \(A_0=A_m\) and \(A_{m+1}=A_1\). These coordinates were defined in \cite{Wachspress:1975} (and are usually called Wachspress coordinates), see e.g.  \cite{Floater:2014wachspress} for the concise introduction. These coordinates are usually normalized by \(\sum \phi_i(O)=1\), but normalization~\eqref{eq:barycentric} will be more convenient for us. Since we need this coordinates only for the origin we will suppress the argument \(O\).

The adjective \emph{barycentric} reflects the property 
\begin{equation}
	\sum\nolimits_{i=1}^m \phi_i \overline{O A_i}=0.
\end{equation} 
The adjective \emph{rational} reflects the property that \(\phi_i(O)\) are given by rational functions on coordinates of \(A_1, \dots, A_m\).

Let \(f_{A_i}\) denotes term of \(f\) corresponding to vertex \(A_i\). Since \(N\) is a Newton polygon of \(f(\lambda,\mu)\) we have \(\deg(f_{A_i})=\overline{OA_i}\).
\begin{Definition}\label{def:proper normalization}
	We will call polynomial \(f_{N,\text{red}}(\mathbf{z})\) to be \emph{properly normalized} if each \(f_{A_i}\) is a monomial in (fractional powers of) \(\mathbf{z}\) with coefficient 1 and \(\prod_{i=1}^m f_{A_i}^{\phi_i}=1\).	
\end{Definition}
This conditions determine \(f_{N,\text{red}}(\mathbf{z})\) uniquely. For example, the polynomial given in Fig.\ref{fi:example polygon and polynomial} is properly normalized.

\begin{Theorem}\label{th:partition normalization}
	Polynomial mutations and permutations of \(\mathbf{z}\) transform properly normalized polynomial \(f_{N,\text{red}}(\mathbf{z})\) to properly normalized polynomial \(f_{\widetilde{N},\text{red}}(\mathbf{\tilde{z}})\).
\end{Theorem}
\begin{Corollary}
\label{cor:polyinvatiance}
	Cluster modular group \(G_{\mathcal{Q}}\) preserves properly normalized polynomial \(f_{N,\text{red}}(\mathbf{z})\).
\end{Corollary}
Note that this corollary is slightly stronger then Theorem~\ref{th:partition invariance} since we have not proved (while expect to be true) that embedding \(W^{\mathrm{ae}}(E_n) \to G_{\mathcal{Q}}\) is an isomorphism.

\begin{Example}
	In the Figs.~\ref{fi:example properly normalized 1}, \ref{fi:example properly normalized 2}, \ref{fi:example properly normalized 3} we show three properly normalized polynomials. The transformation from Fig.~\ref{fi:example properly normalized 1} to Fig.~\ref{fi:example properly normalized 2} is given by mutation in vertex corresponding to \(z_6\). It is given by the following change of variables 
	\begin{equation}
		z_6\mapsto z_6^{-1},\quad z_1\mapsto z_1(1+z_6),\, z_2\mapsto z_2(1+z_6),\quad z_4\mapsto z_4(1+z_6^{-1})^{-1},\, z_5\mapsto z_5(1+z_6^{-1})^{-1}.
	\end{equation}
	\begin{figure}[h]
		\centering 
		\begin{tikzpicture}[scale=0.7, font = \small]
			\def\xs{2}
			\def\ys{2}

			\begin{scope}
				
				\draw[fill] (0,-\ys) circle (1pt) -- (\xs,-\ys) circle (1pt) -- (\xs,0) circle (1pt) -- (0,\ys) circle (1pt)  -- (-\xs,\ys) circle (1pt) -- (-\xs,0) circle (1pt) -- (0,-\ys) circle (1pt) ;
				\draw[fill] (0,0) circle (2pt);	
				
				\node[left] at (-\xs,0.5*\ys) {$(1+z_1)$};
				\node[below left] at (-0.5*\xs,-0.5*\ys) {$(1+z_2)$};
				\node[below] at (0.5*\xs,-\ys) {$(1+z_3)$};
				\node[right] at (\xs,-0.5*\ys) {$(1+z_4)$};
				\node[above right] at (0.5*\xs,0.5*\ys) {$(1+z_5)$};
				\node[above] at (-0.5*\xs,\ys) {$(1+z_6)$};			
								
			\end{scope}

			\begin{scope}[shift={(3+2*\xs,0)}]
				
				\draw[fill] (0,-\ys) circle (1pt) -- (\xs,-\ys) circle (1pt) -- (\xs,0) circle (1pt) -- (0,\ys) circle (1pt)  -- (-\xs,\ys) circle (1pt) -- (-\xs,0) circle (1pt) -- (0,-\ys) circle (1pt) ;
				\draw[fill] (0,0) circle (2pt);	
				
				\node[left] at (-\xs,0) {$1$};
				\node[left] at (0,-\ys) {$1$};
				\node[right] at (\xs,-\ys) {$1$};
				\node[right] at (\xs,0) {$1$};
				\node[right] at (0,\ys) {$1$};
				\node[left] at (-\xs,\ys) {$1$};

			\end{scope}
			\begin{scope}[shift={(6.5+4*\xs,0)}]
				\node at (0,0) {\parbox{8cm}
					{
						\begin{align*}
							f_{N,\text{red}}(\mathbf{z})&= z_2^{1/6}z_3^{1/3}z_4^{1/2}z_5^{2/3}z_6^{5/6}
							\\&\Big( 1+z_1(1+z_2(1
							\\ & +z_3(1+z_4(1+z_5)))) \Big).
						\end{align*}			
					}};
						
			\end{scope}	
		\end{tikzpicture}
		\caption{From left to right: polygon \(N\) with roots if the restriction on the boundary; values of rational barycentric coordinates; properly normalized partition function.\label{fi:example properly normalized 1}}
	\end{figure}
	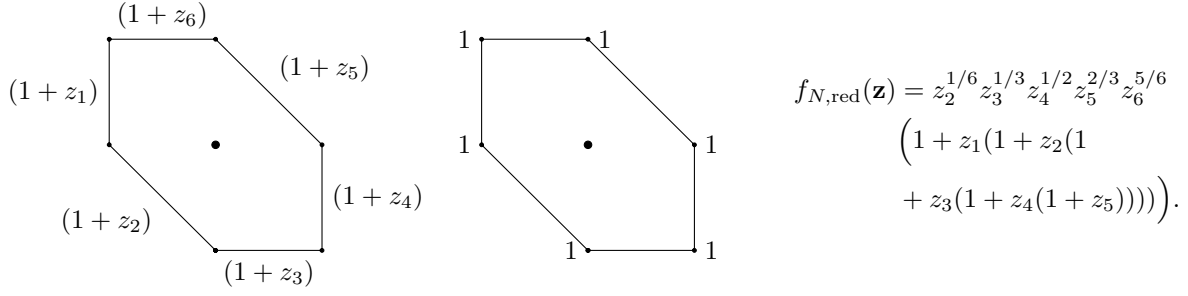
\begin{figure}[h]
	\centering 
		\begin{tikzpicture}[scale=0.7, font = \small]
			\def\xs{2}
			\def\ys{2}
			
			\begin{scope}
				
				\draw[fill] (0,-\ys) circle (1pt) -- (\xs,-\ys) circle (1pt) -- (\xs,0) circle (1pt) -- (0,\ys) circle (1pt)  -- (-\xs,0) circle (1pt) -- (-\xs,-\ys) circle (1pt) -- (0,-\ys) circle (1pt) ;
				\draw[fill] (0,0) circle (2pt);	
				
				\node[above left ] at (-0.5*\xs,0.5*\ys) {$(1+z_1)$};
				\node[left] at (-1*\xs,-0.5*\ys) {$(1+z_2)$};
				\node[below] at (0.5*\xs,-\ys) {$(1+z_3)$};
				\node[right] at (\xs,-0.5*\ys) {$(1+z_4)$};
				\node[above right] at (0.5*\xs,0.5*\ys) {$(1+z_5)$};
				\node[below] at (-0.5*\xs,-\ys) {$(1+z_6)$};
				
			\end{scope}

			\begin{scope}[shift={(3+2*\xs,0)}]
				
				\draw[fill] (0,-\ys) circle (1pt) -- (\xs,-\ys) circle  (1pt) -- (\xs,0) circle (1pt) -- (0,\ys) circle (1pt)  -- (-\xs,0) circle (1pt) -- (-\xs,-\ys) circle (1pt) -- (0,-\ys) circle (1pt) ;
		  		\draw[fill] (0,0) circle (2pt);					
			
				\node[left] at (-\xs,0) {$1$};
				\node[left] at (-\xs,-\ys) {$1$};
				\node[right] at (\xs,-\ys) {$1$};
				\node[right] at (\xs,0) {$1$};
				\node[right] at (0,\ys) {$2$};

			\end{scope}
			\begin{scope}[shift={(6.5+4*\xs,0)}]
				\node at (0,0) {\parbox{8cm}
					{
						\begin{align*}
							f_{N,\text{red}}(\mathbf{z})&= z_2^{1/6}z_3^{1/3}z_6^{1/3}z_4^{1/2}z_5^{2/3}
							\\&\Big( 1+z_1(1+z_2(1
							\\ & +z_3z_6(1+z_4))) \Big).
						\end{align*}			
				}};
				
			\end{scope}	
		\end{tikzpicture}
		\caption{From left to right: polygon \(N\) with roots if the restriction on the boundary; values of rational barycentric coordinates; properly normalized partition function.\label{fi:example properly normalized 2}}
	\end{figure}
	
	The transformation from Fig.~\ref{fi:example properly normalized 2} to Fig.~\ref{fi:example properly normalized 3} is given by mutation in vertex corresponding to \(z_2\). It is given by the following change of variables 
	\begin{equation}
		z_2\mapsto z_2^{-1},\quad z_3\mapsto z_3(1+z_2),\, z_6\mapsto z_6(1+z_2),\quad z_1\mapsto z_1(1+z_2^{-1})^{-1},\, z_5\mapsto z_5(1+z_2^{-1})^{-1}.
	\end{equation}
	\begin{figure}[h]
		\centering 
		\begin{tikzpicture}[scale=0.7, font = \small]
			\def\xs{2}
			\def\ys{2}
			
			\begin{scope}
				
				\draw[fill] (0,-\ys) circle (1pt) -- (\xs,-\ys) circle (1pt) -- (\xs,0) circle (1pt) -- (\xs,\ys) circle (1pt)  -- (0,\ys) circle (1pt) -- (-\xs,-\ys) circle (1pt) -- (0,-\ys) circle (1pt) ;
				\draw[fill] (0,0) circle (2pt);	
				
				\node[above left ] at (-0.5*\xs,0) {$(1+z_1)$};
				\node[below] at (-0.5*\xs,-\ys) {$(1+z_6)$};
				\node[below] at (0.5*\xs,-\ys) {$(1+z_3)$};
				\node[right] at (\xs,-0.5*\ys) {$(1+z_4)$};
				\node[right] at (\xs,0.5*\ys) {$(1+z_2)$};
				\node[above] at (0.5*\xs,1*\ys) {$(1+z_5)$};
				
			\end{scope}

			\begin{scope}[shift={(3+2*\xs,0)}]
				
				\draw[fill] (0,-\ys) circle (1pt) -- (\xs,-\ys) circle (1pt) -- (\xs,0) circle (1pt) -- (\xs,\ys) circle (1pt)  -- (0,\ys) circle (1pt) -- (-\xs,-\ys) circle (1pt) -- (0,-\ys) circle (1pt) ;
				\draw[fill] (0,0) circle (2pt);	
				
				\node[left] at (-\xs,-\ys) {$2$};
				\node[right] at (\xs,-\ys) {$1$};
				\node[right] at (\xs,\ys) {$1$};
				\node[left] at (0,\ys) {$2$};

			\end{scope}
			\begin{scope}[shift={(6.5+4*\xs,0)}]
				\node at (0,0) {\parbox{8cm}
					{
						\begin{align*}
							f_{N,\text{red}}(\mathbf{z})&=  z_3^{1/3}z_6^{1/3}z_2^{1/2}z_4^{1/2}z_5^{2/3}
							\\&\Big( 1+z_1(1+z_3+z_6
							\\ & +z_3z_6(1+z_2)(1+z_4)) \Big).
						\end{align*}			
				}};
				
			\end{scope}	
		\end{tikzpicture}
		\caption{From left to right: polygon \(N\) with roots if the restriction on the boundary; values of rational barycentric coordinates; properly normalized partition function. \label{fi:example properly normalized 3}}
	\end{figure}
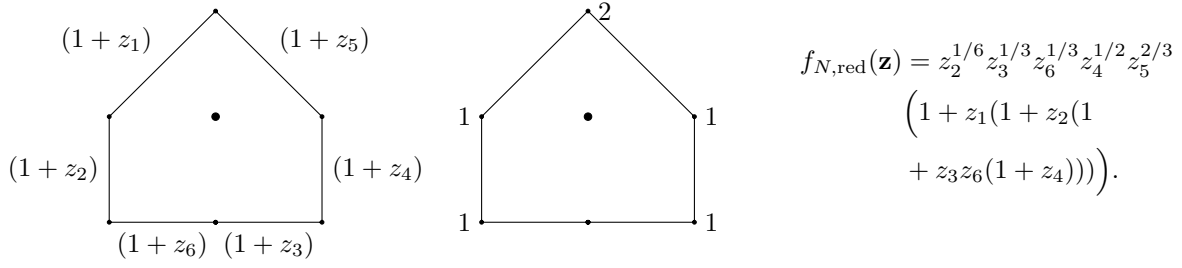
\end{Example}

\begin{proof}[Proof of Theorem~\ref{th:partition normalization}]
	If \(z_1,z_2\) correspond to the same side \(E\) of \(N\), then properly normalized polynomial is symmetric with respect to permutation of \(z_1\) and \(z_2\). Clearly, the nontrivial part of the theorem concerns mutations.
	
	Recall the combinatorial definition of mutation of polygon given in Proposition~\ref{prop:mutatation edges}. The formula \eqref{eq:mut segments} means that some edges remain unchanged and some are transformed by the element conjugated to \(\begin{pmatrix}
		1 & 1 \\ 0 & 1
	\end{pmatrix} \in SL(2,\mathbb{Z})\).
	
	Assume that we perform polynomial mutation for the variable \(z_1\) which corresponds to the side \(A_1A_2\) of \(N\). Assume that there is a side \(A_kA_{k+1}\) which is antiparallel to \(A_1A_2\). Assume that there are other variables corresponding to the side \(A_1A_2\), i.e. after mutation this side do not disappear. 
	
	Under these assumptions there is a natural one to one correspondence between vertices of the polygon \(N\) and vertices of the mutated polygon \(\widetilde{N}\). More precisely, one can assume that vertices \(\tilde{A}_{2},\dots, \tilde{A}_{k}\) of polygon \(\widetilde{N}\) coincide with corresponding vertices of polygon \(N\), while the remaining vertices \(\tilde{A}_{1}, \tilde{A}_{k+1},\dots, \tilde{A}_{n}\) are obtained from the corresponding vertices of \(N\) by the action of matrix conjugated to \(\begin{pmatrix}
		1 & 1 \\ 0 & 1
	\end{pmatrix} \in SL(2,\mathbb{Z})\) shifting parallel to \(A_1A_2\). The formula \eqref{eq:barycentric} for the rational barycentric coordinates is invariant under \(SL(2,\mathbb{Z})\) action. Hence the coordinates for all vertices except \(A_1,A_2,A_k,A_{k+1}\) are preserved under this polygon mutation \(\tilde{\phi}_i=\phi_i\). Moreover, it is easy to check that coordinates for that four vertices are also preserved. For example, for \(A_2\) one can notice that 
	\begin{equation}
		\frac{\operatorname{Area}(A_1A_{2}A_3)}{ \operatorname{Area}(A_{1}A_2O)}
		=\frac{\det(\overline{A_2A_{1}},\overline{A_2A_3})}{\det(\overline{A_2A_1},\overline{A_2O})}
		=\frac{\det(\overline{A_2\tilde{A}_1},\overline{A_2A_{3}})}{\det(\overline{A_2\tilde{A}_1},\overline{A_2O})}
		= \frac{\operatorname{Area}(\tilde{A}_1A_{2}A_3)}{ \operatorname{Area}(\tilde{A}_{1}A_2O)}.
	\end{equation}
	since \(\overline{A_2A_1}\) and \(\overline{A_2\tilde{A}_1}\) are parallel.

	Note that in this mutation the weights of all variables \(\mathbf{z}\) are multiplied by some functions of \(z_1\). Therefore, the coefficients corresponding to vertices are multiplied by the powers of \(z_1\), namely \(\tilde{f}_{\tilde{A}_i}=f_{A_i}z_1^{r_i}\) for some  \(r_i\in \mathbb{Q}\). Therefore 
	\begin{equation}
		\prod_{i=1}^m \tilde{f}_{\tilde{A}_i}^{\phi_i}=\prod_{i=1}^m f_{A_i}^{\phi_i} z_1^{\phi_i r_i}= z_1^{\sum \phi_i r_i}.	
	\end{equation}
	On the other hand, due to definition of  barycentric coordinates we have \(\deg (\prod_{i=1}^m \tilde{f}_{\tilde{A}_i}^{\phi_i})=1 \). Hence \(\sum \phi_i r_i=0\).
	
	It remains to consider case when side \(A_1A_2\) disappears after mutations or there is no side in \(N\) antiparallel to \(A_1A_2\). This situations are inverse to each other, so it is sufficient to consider one. Assume that \(h=h_{A_1A_2}=|A_1A_2|\). Then after polygon mutation the side \(A_1A_2\) disappears, i.e. gets replaced by one vertex, which we denote by \(\tilde{A}_1\). 
	\begin{Lemma}
		Under assumptions above we have \(\tilde{\phi}_{\tilde{A}_1}=\phi_{A_1}+\phi_{A_2}\).
	\end{Lemma}
	\begin{proof}
		Without loose of generality one can assume that side \(A_1A_2\) is horizontal and point \(A_1\) has coordinates  \((0,-h)\). The coordinates of other points involved in computation are 
		\begin{equation}
			A_0=(-a,-b),\;\; A_2=\tilde{A}_1=(h,-h),\;\; A_3=\tilde{A}_2=(c,-d),\;\; \tilde{A}_0=(-a+b,-b).
		\end{equation}
		\begin{figure}[h]
			\centering
			\begin{tikzpicture}[font = \small]
				\def\xs{1}
				\def\ys{1}
				\begin{scope}
			        \node[circle, fill, inner sep=1pt, label=above:{$O$}] (O) at (0,0) {};
			        \node[circle, fill, inner sep=1pt, label=left:{$A_0$}] (A1) at (-1.5*\xs,-0.5*\ys) {};
			        \node[circle, fill, inner sep=1pt, label=below:{$A_1$}] (A2) at (0,-2*\ys) {};
			        \node[circle, fill, inner sep=1pt, label=below:{$A_2$}] (A3) at (2*\xs,-2*\ys) {};
			        \node[circle, fill, inner sep=1pt, label=right:{$A_3$}] (A4) at (4*\xs,-0.5*\ys) {};

`					\draw (A1) to (A2) to (A3) to (A4);
					\draw[dashed] (O) -- (A1) ;
					\draw[dashed] (O) -- (A2) ;
					\draw[dashed] (O) -- (A3) ;
					\draw[dashed] (O) -- (A4) ;
				\end{scope}
				\begin{scope}[shift={(7*\xs,0)}]
					\node[circle, fill, inner sep=1pt, label=above:{$O$}] (O) at (0,0) {};
					\node[circle, fill, inner sep=1pt, label=left:{$\tilde{A}_0$}] (A1) at (-1*\xs,-0.5*\ys) {};
					\node[circle, fill, inner sep=1pt, label=below:{$\tilde{A}_1$}] (A2) at (2*\xs,-2*\ys) {};
					\node[circle, fill, inner sep=1pt, label=right:{$\tilde{A}_2$}] (A3) at (4*\xs,-0.5*\ys) {};

					\draw (A1) to (A2) to (A3);
					\draw[dashed] (O) -- (A1) ;
					\draw[dashed] (O) -- (A2) ;
					\draw[dashed] (O) -- (A3) ;

				\end{scope}
			\end{tikzpicture}
			\caption{\label{fi:side disappear} On the left polygon before mutation, one the right polygon after mutation.}
		\end{figure}
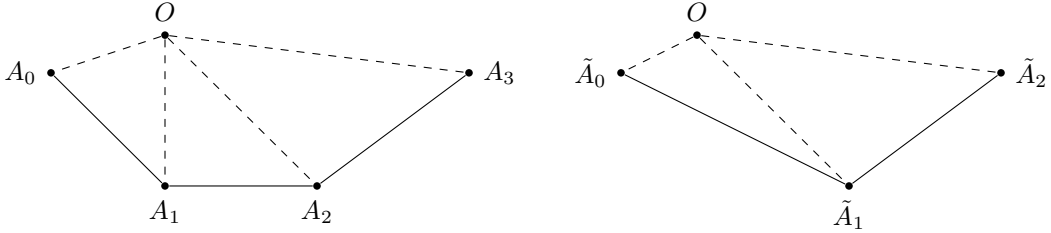
		
		Then we have 
		\begin{equation}
			\phi_{A_1}+\phi_{A_2} = \frac{h (h-b)}{h^2 h a}+ \frac{h (h-d)}{h^2 h (c-d)}=
			\frac{(c-h)(h-b)+(h-d)(a-b+h)}{h a h (c-d)}=\tilde{\phi}_{\tilde{A}_1}.
		\end{equation}
	\end{proof}
	Therefore we can use the same argument as above to formally decompose new vertex \(\tilde{A}_1\) into two with barycentric coordinates \(\phi_{A_1}\) and \(\phi_{A_2}\).
\end{proof}

\begin{Definition}
	The\emph{ elliptic Weyl group} is a semidirect product \(W\ltimes (P\oplus P)\).
\end{Definition}

The elliptic Weyl groups were introduced in \cite{Saito:1997}, see also exposition in \cite[Sec. 1.2]{Ion:2006Triple}. There is also a central extension of elliptic Weyl group by an abelian group \(\mathbb{Z}\). This extension is called \emph{double affine Weyl group}, two lattices \(P\) do not commute there.

\begin{Theorem} \label{th:extended group Painleve}
	There is an action of elliptic Weyl group \(W\ltimes (P \oplus P)\) on the space with cooridinates \(\lambda,\mu,\lambda_d,\mu_d, \mathbf{a}\) that preserves partition function \(\mathcal{Z}\left(\mathbf{a}|\lambda_d,\mu_d;\lambda,\mu\right)\).
\end{Theorem}
\begin{proof}
	Recall formula \eqref{eq:selfduality} for the partition function. It follows from Theorem~\ref{th:partition normalization} that function \(f(\mathbf{a}|\lambda,\mu)\) is invariant under the action of affine Weyl group \(W^{\mathrm{ae}}(E_n)\). Conjugating the action by the element \(\mathrm{w}\in W(E_n)\), we can define action of \(W^{\mathrm{ae}}(E_n)\) on \(\lambda_d,\mu_d\) which preserves the function \(f(\mathrm{w}(\mathbf{a})|\lambda_d,\mu_d)\). Therefore the group \(W^{\mathrm{ae}}(E_n)\) preserves partition function \(\mathcal{Z}\left(\mathbf{a}|\lambda_d,\mu_d;\lambda,\mu\right)\). 
	
	Note that translation subgroup \(P\subset W^{\mathrm{ae}}(E_n)\) preserves root variables \(\mathbf{a}\). Hence any translation can be decomposed into a product of two commuting transformation, where one of them transforms \(\lambda,\mu\) and another one transforms \(\lambda_d,\mu_d\). Each of these factors preserves partition function \(\mathcal{Z}\left(\mathbf{a}|\lambda_d,\mu_d;\lambda,\mu\right)\).  Thus we obtained action of the group \(W\ltimes (P\oplus P)\).
\end{proof}

\begin{Remark}
	We can consider partition function \(\mathcal{Z}\left(\mathbf{a}|\lambda_d,\mu_d;\lambda,\mu\right)\) as a polynomial in 4 variables \(\lambda,\mu,\lambda_d,\mu_d\). Then the Newton \emph{polytope} \(\hat{N}\) is 4-dimensional. Original polygon \(N\) and dual polygon \(N^D\) are projections on \(\hat{N}\) (c.f. \cite[Sec. 4]{Franco:2023twin}). Moreover, both zigzag and face mutations are mutations of \(\mathcal{Z}\left(\lambda_d,\mu_d;\lambda,\mu\right)\) as a polynomial on 4 variables. Therefore the elliptic Weyl group can be considered as a subgroup in a group of polynomial mutations of  \(\mathcal{Z}\left(\lambda_d,\mu_d;\lambda,\mu\right)\).
\end{Remark}

\begin{Remark}
	It would be interesting to find an analog of elliptic Weyl group action after deautonomization, i.e. without condition \(q=1\). This could be connected to the quantization of the spectral parameters \(\lambda,\mu\). In terms of group it is tempting to suggest that this would lead  to the mentioned above central extension of double affine Weyl group.
\end{Remark}

\begin{Remark}
	In the proof of Theorem~\ref{th:extended group Painleve} we defined action of the elements of Weyl group on both spectral and dynamical variables. From the cluster point of view the same element of the group is realized both by composition of mutations and permutations on face quiver and composition of mutations and permutations on zigzag quiver. This can be quite nontrivial, for example, simple transposition of two zigzag variables correspond to transformation \(R\) given in formula \eqref{eq:permutation of zigzags} in terms of face quiver.
\end{Remark}

\section{\(q\)-Painlev\'e for reflexive polygons} \label{sec:examples}

We denote by \(\mu_i\) mutations in the vertices of the face quiver \(\mathcal{Q}\) and \(\mu_{d,i}\) in the vertices of the zigzag quiver \(\mathcal{Q}^D\). In presentation of elliptic Weyl group \(W\ltimes (P\oplus P)\) we denote by \(T\) and \(T_d\) translation given by compositions of mutations in quiver \(\mathcal{Q}\) and \(\mathcal{Q}^D\) correspondingly. In particular, translations \(T\) preserve spectral variables \(\lambda,\mu\) and acts non-trivially on dual spectral variables \(\lambda_d, \mu_d\). Contrary \(T_d\) preserves \(\lambda_d, \mu_d\) and transforms \(\lambda,\mu\).

In this and following sections we will be using shorthand notations 
\(a_{i_1i_2\dots i_k}=a_{i_1}\cdot a_{i_2}\cdot\dots \cdot a_{i_k}\), \(x_{i_1i_2\dots i_k}=x_{i_1}\cdot x_{i_2}\cdot\dots \cdot x_{i_k}\), and  \(s_{i_1i_2\dots i_k}=s_{i_1}\cdot s_{i_2}\cdot\dots \cdot s_{i_k}\).

The Painlev\'e equations below are parametrized by pairs symmetry type/surface type. The symmetry types were given in the Fig.~\ref{fi:qPeq} above, and surface type for the symmetry type \(E_n^{(1)}\) is given by \(A_{8-n}^{(1)}\)  We included surface types to make comparison with literature on \(q\)-difference Painlev\'e equations easier.

In figures below we label faces and zigzags by their variables. We also encircle face variables. The edge weights are depicted in blue. The sign in front of weight represents Kasteleyn sign. 

\subsection{\(E_1^{(1)}/A_7^{(1)}\)}

The bipartite graph, Newton polygon and quiver corresponding to this example were already presented on Fig.~\ref{fi:ex square}. In order to compute dimer partition function by the formula \eqref{eq:ZKast} we will need to specify edge weights and Kasteleyn signs. They are given in Fig. \ref{fi:E1p}. 
\begin{figure}[h]
	\centering
    \includegraphics[]{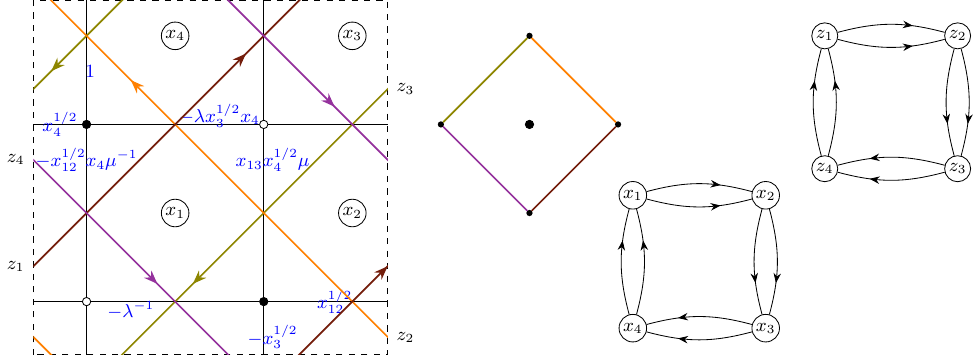}

	\caption{\label{fi:E1p}
		For \(E_1^{(1)}/A_7^{(1)'}\) Painlev\'e equation: bipartite graph with edge weights, Newton polygon, face quiver with face variables, zigzag quiver with zigzag variables.}
\end{figure}

The dual spectral parameters can be introduced by
\begin{equation}
	\lambda_d = \sqrt{x_1 x_2}, \; \mu_d = \sqrt{\frac{x_2}{x_1}}, \; a_0 = x_1 x_3, \; a_1 = x_2 x_4.
\end{equation}
Zigzag variables can be expressed through $\lambda,\mu, a_0, a_1$ via formulas
\begin{equation}
	z_1 = \lambda \mu, \; z_2 = \frac{\mu}{\lambda} a_0, \; z_3 = \frac{1}{\lambda \mu} \frac{1}{a_0}, \; z_4 = \frac{\lambda}{\mu}.
\end{equation}

The Hamiltonian \eqref{TodaHam} comes from the normalized partition function of dimers
\begin{equation}
	\mathcal{Z}\left(\mathbf{a}|\lambda_d,\mu_d;\lambda,\mu\right)=f(z|\lambda,\mu)-f(z|\lambda_d,\mu_d)
\end{equation}
where
\begin{equation}
	f(\mathbf{a}|\lambda,\mu)= a_1^{1/4}(\lambda + \lambda^{-1} + a_0 \mu + \mu^{-1}).
\end{equation}

The finite symmetry group is generated by order four element $\pi$ acting on face and zigzag variables by transformations
\begin{subequations}
	\begin{alignat}{3}
		&\pi|_{\text{face variables}}= (4,3,2,1),& \qquad &\pi|_{\text{zigzag variables}}= (4_d,3_d,2_d,1_d).
	\end{alignat}
\end{subequations}
Its action on $\mathbf{a}$ variables given on Fig.~\ref{fi:E1pdynkin}, right.
\begin{figure}[h]
	\begin{center}
		\begin{tikzpicture}[elt/.style={circle,draw=black!100,thick, inner sep=0pt,minimum size=2mm},scale=1.25]
			
			\path 	(0,0) 	node 	(a0) [elt] {}
			(1,0) 	node 	(a1) [elt] {};
			\draw [<->,black,line width=1pt ] (a0) -- (a1);
			
			\node at ($(a0.north) + (0.05,0.2)$) 	{$a_{0}$};			    
			\node at ($(a1.north) + (0.05,0.2)$) 	{$a_{1}$};
			
		\end{tikzpicture}
		\hspace{2cm}
		\begin{tikzpicture}
			\node [shape=rectangle] {
				\begin{tabular}{|c||c|c|}
					\hline
					& \(a_0\) & \(a_1\) \\
					\hline
					\(\pi\)	& \(a_1\) & \(a_0\)
					\\
					\hline	
				\end{tabular}
			}; 
		\end{tikzpicture}
	\end{center}
	\caption{\label{fi:E1pdynkin}
		For \(A_1^{(1)}/A_7^{(1)'}\) Painlev\'e equation: Dynkin diagram and action on $\mathbf{a}$ variables}
\end{figure}

Action of $\pi$ on spectral and dual spectral variables is as follows:
\begin{center}
	\begin{tabular}{|c||c|c|c|c|c|}
		\hline
		& \(\lambda\) & \(\mu\) & \(\lambda_d\) & \(\mu_d\) \\
		\hline
		\(\pi\)	&  \(\sqrt{a_0} \mu\) & \(\frac{\sqrt{a_0}}{\lambda}\) & \( \sqrt{a_0} \mu_d \)& \( \frac{\sqrt{a_0}}{\lambda_d} \)
		\\
		\hline 	
	\end{tabular}
\end{center}
The symmetry group also contains lattice of translations of the form \(P\oplus P_d\) were \(P\simeq P_d\simeq \mathbb{Z}\oplus \mathbb{Z}/2\mathbb{Z}\). It is generated by infinite order translations  \(T=(1,2)(3,4)\mu_3\mu_1\), \(T_d=(1_d,2_d)(3_d,4_d)\mu_{d,3}\mu_{d,1}\), and order two ''translations'' given by $t = (1,3)(2,4)$,  $t_d = (1_d,3_d)(2_d,4_d)$, which act on spectral and dual spectral parameters by
{\scriptsize
	\begin{subequations}
		\begin{align}
			&T(\lambda_d)= \frac{ \mu_d(\lambda_d + a_0 \mu_d)}{\lambda_d  + \mu_d}, &		&T(\mu_d)=\frac{\lambda_d + \mu_d}{  \lambda_d(\lambda_d +  a_0 \mu_d)},
			\\
			&t(\lambda_d)= \frac{1}{\lambda_d}, &		&t(\mu_d)=\frac{a_1}{\mu_d},
			\\
			&T_{d}(\lambda)= \frac{\mu (a_0 \lambda  \mu + 1)}{\lambda  \mu + 1}, &		&T_{d}(\mu)= \frac{a_0 \lambda  \mu + 1}{a_0 \lambda(\lambda  \mu + 1)},
			\\
			&t_{d}(\lambda)= \frac{1}{\lambda}, &		&t_{d}(\mu)= \frac{a_1}{\mu}.
		\end{align}
	\end{subequations}
}
The relations between translations and $\pi$ are
\begin{equation}
	\pi^2 = t_d t, ~~~ \pi t = t \pi, ~~~ \pi T \pi = T^{-1}, ~~~ \pi t_d = t_d \pi, ~~~ \pi T_d \pi = T^{-1}_d.
\end{equation}

\subsection{\((E_1^{(1)})_{|\alpha^2|=8}/A_7^{(1)}\)}

We choose the bipartite graph and weights of the edges as on the Fig.~\ref{fi:E1}.

\begin{figure}[h]
	\centering
	\includegraphics[]{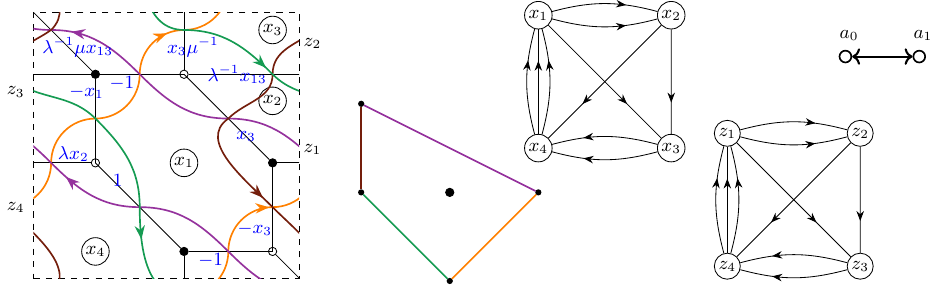}

	\caption{\label{fi:E1}
		For \(E_1^{(1)}/A_7^{(1)}\) Painlev\'e equation: bipartite graph with edge weights, Newton polygon, face quiver with face variables, zigzag quiver with zigzag variables, Dynkin diagram}
\end{figure}

The root variables and dual spectral parameters $\lambda_d, \mu_d$ can be chosen as
\begin{equation}
    \lambda_d = \frac{x_3}{x_2}, \; \mu_d = \frac{1}{x_2}, \; a_0 = \frac{(x_2)^2 x_4}{x_3}, \; a_1 = \frac{x_1 (x_3)^2}{x_2}
\end{equation}
and zigzag variables can be chosen as 
\begin{equation}
    z_1 = \frac{\mu}{\lambda^2} a_1, \; z_2 = \frac{1}{\mu}, \; z_3 = \frac{\lambda}{\mu}, \; z_4 = \lambda \mu \frac{1}{a_1}.
\end{equation}
The dimer partition function has form 
\begin{equation}
	\mathcal{Z}=f(\mathbf{a}|\lambda,\mu)-f(\mathbf{a}|\lambda_d,\mu_d),
\end{equation}
where
\begin{equation}
	f(\mathbf{a}|\lambda,\mu)=a_1^{-1}\lambda+\lambda^{-1}\mu+\lambda^{-1}+\mu^{-1}.
\end{equation}
The symmetry group in this case coinsides with translation group $P\oplus P_d$ where $P \simeq P_d \simeq \mathbb{Z}$. Generators of this group has claster description \(T=(1324)\mu_3\) and $T_d=(1_d 3_d 2_d 4_d)\mu_{d,3}$. Their action on spectral and dual spectral parameters given by 
{\scriptsize
	\begin{subequations}
		\begin{align}
			&T(\lambda_d)= \frac{a_1(\lambda_d + \mu_d)}{\lambda_d \mu_d}, &		&T(\mu_d)=\frac{\lambda_d}{\mu_d},
			\\
            &T_{d}(\lambda)= \frac{a_1(\lambda + \mu)}{\lambda \mu}, &		&T_{d}(\mu)=\frac{\lambda}{\mu}.
			\end{align}
	\end{subequations}
}

\subsection{\(E_2^{(1)}/A_6^{(1)}\)}
We choose the bipartite graph and weights of the edges as on the Fig.~\ref{fi:E2}.

\begin{figure}[h]
    \centering
    \includegraphics{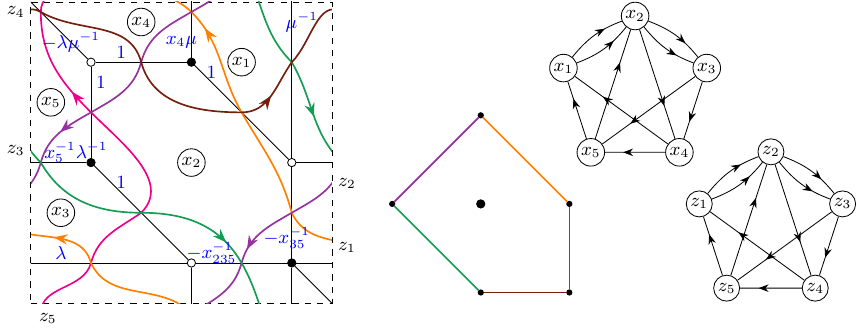}
	\caption{\label{fi:E2}
		For \(E_2^{(1)}/A_6^{(1)}\) Painlev\'e equation: bipartite graph with edge weights, Newton polygon, face quiver with face variables, zigzag quiver with zigzag variables, Dynkin diagram}
\end{figure}
The root variables \(\mathbf{a}\) and dual spectral parameters \(\lambda_d,\mu_d\) can be chosen as follows
\begin{equation}
	\lambda_d=x_4,\; \mu_d=x_5,\; a_0=x_1 x_3,\; a_2=\frac{x_3}{x_1} \left(\frac{x_5}{x_4}\right)^2.
\end{equation}
The zigzag variables are equal to 
\begin{equation}
	z_1=\dfrac{\mu}{\lambda} \dfrac{1}{a_0 b},\; z_2=\dfrac{1}{\lambda \mu} a_0,\; z_3= \dfrac{\lambda}{\mu} b,\; z_4=\lambda,\; z_5 = \mu,
\end{equation}
where we introduced variable $b = \sqrt{\frac{a_2}{a_0}}$. This variable has geometric meaning, in this case the rational surface \(\mathcal{S}_{\mathbf{a}}\) by construction depends on it, see e.g. \cite[p. 216]{Sakai:2001}. This agrees with the fact that this square root can be taken in cluster variables.

The formula for the normalized partition function has the form (see Theorem~\ref{th:selfduality})
\begin{equation}
	\mathcal{Z}\left(\mathbf{a}|\lambda_d,\mu_d;\lambda,\mu\right) =f(\mathbf{a}|\lambda,\mu)-f(s_{1}\mathbf{a}|\lambda_d,\mu_d),
\end{equation}
where 
\begin{equation}
	f(\mathbf{a}|\lambda,\mu) = \lambda + \lambda\mu^{-1} + \mu^{-1} + b^{-1}\lambda^{-1} + a_0^{-1}b^{-1}\mu
	.	
\end{equation}
The symmetry group contains generators
\begin{subequations}
	\begin{alignat}{3}
		&s_0|_{\text{face variables}} = \mu_3 (1,3) \mu_3,& \qquad &s_0|_{\text{zigzag variables}}= \mu_{d,3} (1_d,3_d) \mu_{d,3}, 
		\\
		&s_1|_{\text{face variables}}= \mu_2 \mu_5 (4,5) \mu_5 \mu_2,& \qquad &s_{1}|_{\text{zigzag variables}} = \mu_{d,2} \mu_{d,5} (4_d,5_d) \mu_{d,5} \mu_{d,2},
		\\
        &\pi|_{\text{face variables}}= (1,2,3,5,4)\mu_5,& \qquad &\pi|_{\text{zigzag variables}}= (1_d,2_d,3_d,4_d,5_d)\mu_{d,5}.
	\end{alignat}
\end{subequations}
which constitute an affine Weyl group $W(A_1^{(1)})$ extended by generator $\pi$ of infinite order acting on $A_1^{(1)}$ by external automorpism in agreement with \cite[p. 226]{Sakai:2001}. The finite symmetry group acts on root variables as in  Fig.~\ref{fi:E2dynkin}.
\begin{figure}[h]
\begin{center}
		\begin{tikzpicture}[elt/.style={circle,draw=black!100,thick, inner sep=0pt,minimum size=2mm},scale=1.25]

			\path 	(0,0) 	node 	(a0) [elt] {}
				(1,0) 	node 	(a1) [elt] {}
				(2,0) 	node 	(a2) [elt] {}
				(3,0)	node 	(a3) [elt] {};
			\draw [<->,black,line width=1pt] (a0) -- (a1);
			\draw [<->,black,line width=1pt ] (a2) -- (a3);
			
			\node at ($(a0.north) + (0.05,0.2)$) 	{$a_{0}$};			    
			\node at ($(a1.north) + (0.05,0.2)$) 	{$a_{1}$};
			\node at ($(a2.north) + (0.05,0.2)$)  {$a_{2}$};
			\node at ($(a3.north) + (0.05,0.2)$)  {$a_{3}$};

        \end{tikzpicture}
\hspace{2cm}
\begin{tikzpicture}
\node [shape=rectangle] {
	\begin{tabular}{|c||c|c|}
		\hline
		& \(a_0\) & \(a_2\)\\
		\hline
		\(s_0\)	& \(a_0^{-1}\) & \(a_2\)
		\\
		\hline	
		\(s_1\)	& \(a_0^{-1}\) & \(a_2\)
		\\
		\hline	
		\(\pi\)	& \(a_0^{-1}\) & \(a_2\)
		\\
		\hline	
	\end{tabular}
}; 
\end{tikzpicture}
\end{center}
\caption{\label{fi:E2dynkin}
		For \(E_2^{(1)}/A_6^{(1)}\) Painlev\'e equation: Dynkin diagram and action on $\mathbf{a}$ variables}
\end{figure}

The action on spectral parameters and dual spectral parameters is presented in the following table.
\begin{center}
	\begin{tabular}{|c||c|c|c|c|}
		\hline
		& \(\lambda\) & \(\mu\) & \(\lambda_d\) & \(\mu_d\) \\
		\hline
		\(s_0\)	& \( \frac{\lambda (a_0 b \lambda + \mu )}{a_0 (b \lambda + \mu )} \) & \(\frac{\mu (a_0 b \lambda + \mu )}{a_0(b \lambda + \mu )}\) & \(\frac{a_0 \lambda_d \left(b \lambda_d + \mu_d \right)}{a_0 b \lambda_d + \mu_d}\) & \(\frac{a_0 \mu_d (b \lambda_d + \mu_d )}{a_0 b \lambda_d +\mu_d}\)
		\\
		\hline
		\(s_1\) & \(\frac{\lambda  (\lambda  \mu +\lambda +1)}{\lambda \mu + a_0 \lambda +a_0}\) & \( \frac{\mu  (\lambda  \mu + \lambda + a_0)}{a_0 (\lambda  \mu +\lambda +1)} \) & \(\frac{a_0 \lambda_d (\lambda_d \mu_d + \lambda_d+1)}{a_0 \lambda_d \mu_d + \lambda_d + 1}\) & \(\frac{\mu_d (a_0 \lambda_d \mu_d+a_0 \lambda_d+1)}{\lambda_d \mu_d+\lambda_d+1}\)
		\\
		\hline		
		\(\pi\) & \(\mu^{-1}\) & \( b \lambda  (1 + \mu^{-1}) \) & \(\mu_d^{-1}\) & \(a_0 b \lambda_d(1+\mu_d^{-1})\)
		\\
		\hline		
	\end{tabular}
\end{center}

The translations $T_1, T_{1,d}$ and $T_2, T_{2,d}$ originate from $s_0\pi$ and $\pi^2$, preserve parameters \(\mathbf{a}\), and act as 
{\scriptsize
	\begin{subequations}
		\begin{align}
			&T_{1}(\lambda_d)= \frac{\lambda_d \mu_d+\lambda_d+1}{\mu_d(a_0 \lambda_d\mu_d+a_0 \lambda_d+1)}, &		&T_{1}(\mu_d)=\frac{a_0 b \lambda_d (\mu_d+1) (\lambda_d \mu_d+\lambda_d+1)}{\mu_d(a_0 \lambda_d \mu_d + a_0 \lambda_d + 1)},
			\\
			&T_{2}(\lambda_d)= \frac{\mu_d}{a_0 b \lambda_d (\mu_d+1)}, &		&T_{2}(\mu_d)=\frac{a_0 b \lambda_d \mu_d+a_0 b \lambda_d+\mu_d}{a_0 \lambda_d \mu_d (\mu_d+1)},
			\\
            &T_{1,d}(\lambda)= \frac{a_0 (\lambda  \mu +\lambda +1)}{\mu  (\lambda  \mu + \lambda + a_0)}, &		&T_{1,d}(\mu)= \frac{a_0 b \lambda  (\mu +1) (\lambda  \mu +\lambda +1)}{\mu  (\lambda  \mu +\lambda + a_0)},
            \\
            &T_{2,d}(\lambda)= \frac{\mu }{b \lambda  (\mu +1)}, &		&T_{2,d}(\mu)= \frac{a_0 (b \lambda  \mu +b \lambda +\mu )}{\lambda  \mu  (\mu +1)}.
			\end{align}
	\end{subequations}
}
The group contains \(W \ltimes (P\oplus P_d)\) in agreement with Theorem \ref{th:extended group Painleve}. Here \(W\simeq W(A_1)\) is generated by $s_1$ and \(P\simeq P_d\simeq \mathbb{Z}^2\) is generated by $T_1$ and $T_2$.

\subsection{\(E_3^{(1)}/A_5^{(1)}\)}
We choose the bipartite graph and weights of the edges as on the Fig.~\ref{fi:E3}.

\begin{figure}[h]
	\centering
	\includegraphics[]{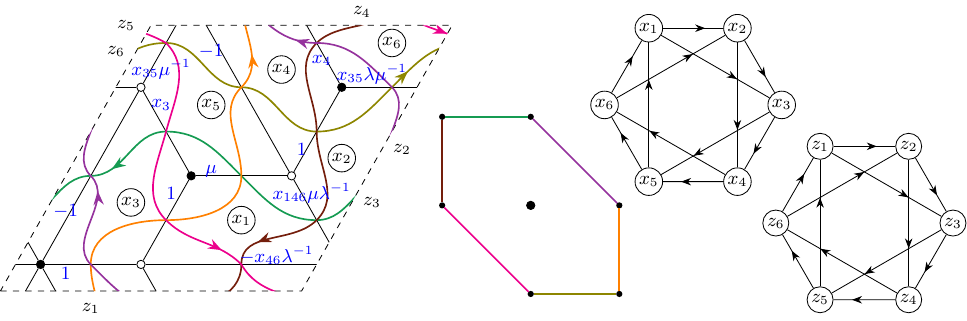}

	\caption{\label{fi:E3}
		For \(E_3^{(1)}/A_5^{(1)}\) Painlev\'e equation: bipartite graph with edge weights, Newton polygon, face quiver with face variables, zigzag quiver with zigzag variables}
\end{figure}
The root variables \(\mathbf{a}\) and dual spectral parameters \(\lambda_d,\mu_d\) can be chosen as follows
\begin{equation}
	\lambda_d=x_1,\; \mu_d=x_2,\; a_1=\frac{1}{x_1 x_4},\; a_2=\frac{1}{x_2 x_5},\; a_4=\frac{1}{x_1 x_3 x_5}.
\end{equation}
The zigzag variables are equal to 
\begin{equation}
	z_1=\mu,\; z_2=\lambda^{-1}\mu a_1^{-1}a_4,\; z_3=\lambda^{-1}a_2,\; z_4=\mu^{-1}a_1,\; z_5=\lambda\mu^{-1}a_2^{-1}a_4^{-1},\; z_6=\lambda.
\end{equation}

The formula for the normalized partition function has the form (see Theorem~\ref{th:selfduality})
\begin{equation}
	\mathcal{Z}\left(\mathbf{a}|\lambda_d,\mu_d;\lambda,\mu\right) =f(\mathbf{a}|\lambda,\mu)-f(s_{1214}\mathbf{a}|\lambda_d,\mu_d),
\end{equation}
where 
\begin{equation}
	f(\mathbf{a}|\lambda,\mu)=a_1^{1/3}a_2^{-1/3}a_4^{1/2}\Big(a_1^{-1}a_2 \mu \lambda^{-1}+a_1^{-1}\mu +a_2\lambda^{-1}+a_4^{-1} \lambda \mu^{-1}+a_4^{-1}\lambda+a_4^{-1}\mu^{-1}\Big).
\end{equation}
The generators of simple reflections have the form
\begin{subequations}
	\begin{alignat}{3}
		&s_1|_{\text{face variables}}= \mu_1 (1,4) \mu_1,& \qquad &s_1|_{\text{zigzag variables}}= \mu_{d,1} (1_d,4_d) \mu_{d,1}, 
		\\
		&s_2|_{\text{face variables}}= \mu_2 (2,5) \mu_2,& \qquad &s_2|_{\text{zigzag variables}}= \mu_{d,3} (3_d,6_d) \mu_{d,3}, 
		\\
		&s_0|_{\text{face variables}}= \mu_3 (3,6) \mu_3,& \qquad &s_0|_{\text{zigzag variables}}= \mu_{d,2} (2_d,5_d) \mu_{d,2}, 
		\\
		&s_3|_{\text{face variables}}= \mu_1 \mu_3 (3,5) \mu_3 \mu_1,& \qquad & s_3|_{\text{zigzag variables}}= \mu_{d,1} \mu_{d,3} (3_d,5_d) \mu_{d,3} \mu_{d,1},
		\\
		&s_4|_{\text{face variables}}= \mu_2 \mu_4(4,6) \mu_4 \mu_2,& \qquad & s_4|_{\text{zigzag variables}}= \mu_{d,2} \mu_{d,4} (4_d,6_d) \mu_{d,4} \mu_{d,2}, 
		\\
		&\pi|_{\text{face variables}}= (1,2,3,4,5,6),& \qquad &\pi|_{\text{zigzag variables}}= (6,5,4,3,2,1). 				
	\end{alignat}
\end{subequations}



The finite symmetry group acts on root variables as in  Fig.~\ref{fi:E3dynkin}.
\begin{figure}[h]
\begin{center}
		\begin{tikzpicture}[elt/.style={circle,draw=black!100,thick, inner sep=0pt,minimum size=2mm},scale=1.25]

			\path 	(0,0) 	node 	(a1) [elt] {}
				(1,0) 	node 	(a2) [elt] {}
				(0.5,0.866) 	node 	(a0) [elt] {}
				(2,0) 	node 	(a3) [elt] {}
				(3,0)	node 	(a4) [elt] {};
			\draw [black,line width=1pt] (a0) -- (a1) -- (a2) -- (a0);
			\draw [<->,black,line width=1pt ] (a3) -- (a4);
			
			\node at ($(a0.west) + (-0.2,0)$) 	{$a_{0}$};			    
			\node at ($(a1.south west) + (-0.2,-0.2)$) 	{$a_{1}$};
			\node at ($(a2.south east) + (0.2,-0.2)$)  {$a_{2}$};
			\node at ($(a3.north) + (0,0.2)$)  {$a_{3}$};
			\node at ($(a4.north) + (0,0.2)$)  {$a_{4}$};	

		\end{tikzpicture}
\hspace{2cm}
\begin{tikzpicture}
\node [shape=rectangle] {
	\begin{tabular}{|c||c|c|c|c|c|}
		\hline
		& \(a_1\) & \(a_2\) & \(a_4\) \\
		\hline
		\(s_1\)	& \(a_1^{-1}\) & \(a_1 a_2\) & \(a_4\)
		\\
		\hline	
		\(s_2\)	& \(a_1 a_2\) & \(a_2^{-1}\) & \(a_4\)
		\\
		\hline	
		\(s_4\)	& \(a_1\) & \(a_2\) & \(a_4^{-1}\)
		\\
		\hline	
	\end{tabular}
}; 
\end{tikzpicture}
\end{center}
\caption{\label{fi:E3dynkin}
		For \(E_3^{(1)}/A_5^{(1)}\) Painlev\'e equation: Dynkin diagram and action on $\mathbf{a}$ variables}
\end{figure}

Action on (dual-) spectral parameters is contained in the following table.
\begin{center}
	\begin{tabular}{|c||c|c|c|c|c|c|c|}
		\hline
		& \(\lambda\) & \(\mu\) & \(\lambda_d\) & \(\mu_d\) \\
		\hline
		\(s_1\)	& \(\frac{\lambda(1+\mu)a_1}{a_1+\mu}\) & \(a_1^{-1}\mu\) & \(a_1 \lambda_d\) & \(\mu_d\frac{1+a_1\lambda_d}{a_1(1+\lambda_d)}\)
		\\
		\hline 
		\(s_2\)& \(a_2^{-1}\lambda\) &\(\mu\frac{a_2+\lambda}{1+\lambda }\) & \(\lambda_d\frac{\mu_d+1}{a_2 \mu_d+1}\) & \(a_2\mu_d\)
		\\
		\hline
		\(s_4\)& \( \frac{\lambda(a_4\mu+a_1\lambda +a_1)}{a_4(\mu+a_1\lambda +a_1 )}\) & \( \frac{\mu(a_4\mu+a_1\lambda +a_1a_4)}{\mu+a_1\lambda +a_1 }\) & \(\frac{\lambda_d(a_2\mu_d+a_4\lambda_d+a_4)}{a_2\mu_d+\lambda_d+1}\) & \(\frac{\mu_d(a_2 \mu_d+a_4 \lambda_d+1)}{a_4(a_2\mu_d+\lambda_d+1)} \)
		\\
		\hline		
	\end{tabular}
\end{center}

The translations preserve parameters \(\mathbf{a}\) and act as 
{\scriptsize
	\begin{subequations}
		\begin{align}
			&T_{1}(\lambda_d)=\frac{\mu_d(1+a_1\lambda_d) (a_1a_2\lambda_d\mu_d+a_2\mu_d+\lambda_d+1)}{a_1a_4\lambda_d(\lambda_d+1) (a_1\lambda_d \mu_d+a_1\lambda_d+a_1+\mu_d)}, &		&T_{1}(\mu_d)=\frac{a_1a_2\lambda_d\mu_d+a_2\mu_d+\lambda_d+1}{a_2\lambda_d (a_1\lambda_d \mu_d+a_1\lambda_d+a_1+\mu_d)}
			\\	
			&T_{2}(\lambda_d)=\frac{a_2 \mu_d+\lambda_d \mu_d+\lambda_d+1}{\mu_d (a_1 a_2 \lambda_d \mu_d+a_1 a_2 \lambda_d+a_2 \mu_d+1)}, &	&T_{2}(\mu_d)=\frac{a_1 a_4 \lambda_d (\mu_d+1) (a_2\mu_d+\lambda_d\mu_d+\lambda_d+1)}{\mu_d (a_2 \mu_d+1) (a_1 a_2\lambda_d \mu_d+a_1 a_2 \lambda_d+a_2 \mu_d+1)}
			\\
			&T_{4}(\lambda_d)=\frac{a_2 \mu_d+\lambda_d+1}{a_1 \lambda_d (a_2\mu_d+a_4 \lambda_d+a_4)}, &
			&T_{4}(\mu_d) =	\frac{a_4 (a_2 \mu_d+\lambda_d+1)}{a_2 \mu_d (a_2\mu_d+a_4 \lambda_d+1)}
			\\
			&T_{1,d}(\lambda)=\frac{a_2 (a_1 \lambda  \mu +a_1 \lambda +a_1+\mu )}{\mu  (a_1 a_2+a_2 \mu +\lambda  \mu +\lambda )},& &T_{1,d}(\mu)=\frac{a_1 \lambda  (\mu +1) (a_1 \lambda  \mu +a_1 \lambda +a_1+\mu )}{a_4 \mu  (a_1+\mu ) (a_1 a_2+a_2 \mu +\lambda  \mu +\lambda )}
			\\
			&T_{2,d}(\lambda)=\frac{a_4 \mu  (a_2+\lambda ) (a_1 a_2 \lambda +a_1 a_2+a_2 \mu +\lambda  \mu )}{a_1 \lambda  (\lambda +1) (a_2 \mu +\lambda  \mu +\lambda +1)},& &T_{2,d}(\mu)=\frac{a_1 a_2 \lambda +a_1 a_2+a_2 \mu +\lambda  \mu }{\lambda  (a_2 \mu +\lambda  \mu +\lambda +1)}
			\\
			&T_{4,d}(\lambda)=\frac{a_2 a_4 (a_1 \lambda +a_1+\mu )}{\lambda  (a_1 \lambda +a_1+a_4 \mu )},& &T_{4,d}(\mu)=\frac{a_1 (a_1 \lambda +a_1+\mu )}{\mu  (a_1 a_4+a_1 \lambda +a_4 \mu )}
			\end{align}
	\end{subequations}
}

The whole group is \(W \ltimes (P\oplus P_d)\) in agreement with Theorem \ref{th:extended group Painleve}. Here \(W\simeq W(A_2+A_1)\) and \(P\simeq P_d\simeq \mathbb{Z}^3\)

\subsection{\(E_4^{(1)}/A_4^{(1)}\)}
We choose the bipartite graph and weights of the edges as on the Fig.~\ref{fi:E4}.

\begin{figure}[h]
    \centering
    \includegraphics{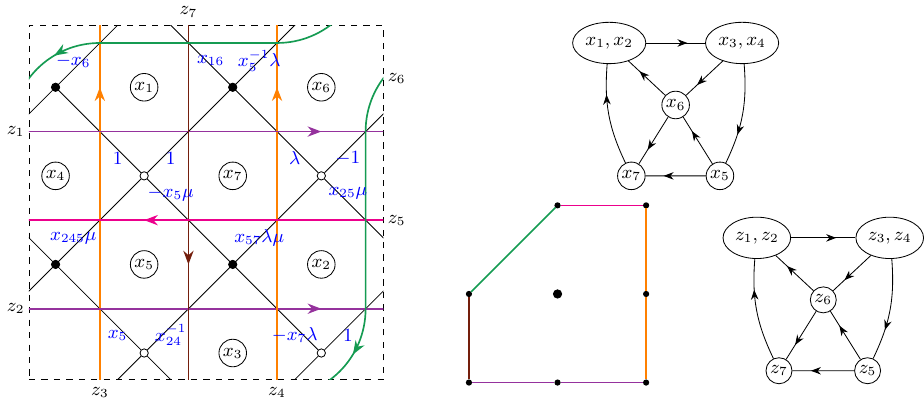}
\caption{\label{fi:E4}
		For \(E_4^{(1)}/A_4^{(1)}\) Painlev\'e equation: bipartite graph with edge weights, Newton polygon, face quiver with face variables, zigzag quiver with zigzag variables}
\end{figure}
The root variable \(\mathbf{a}\) and dual spectral parameters \(\lambda_d,\mu_d\) can be chosen as follows
\begin{equation}
	\lambda_d=x_1,\; \mu_d=x_4,\; a_0 = x_1 x_3 x_6, \; a_1=\frac{x_2}{x_1},\; a_2=x_1 x_5,\; a_3=x_3 x_7, \; a_4 = \frac{x_4}{x_3}.
\end{equation}
The zigzag variables are equal to 
\begin{equation}
	z_1=\lambda,\; z_2=\lambda \frac{1}{a_0},\; z_3=\mu a_0 a_1 a_4,\; z_4=\mu,\; z_5=\frac{1}{\lambda} \frac{1}{a_3 a_4}\; z_6 = \frac{1}{\lambda \mu} \frac{1}{a_1},\; z_7 = \frac{1}{\mu} a_3.
\end{equation}
The formula for the normalized partition function has the form (see Theorem~\ref{th:selfduality})
\begin{equation}
	\mathcal{Z}\left(\mathbf{a}|\lambda_d,\mu_d;\lambda,\mu\right) =f(\mathbf{a}|\lambda,\mu)-f(s_{2324}\mathbf{a}|\lambda_d,\mu_d),
\end{equation}
where 
\begin{multline}
	f(\mathbf{a}|\lambda,\mu)=a_0^{-5/5}a_1^{-3/5} a_2^{-6/5} a_3^{-4/5} a_4^{-2/5} \Big(\lambda^{-1} a_0 a_2 + \mu a_0 a_1 a_2 + \\ + \lambda  \mu + \lambda (1+a_2 a_3) + \lambda \mu^{-1} a_2 a_3 + \mu^{-1} a_2 a_3(1+a_0) + \lambda^{-1} \mu^{-1} a_0 a_2 a_3\Big).
\end{multline}
The simple reflections act by
\begin{subequations}
\label{eq:A4reflections}
	\begin{alignat}{3}
		&s_0|_{\text{face variables}}= \mu_6 \mu_3 (1,3) \mu_3 \mu_6,& \qquad &s_0|_{\text{zigzag variables}}= (1_d,2_d), 
		\\
		&s_1|_{\text{face variables}}= (1,2),& \qquad &s_1|_{\text{zigzag variables}}= \mu_{d,2} \mu_{d,7} \mu_{d,5} (3_d,5_d) \mu_{d,5} \mu_{d,7} \mu_{d,2}, 
		\\
		&s_2|_{\text{face variables}}= \mu_1(1,5)\mu_1,& \qquad &s_2|_{\text{zigzag variables}} = \mu_{d,3} (3_d,7_d) \mu_{d,3}, 
		\\
		&s_3|_{\text{face variables}}= \mu_3 (3,7) \mu_3,& \qquad & s_3|_{\text{zigzag variables}} = \mu_{d,4} (4_d,7_d) \mu_{d,4},
		\\
		&s_4|_{\text{face variables}}= (3,4),& \qquad & s_4|_{\text{zigzag variables}}= \mu_{d,6} \mu_{d,3} (2_d,3_d) \mu_{d,3} \mu_{d,6}, 
		\\
		&\pi|_{\text{face variables}}= (1,3)(2,4,5,6,7)\mu_3,& \qquad &\pi|_{\text{zigzag variables}}= (1_d,4_d,5_d,6_d,3_d)(2_d,7_d)\cdot \nonumber \\
        & & \qquad &  \hspace{4.5cm} \cdot \mu_{d,1}\mu_{d,2}\mu_{d,6}\mu_{d,1}\mu_{d,3}.
	\end{alignat}
\end{subequations}

The reflections affine Weyl group $A_4^{(1)}$ whose Dynkin diagram and action on $\mathbf{a}$ variables is given on Fig.~\ref{fi:E4dynkin}.
\begin{figure}[h]
\begin{center}
\begin{tikzpicture}[elt/.style={circle,draw=black!100,thick, inner sep=0pt,minimum size=2mm},scale=1.25]

			\path (0.809,0.588) 	node 	(a0) [elt] {}
			(1.618,0) 	node 	(a1) [elt] {}
			(1.309,-0.951) 	node 	(a2) [elt] {}
			(0.309,-0.951)	node 	(a3) [elt] {}
           	(0,0) 	node 	(a4) [elt] {};
	\draw [black,line width=1pt] (a0) -- (a1) -- (a2) -- (a3) -- (a4) -- (a0);

			\node at ($(a0.north) + (0,0.2)$) 	{$a_{0}$};
	\node at ($(a1.north east) + (0.2,0)$) {$a_{1}$};
	\node at ($(a2.east) + (0.2,0)$) {$a_{2}$};
	\node at ($(a3.west) + (-0.2,0)$) {$a_{3}$};
	\node at ($(a4.north west) + (-0.2,0)$)  {$a_{4}$};	

\end{tikzpicture}
\hspace{2cm}
\begin{tikzpicture}
\node [shape=rectangle] {
	\begin{tabular}{|c||c|c|c|c|c|}
		\hline
		& \(a_0\) & \(a_1\) & \(a_2\) & \(a_3\) & \(a_4\) \\
		\hline
		\(s_0\)	& \(a_0^{-1}\) & \(a_0 a_1\) & \(a_2\) & \(a_3\) & \(a_0 a_4\)
		\\
		\hline
		\(s_1\)	& \(a_1 a_0 \) & \(a_1^{-1}\) & \(a_1 a_2\) & \(a_3\) & \(a_4\)
        \\
        \hline
		\(s_2\)	& \(a_0 \) & \(a_2 a_1\) & \(a_2^{-1}\) & \(a_2 a_3\) & \(a_4\)
        \\
        \hline
		\(s_3\)	& \(a_0 \) & \(a_1\) & \(a_3 a_2\) & \(a_3^{-1}\) & \(a_3 a_4\)
        \\
        \hline
		\(s_4\)	& \(a_4 a_0 \) & \(a_1\) & \(a_2\) & \(a_4 a_3\) & \(a_4^{-1}\)
        \\
        \hline
		\(\pi\)	& \(a_2\) & \(a_3\) & \(a_4\) & \(a_0\) & \(a_1\)
        \\
        \hline
	\end{tabular}
}; 
\end{tikzpicture}
\end{center}
\caption{\label{fi:E4dynkin}
		For \(E_4^{(1)}/A_4^{(1)}\) Painlev\'e equation: Dynkin diagram and action on $\mathbf{a}$ variables}
\end{figure}

The action of reflections on $\lambda,\lambda$ and $\mu,\mu_d$ is provided in the following table:
\begin{center}
	\begin{tabular}{|c||c|c|c|c|}
		\hline
		& \(\lambda\) & \(\mu\) & \(\lambda_d\) & \(\mu_d\) \\
		\hline
		\(s_0\) & \(\frac{\lambda }{a_0}\) & \( \mu \) & \(\frac{\lambda_d (\lambda_d \mu_d + a_4 \lambda_d + a_4)}{\lambda_d \mu_d + a_0 a_4 \lambda_d + a_0 a_4}\) & \(\frac{\mu_d ( \lambda_d \mu_d + a_4 \lambda_d + a_0 a_4)}{\lambda_d \mu_d + a_4 \lambda_d + a_4}\)
		\\
		\hline
		\(s_1\) & \(\frac{a_1 \lambda ( \lambda \mu + a_2 a_3 \lambda + a_0 a_1 a_2  \mu + a_0 a_2 a_3 )}{\lambda \mu + a_1 a_2 a_3 \lambda + a_0 a_1 a_2 \mu + a_0 a_1 a_2 a_3 }\) & \( \frac{\mu  (\lambda \mu + a_2 a_3 \lambda + a_0 a_1 a_2 \mu + a_0 a_1 a_2 a_3 )}{\lambda \mu + a_2 a_3 \lambda + a_0 a_2 \mu + a_0 a_2 a_3} \) & \(a_1 \lambda_d\) & \(\mu_d\)
        \\
        \hline
		\(s_2\) & \(\frac{\lambda  (\mu + a_2 a_3)}{a_2 (\mu + a_3)}\) & \( \mu \) & \( \frac{\lambda_d}{a_2} \) & \( \frac{\mu_d (\lambda_d+a_2)}{\lambda_d+1} \)
        \\
        \hline
		\(s_3\) & \(\frac{a_3 \lambda  (\mu +1)}{\mu + a_3}\) & \( \frac{\mu }{a_3} \) & \( \frac{a_3 \lambda_d (\mu_d + a_4)}{\mu_d + a_3 a_4} \) & \(\mu_d\)
        \\
        \hline
		\(s_4\) & \(\frac{a_4 \lambda  (\lambda  \mu + a_2 a_3 \lambda + a_0 a_2 a_3 )}{\lambda  \mu + a_2 a_3 a_4 \lambda + a_0 a_2 a_3 a_4 }\) & \( \frac{\mu  ( \lambda  \mu + a_2 a_3 \lambda + a_0 a_2 a_3 a_4 )}{\lambda  \mu + a_2 a_3 \lambda + a_0 a_2 a_3} \) & \( \lambda_d \) & \(\frac{\mu_d}{a_4}\)
        \\
        \hline
		\(\pi\) & \( \frac{\lambda  \mu + a_0 a_2 a_3 + a_2 a_3 \lambda }{a_0 a_1 \mu  (\lambda  \mu + a_0 a_2 a_3 a_4 + a_2 a_3 \lambda )}\) & \( \frac{\lambda(\mu + a_2 a_3)}{a_2 a_3} \) & \( \frac{a_4}{\mu_d} \) & \( \frac{a_1 \lambda_d (\mu_d + a_4)}{a_4} \)
        \\
        \hline
	\end{tabular}
\end{center}

The translations $T_i$ and $T_{i,d}$ preserving parameters \(\mathbf{a}\) are inherited from the translations in the affine Weyl group. The formulas are lengthy so we will not provide them here. The whole group is \(W \ltimes (P\oplus P_d)\) in agreement with Theorem \ref{th:extended group Painleve}, with the finite symmetry group \(W\simeq W(A_4)\) generated by $s_1, s_2, s_3, s_4$ and with the lattice \(P\simeq P_d\simeq \mathbb{Z}^4\).

\subsection{\(E_5^{(1)}/A_3^{(1)}\)}
We choose the bipartite graph and weights of the edges as on the Fig.~\ref{fi:D5}
\begin{figure}[h]
    \centering
    \includegraphics{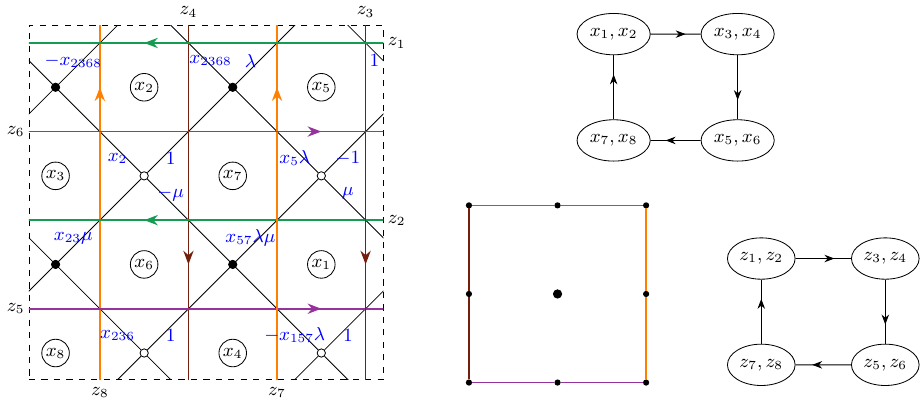}
	\caption{\label{fi:D5}
		Bipartite graph, Newton polygon, face quiver with face variables, zigzag quiver with zigzag variables}
\end{figure}

The root variables \(\mathbf{a}\) and dual spectral parameters \(\lambda_d,\mu_d\) can be chosen as follows
\begin{multline}
	\lambda_d=x_2^{-1},\; \mu_d=x_3^{-1},\; a_0=x_2 x_3 x_5 x_7,\; a_1=x_1 x_3 x_6 x_7,\;a_2=(x_3 x_7)^{-1},\\
    a_3=(x_2 x_6)^{-1},\; a_4=x_2 x_4 x_6 x_7,\; a_5=x_2 x_3 x_6 x_8.
\end{multline}
Notice that the \(\mathbf{a}\) variables satisfy constraint $a_0 a_1 (a_2)^2 (a_3)^2 a_4 a_5 = 1$. The zigzag variables are equal to 
\begin{multline}
	z_{1} = \lambda^{-1}, \; z_{2} = \lambda^{-1}(a_0)^{-1}, \; z_{3} = \mu^{-1}, \;	
	z_{4} = \mu^{-1} (a_{5})^{-1}, \; \\  z_{5} = \lambda a_{0}a_1 a_2, \;
	z_{6} = \lambda a_0 a_2, \; z_{7} = \mu a_{3} a_4 a_5,  \;
	z_{8} = \mu a_{3}a_5.
\end{multline}

The formula for the normalized partition function has the form (see Theorem~\ref{th:selfduality})
\begin{equation}
	\mathcal{Z}\left(\mathbf{a}|\lambda_d,\mu_d;\lambda,\mu\right) =f(\mathbf{a}|\lambda,\mu)-f(s_{1523452}\mathbf{a}|\lambda_d,\mu_d),
\end{equation}
where 
\begin{multline}
	f(\mathbf{a}|\lambda,\mu) = a_5^{1/4} a_1^{2/4} a_3^{2/4} a_4^{3/4} a_2^{4/4} \Big( a_0^2 a_1 a_2^2 a_3 \lambda \mu^{-1} + a_0 a_3 a_5 \lambda  \mu + a_5 a_3 \lambda^{-1} \mu + a_3 \lambda^{-1}  \mu^{-1} + \\ + a_3 (1 + a_5) \lambda^{-1} + a_0(1+a_0 a_1 a_2^2 a_3^2 a_5 ) \lambda + a_0 a_2 a_3(1+a_1) \mu^{-1} + a_3 a_5(1+a_0) \mu \Big).
\end{multline}

The generators of simple reflections have the form
\begin{subequations}
	\begin{alignat}{3}
		&s_0|_{\text{face variables}}= \mu_5 \mu_7 \mu_2 (2,3) \mu_2 \mu_7 \mu_5 ,& \qquad & s_0|_{\text{zigzag variables}}= (1_d,2_d), 	
		\\
		&s_1|_{\text{face variables}}= \mu_1 \mu_3 \mu_6 (6,7) \mu_6 \mu_3 \mu_1,& \qquad &s_1|_{\text{zigzag variables}}= (5_d,6_d),
		\\
		&s_2|_{\text{face variables}}= \mu_3 (3,7) \mu_3 ,&\qquad  &s_2|_{\text{zigzag variables}}= \mu_{d,2} (2_d,6_d) \mu_{d,2},
		\\
		&s_3|_{\text{face variables}}= \mu_2 (2,6) \mu_2,& \qquad & s_3|_{\text{zigzag variables}}=\mu_{d,4} (4_d,8_d) \mu_{d,4},
		\\
		&s_4|_{\text{face variables}}= \mu_2 \mu_4 \mu_6 (6,7) \mu_6 \mu_4 \mu_2,& \qquad & s_4|_{\text{zigzag variables}}= (7_d,8_d),
		\\
		&s_5|_{\text{face variables}}= \mu_2 \mu_3 \mu_6 (6,8) \mu_6 \mu_3 \mu_2,& \qquad & s_5|_{\text{zigzag variables}}= (3_d,4_d),
		\\
		&\pi|_{\text{face variables}}=(2,3,6,7)(1,4,5,8),& \qquad &\pi|_{\text{zigzag variables}} = (1,3,5,7)(2,4,6,8). 				
	\end{alignat}
\end{subequations}
The corresponding Dynkin diagram and action of generators on root variables is presented on Fig.~\ref{fi:E5dynkin}.
\begin{figure}[h]
\begin{center}
	\begin{tikzpicture}[elt/.style={circle,draw=black!100,thick, inner sep=0pt,minimum size=2mm},scale=1.25]
		\path 	( -1.5,1)	node 	(a0) [elt] {}
        (-1.5,-1) 	node 	(a1) [elt] {}
		(-0.5,0) 	node 	(a2) [elt] {}
		( 0.5,0)  node  	(a3) [elt] {}
		( 1.5,1) 	node  	(a4) [elt] {}
		( 1.5,-1) 	node 	(a5) [elt] {};
	\draw [black,line width=1pt ] (a0) -- (a2) -- (a3) -- (a4) (a1) --(a2) (a3) --  (a5);
	\node at ($(a1.south) + (0,-0.2)$) 	{$a_{1}$};
	\node at ($(a2.south) + (0,-0.2)$)  {$a_{2}$};
	\node at ($(a3.south) + (0,-0.2)$)  {$a_{3}$};
	\node at ($(a4.south) + (0,-0.2)$)  {$a_{4}$};	
	\node at ($(a5.south) + (0,-0.2)$)  {$a_{5}$};		
	\node at ($(a0.south) + (0,-0.2)$) 	{$a_{0}$};		
\end{tikzpicture}
\hspace{2cm}
\begin{tikzpicture}
\node [shape=rectangle] {
	\begin{tabular}{|c||c|c|c|c|c|c|}
		\hline
		& \(a_0\) & \(a_1\) & \(a_2\) & \(a_3\) & \(a_4\) & \(a_5\) \\
		\hline
		\(s_0\)	& \(a_0^{-1}\) & \(a_1\) & \(a_0 a_2\) & \(a_3\) & \(a_4\) & \(a_5\)
		\\
		\hline	
		\(s_1\)	& \(a_0\) & \(a_1^{-1}\) & \(a_1 a_2\) & \(a_3\) & \(a_4\) & \(a_5\)
		\\
		\hline	
		\(s_2\)	& \(a_0 a_2\) & \(a_1 a_2\) & \(a_2^{-1}\) & \(a_2 a_3\) & \(a_4\) & \(a_5\)
		\\
		\hline	
		\(s_3\)	& \(a_0\) & \(a_1\) & \(a_2 a_3\) & \(a_3^{-1}\) & \(a_3 a_4\) & \(a_3 a_5\)
		\\
		\hline	
		\(s_4\)	& \(a_0\) & \(a_1\) & \(a_2\) & \(a_3 a_4\) & \(a_4^{-1}\) & \(a_5\)
		\\
		\hline	
		\(s_5\)	& \(a_0\) & \(a_1\) & \(a_2\) & \(a_3 a_5\) & \(a_4\) & \(a_5^{-1}\)
		\\
		\hline	
		\(\pi\)	& \(a_4\) & \(a_5\) & \(a_3\) & \(a_2\) & \(a_1\) & \(a_0\)
		\\
		\hline	
	\end{tabular}
}; 
\end{tikzpicture}
\end{center}
\caption{\label{fi:E5dynkin}
		For \(D_5^{(1)}/A_3^{(1)}\) Painlev\'e equation: Dynkin diagram and action on $\mathbf{a}$ variables}
\end{figure}

The action on spectral parameters and dual spectral parameters is given in the table below.
\begin{center}
	\begin{tabular}{|c||c|c|c|c|}
		\hline
		& \(\lambda\) & \(\mu\) & \(\lambda_d\) & \(\mu_d\)
        \\
		\hline
		\(s_0\)	& \(a_0 \lambda\) & \(\mu\) & \(\frac{\lambda_d (a_0 (\lambda_d+1) \mu_d+a_2 a_0 \lambda_d+1)}{a_0 a_2 \lambda_d+\lambda_d \mu_d+\mu_d+1}\) & \(\frac{a_0 \mu_d (a_2 \lambda_d+\lambda_d \mu_d+\mu_d+1)}{a_0 \lambda_d \mu_d+a_0 a_2 \lambda_d+\mu_d+1}\)
		\\
		\hline	
		\(s_1\) & \(\lambda\) & \(\mu\) & \(\frac{\lambda_d (\lambda_d (a_2+\mu_d)+a_2 a_3 (a_1 \mu_d+1))}{\lambda_d (a_1 a_2+\mu_d)+a_1 a_2 a_3 (\mu_d+1)}\) & \(\frac{\mu_d (\lambda_d (a_2+\mu_d)+a_1 a_2 a_3 (\mu_d+1))}{\lambda_d (a_2+\mu_d)+a_2 a_3 (\mu_d+1)}\)
		\\
		\hline	
		\(s_2\) & \(\lambda\) & \( \frac{\mu(1+a_0 \lambda) }{1 + a_0 a_2 \lambda} \) & \(\frac{\lambda_d (a_2+\mu_d)}{\mu_d+1}\) & \(\frac{\mu_d}{a_2}\)
		\\
		\hline	
		\(s_3\) & \(\frac{\lambda(1 + a_3 a_5  \mu )}{a_3 (1 + a_5 \mu)}\) & \(\mu\) & \(\frac{\lambda_d}{a_3}\) & \(\frac{a_3 (\lambda_d+1) \mu_d}{a_3+\lambda_d}\)
		\\
		\hline	
		\(s_4\) & \(\lambda\) & \(\mu\) & \(\frac{a_4 \lambda_d (a_2 (a_3+\lambda_d)+(\lambda_d+1) \mu_d)}{a_2 a_4 (a_3+\lambda_d)+(\lambda_d+1) \mu_d}\) & \(\frac{\mu_d (a_2 (a_3 a_4+\lambda_d)+(\lambda_d+1) \mu_d)}{a_4 \lambda_d \mu_d+a_2 a_4 (a_3+\lambda_d)+\mu_d}\)
		\\
		\hline	
		\(s_5\) & \(\lambda\) & \(a_5 \mu\) & \(\frac{\lambda_d (a_3 a_5 \lambda_d \mu_d+a_3 a_5 \mu_d+a_3+\lambda_d)}{a_3 \lambda_d \mu_d+a_3 (\mu_d+1)+\lambda_d}\) & \(\frac{\mu_d (a_3 a_5 \lambda_d \mu_d+a_3 a_5 (\mu_d+1)+\lambda_d)}{a_3 a_5 \lambda_d \mu_d+a_3 (\mu_d+1)+\lambda_d}\)
		\\
		\hline
		\(\pi\) & \(\frac{1}{a_3 a_4 a_5 \mu}\) & \(\lambda\) & \(\frac{a_2}{\mu_d}\) & \(\lambda_d\)
		\\
		\hline
	\end{tabular}
\end{center}

The formulas for translations $T_1,...,T_5$, $T_{1,d},...,T_{5,d}$ are quite lengthy, so we will not provide them. However they can be obtained from translations in $W^{\mathrm{ae}}(D_5)$ written using generators above. The whole group is \(W \ltimes (P\oplus P_d)\) in agreement with Theorem \ref{th:extended group Painleve}. Here \(W\simeq W(D_5)\) and \(P\simeq P_d\simeq \mathbb{Z}^5\).

\subsection{\(E_6^{(1)}/A_2^{(1)}\)}
We choose the bipartite graph and weights of the edges as on the Fig.\ref{fi:E6}
\begin{figure}[h]
    \centering
    \includegraphics{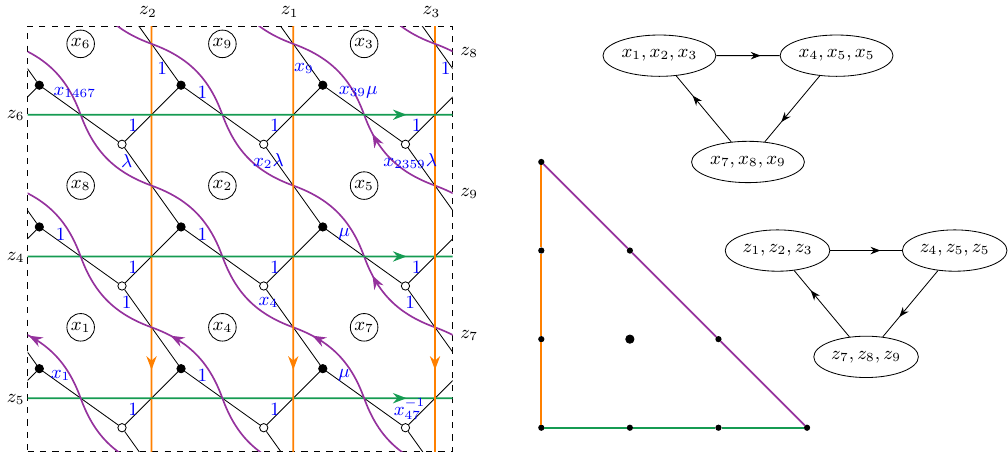}
%
	\caption{\label{fi:E6}
		Bipartite graph, Newton polygon, face quiver with face variables, zigzag quiver with zigzag variables}
\end{figure}

The root variables \(\mathbf{a}\) and dual spectral parameters \(\lambda_d,\mu_d\) can be chosen as follows
\begin{multline}
	\lambda_d=x_1^{-1},\; \mu_d=x_4,\; a_1=x_2x_3^{-1},\; a_2=x_1x_2^{-1},\;a_3=x_1^{-1}x_4^{-1}x_7^{-1},\\
    a_4=x_4x_5^{-1},\; a_5=x_5x_6^{-1},\; a_6=x_7x_8^{-1},\; a_0=x_8x_9^{-1}.
\end{multline}
The zigzag variables are equal to 
\begin{multline}
	z_{1} = \lambda^{-1}, \; z_{2} = \lambda^{-1} a_{3456}^{-1}, \; z_{3} = \lambda^{-1}a_{12334456}^{-1}, \;
	z_{4} = \mu, \; z_{5} = \mu a_{3}^{-1}, 
	\\ 
	z_{6} = \mu a_{2346}, \; z_{7} = \lambda \mu^{-1} a_{2}^{-1}, \;	
	z_{8} = \lambda \mu^{-1}a_{345}^{-1}, \; z_{9} = \lambda \mu^{-1} a_{1233456}^{-1}
	.
\end{multline}

The formula for the normalized partition function has the form (see Theorem~\ref{th:selfduality})
\begin{equation}
	\mathcal{Z}\left(\mathbf{a}|\lambda_d,\mu_d;\lambda,\mu\right) =f(\mathbf{a}|\lambda,\mu)-f(s_{435624346}\mathbf{a}|\lambda_d,\mu_d),
\end{equation}
where 
\begin{multline}
	f(\mathbf{a}|\lambda,\mu) =a_{1}^{-\frac{1}{3}} a_{2}^{-\frac{2}{3}} a_{3}^{-2} a_{4}^{-\frac{4}{3}} a_{5}^{-\frac{2}{3}} a_{6}^{-1} \Big(\lambda^2 \mu^{-1} a_{1} a_{2} a_{3}^{4} a_{4}^{3} a_{5}^{2} a_{6}^{2} 
	\\
	+ \lambda \mu^{-1} a_{3}^{2} a_{4} a_{5} a_{6} \left(1 + a_{1} a_{2} a_{3} a_{4} + a_{1} a_{2} a_{3}^{2} a_{4}^{2} a_{5} a_{6}\right) 
	+ \lambda a_{3}^{2} a_{4}^{2} a_{5} a_{6} \left(1 + a_{1} a_{2} a_{3} a_{6} + a_{1} a_{2}^{2} a_{3}^{2} a_{4} a_{5} a_{6}\right) 
	\\
	+ \mu^{-1} a_{3} \left(1 + a_{3} a_{4} a_{5} a_{6} + a_{1} a_{2} a_{3}^{2} a_{4}^{2} a_{5} a_{6}\right) 
	+ \mu a_{3} a_{4} a_{6} \left(1 + a_{2} a_{3} a_{4} a_{5} + a_{1} a_{2}^{2} a_{3}^{2} a_{4} a_{5} a_{6}\right) 
	+ \lambda^{-1} \mu^{-1} a_{3} 
	\\
	+ \lambda^{-1} \left(1 + a_{3} + a_{2} a_{3}^{2} a_{4} a_{6}\right) 
	+ \lambda^{-1} \mu \left(1 + a_{2} a_{3} a_{4} a_{6} + a_{2} a_{3}^{2} a_{4} a_{6}\right) + \lambda^{-1} \mu^{2} a_{2} a_{3} a_{4} a_{6}\Big)	
\end{multline}

Probably it is more transparent to write restrictions of the dimer partition function \(\mathcal{Z}|_E\) and Hamiltonian \(f|_E\) on the the sides \(E\) of \(N\) (c.f. Lemma~\ref{lem:Painleve by boundary}). In the Fig.~\ref{fi:E6 roots} we present these polynomials in factorized form (omitting the constant factors).
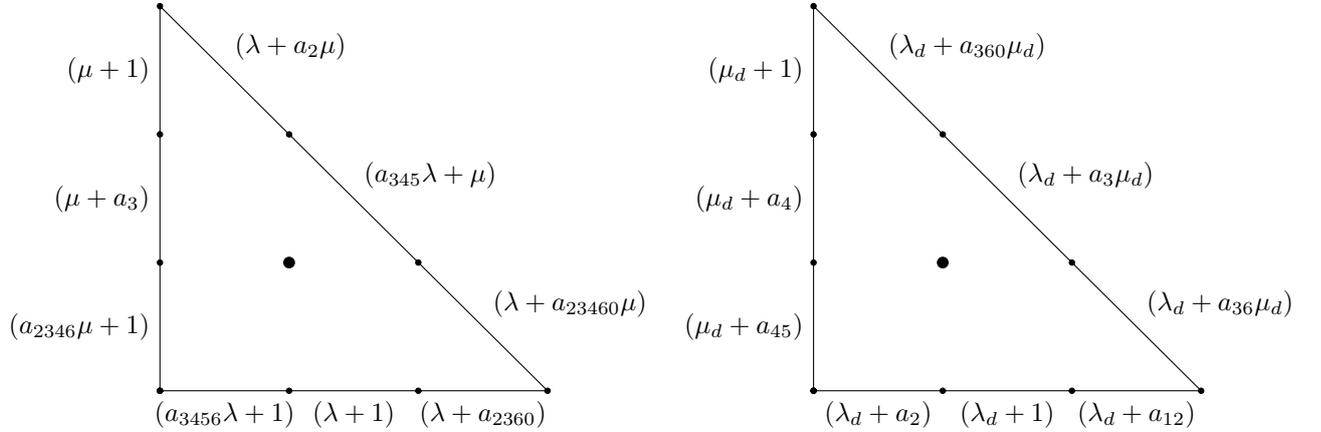
\begin{figure}[h]
	\begin{center}
			\begin{tikzpicture}[font = \small]
				\def\xs{1.7};
				\def\ys{1.7};
					
				\begin{scope}

					\node[left] at (-\xs,1.5*\ys) {$(\mu +1) $};
					\node[left] at (-\xs,0.5*\ys) {$ (\mu+a_3)  $};
					\node[left] at (-\xs,-0.5*\ys) {$(a_{2346}\mu +1)$};
						
					\node[below] at (-0.5*\xs,-1*\ys) {$ (a_{3456}\lambda +1)$};
					\node[below] at (0.5*\xs,-1*\ys) {$(\lambda +1)  $};
					\node[below] at (1.5*\xs,-1*\ys) {$(\lambda +a_{2360})$};
						
					\node[above right] at (-0.5*\xs,1.5*\ys) {$(\lambda+a_2 \mu)  $};

					\node[above right] at (0.5*\xs,0.5*\ys) {$ (a_{345} \lambda +\mu )$};

					\node[above right] at (1.5*\xs,-0.5*\ys) {$ (\lambda +a_{23460}\mu )$};

					\draw[fill] (-1*\xs,-1*\ys) circle (1pt) -- (0,-1*\ys) circle  (1pt) -- (1*\xs,-1*\ys) circle (1pt) -- (2*\xs,-1*\ys) circle (1pt) -- (1*\xs,0) circle (1pt) -- (0,1*\ys) circle (1pt) -- (-1*\xs,2*\ys) circle (1pt) -- (-1*\xs,1*\ys) circle (1pt) -- (-1*\xs,0) circle (1pt) -- (-1*\xs,-1*\ys) circle (1pt) ;
					\draw[fill] (0,0) circle (2pt);	
				\end{scope}		
			
				\begin{scope}[shift={(3*\xs+3.5,0)}]
					
					\node[left] at (-\xs,1.5*\ys) {$(\mu_d+1) $};
					\node[left] at (-\xs,0.5*\ys) {$ (\mu_d+a_4) $};
					\node[left] at (-\xs,-0.5*\ys) {$(\mu_d+a_{45})$};
					
					\node[below] at (-0.5*\xs,-1*\ys) {$ (\lambda_d +a_2)$};
					\node[below] at (0.5*\xs,-1*\ys) {$(\lambda_d +1) $};
					\node[below] at (1.5*\xs,-1*\ys) {$ (\lambda_d +a_{12})$};
%


					\node[above right] at (-0.5*\xs,1.5*\ys) {$(\lambda_d+a_{360}\mu_d) $};
					
					\node[above right] at (0.5*\xs,0.5*\ys) {$ (\lambda_d+a_3 \mu_d)$};
					
					\node[above right] at (1.5*\xs,-0.5*\ys) {$ (\lambda_d+a_{36} \mu_d)$};

%

						\draw[fill] (-1*\xs,-1*\ys) circle (1pt) -- (0,-1*\ys) circle  (1pt) -- (1*\xs,-1*\ys) circle (1pt) -- (2*\xs,-1*\ys) circle (1pt) -- (1*\xs,0) circle (1pt) -- (0,1*\ys) circle (1pt) -- (-1*\xs,2*\ys) circle (1pt) -- (-1*\xs,1*\ys) circle (1pt) -- (-1*\xs,0) circle (1pt) -- (-1*\xs,-1*\ys) circle (1pt) ;
						\draw[fill] (0,0) circle (2pt);	
							
				\end{scope}
			\end{tikzpicture}			
			\caption{\label{fi:E6 roots} The factorized formulas for the restriction to boundary intervals of the dimer partition function and of the Hamiltonian for \(E_6^{(1)}/A_2^{(1)}\) case, up to multiplication by constant.}
		\end{center}
\end{figure}

The generators of simple reflections have the form
\begin{subequations}
	\begin{alignat}{3}
		&s_1|_{\text{face variables}}= (2,3),& \qquad &s_1|_{\text{zigzag variables}}= \mu_{d,3} \mu_{d,6} (6_d,8_d) \mu_{d,6} \mu_{d,3},
		\\
		&s_2|_{\text{face variables}}= (1,2),&\qquad  &s_2|_{\text{zigzag variables}}= \mu_{d,1} \mu_{d,4} (4_d,7_d) \mu_{d,4} \mu_{d,1},
		\\
		&s_3|_{\text{face variables}}= \mu_1 \mu_4 (4,7) \mu_4 \mu_1,& \qquad & s_3|_{\text{zigzag variables}}=(4_d,5_d),
		\\
		&s_4|_{\text{face variables}}= (4,5),& \qquad & s_4|_{\text{zigzag variables}}= \mu_{d,4} \mu_{d,3} (3_d,9_d) \mu_{d,3} \mu_{d,4},
		\\
		&s_5|_{\text{face variables}}= (5,6),& \qquad & s_5|_{\text{zigzag variables}}= \mu_{d,2} \mu_{d,6} (6_d,7_d) \mu_{d,6} \mu_{d,2},
		\\
		&s_6|_{\text{face variables}}= (7,8),& \qquad & s_6|_{\text{zigzag variables}}= \mu_{d,4} \mu_{d,2} (2_d,8_d) \mu_{d,2} \mu_{d,4},
		\\
		&s_0|_{\text{face variables}}= (8,9),& \qquad & s_0|_{\text{zigzag variables}}= \mu_{d,2} \mu_{d,4} (4_d,6_d) \mu_{d,4} \mu_{d,2}, 		
		\\
		&\pi|_{\text{face variables}}=(1,4,7)(2,5,8)(3,6,9),& \qquad &\pi|_{\text{zigzag variables}}= (3,2,1)(9,8,7). 				
	\end{alignat}
\end{subequations}
Their action on $\mathbf{a}$ variables is according to (\ref{eq:actiononroot}) and Dynkin diagram drawn on Fig.~\ref{fi:E6dynkin}.
\begin{figure}[h]
\begin{center}
\begin{tikzpicture}[elt/.style={circle,draw=black!100,thick, inner sep=0pt,minimum size=2mm},scale=1.25]
	\path 	(-2,0) 	node 	(a1) [elt] {}
	(-1,0) 	node 	(a2) [elt] {}
	( 0,0) node  	(a3) [elt] {}
	( 1,0) 	node  	(a4) [elt] {}
    ( 2,0) 	node 	(a5) [elt] {}
	( 0,1)	node 	(a6) [elt] {}
	( 0,2)	node 	(a0) [elt] {};
	\draw [black,line width=1pt ] (a1) -- (a2) -- (a3) -- (a4) -- (a5) (a3) --  (a6) -- (a0);
	\node at ($(a1.south) + (0,-0.2)$) 	{$a_{1}$};
	\node at ($(a2.south) + (0,-0.2)$)  {$a_{2}$};
	\node at ($(a3.south) + (0,-0.2)$)  {$a_{3}$};
	\node at ($(a4.south) + (0,-0.2)$)  {$a_{4}$};	
	\node at ($(a5.south) + (0,-0.2)$)  {$a_{5}$};		
	\node at ($(a6.west) + (-0.2,0)$) 	{$a_{6}$};	
	\node at ($(a0.west) + (-0.2,0)$) 	{$a_{0}$};		
\end{tikzpicture}
\hspace{1cm}
\begin{tikzpicture}
\node [shape=rectangle] {
	\begin{tabular}{|c||c|c|c|c|c|c|c|}
		\hline
		& \(a_0\) & \(a_1\) & \(a_2\) & \(a_3\) & \(a_4\) & \(a_5\) & \(a_6\) \\
		\hline	
		\(s_1\)	& \(a_0\) & \(a_1^{-1}\) & \(a_1 a_2\) & \(a_3\) & \(a_4\) & \(a_5\) & \(a_6\)
		\\
		\hline	
		\(s_2\)	& \(a_0\) & \(a_2 a_1\) & \(a_2^{-1}\) & \(a_2 a_3\) & \(a_4\) & \(a_5\) & \(a_6\)
		\\
		\hline	
		\(s_3\)	& \(a_0\) & \(a_1\) & \(a_3 a_2\) & \(a_3^{-1}\) & \(a_3 a_4\) & \(a_5\) & \(a_3 a_6\)
		\\
		\hline	
		\(s_4\)	& \(a_0\) & \(a_1\) & \(a_2\) & \(a_4 a_3\) & \(a_4^{-1}\) & \(a_4 a_5\) & \(a_6\)
		\\
		\hline	
		\(s_5\)	& \(a_0\) & \(a_1\) & \(a_2\) & \(a_3\) & \(a_5 a_4\) & \(a_5^{-1}\) & \(a_6\)
		\\
		\hline	
		\(s_6\)	& \(a_6 a_0\) & \(a_1\) & \(a_2\) & \(a_6 a_3\) & \(a_4\) & \(a_5\) & \(a_6^{-1}\)
		\\
		\hline
		\(s_0\)	& \(a_0^{-1}\) & \(a_1\) & \(a_2\) & \(a_3\) & \(a_4\) & \(a_5\) & \(a_0 a_6\)
		\\	
        \hline	
	\end{tabular}
}; 
\end{tikzpicture}
\end{center}
\caption{\label{fi:E6dynkin}
		For \(E_6^{(1)}/A_2^{(1)}\) Painlev\'e equation: Dynkin diagram and action on $\mathbf{a}$ variables}
\end{figure}

The formulas for translations $T_1,...,T_6$, $T_{1,d},...,T_{6,d}$ are quite lengthy, so we will not provide them. However they can be obtained from translations in $W^{\mathrm{ae}}(E_6)$ written using generators above. The whole group is \(W \ltimes (P\oplus P_d)\) in agreement with Theorem \ref{th:extended group Painleve}. Here \(W\simeq W(E_6)\) and \(P\simeq P_d\simeq \mathbb{Z}^6\).


\section{\(q\)-Painlev\'e \(E_7^{(1)}\) } \label{sec:E7}

\subsection{Reduction of cluster variety and \(W^{\mathrm{ae}}(E_7)\)}

The pointed Newton polygon \(N\) for $E_7^{(1)}$ case is a rectangle with the sides 2 and 4 drawn on a Fig.~\ref{fig:E7polygon}, see Section \ref{sec:Painleve} and in particular Fig.~\ref{fi:Pain polig} above. We choose consistent  bipartite graph as on Fig.~\ref{fig:E7polygon}, right. The corresponding quiver is drawn on the Fig.~\ref{fig:E7quivers}, left. We denote the corresponding cluster seed by \(\mathsf{s}\).

\begin{Remark} Notice that the bipartite graph which we chose is different from more simple so called the fence-net bipartite graph used \cite[Fig. 6]{MS:2019}. Two graphs are related by sequence of 4-gon mutations $\mu_8 \mu_6 \mu_4 \mu_2$. This makes parallel zigzags to be adjacent, preparing two $\Pi_{4, 2}$ patches for the reductions. 
\end{Remark}

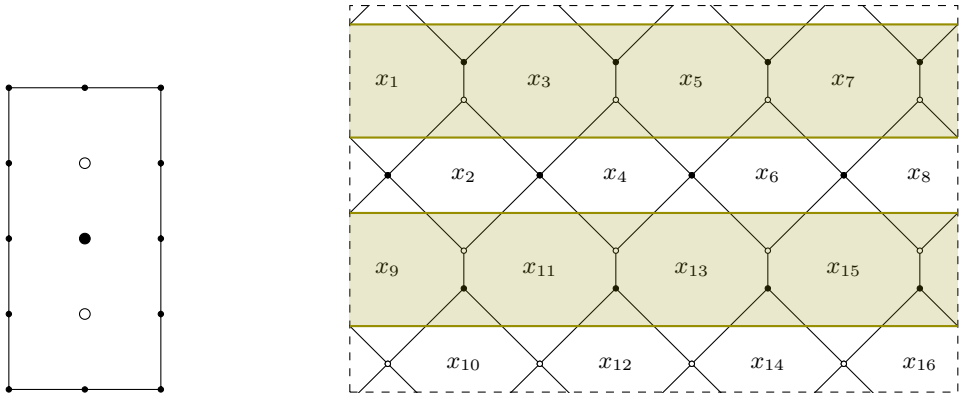
\begin{figure}[h]
	\begin{center}
		\begin{tikzpicture}[scale=1.0, font = \small]
			\draw[fill] (0,0) circle (1pt) -- (1,0) circle (1pt) -- (2,0) circle (1pt) -- (2,1) circle (1pt) -- (2,2) circle (1pt) -- (2,3) circle (1pt) -- (2,4) circle (1pt) -- (1,4) circle (1pt) -- (0,4) circle (1pt) -- (0,3) circle (1pt) -- (0,2) circle (1pt) -- (0,1) circle (1pt) -- (0,0);
			\draw[fill] (1,2) circle (2pt);
			\draw (1,1) circle (2pt);
			\draw (1,3) circle (2pt);
		\end{tikzpicture}
        \qquad\qquad\qquad
        \includegraphics{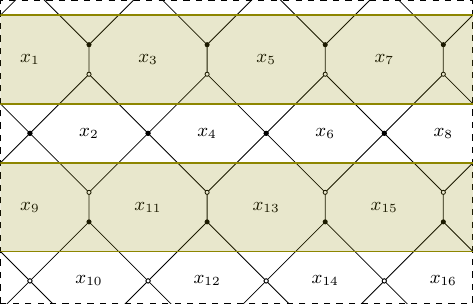}
		\caption{On the left: pointed Newton polygon for the Painlev\'e $E_7^{(1)}$ reduction. On the right: corresponding bipartite graph drawn on the torus. Patches between parallel zigzag paths, where reductions is performed, are indicated by olive strips.}
		\label{fig:E7polygon}
	\end{center}
\end{figure}

\begin{figure}[h]
	\begin{center}
        \includegraphics{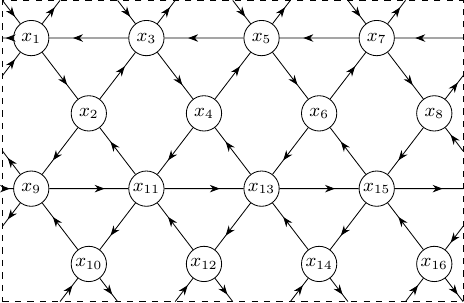}	
        \quad
        \includegraphics{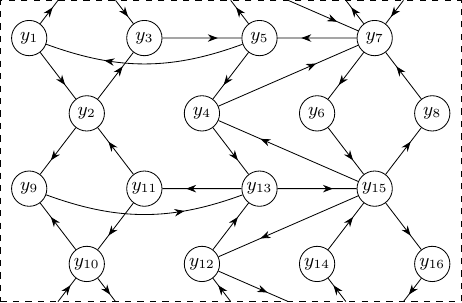}	
		\caption{On the left: quiver for the dimer model of Fig.~\ref{fig:E7polygon}. On the right: quiver after mutation sequence $\mu_{13}\mu_{15}\mu_{5}\mu_{7}$.}
		\label{fig:E7quivers}
	\end{center}
\end{figure}

We impose then two reduction conditions with \(h=2\) corresponding to the top and bottom sides of the rectangle, both of integer length 2. Hence, by the formula \eqref{eq:dim X_red} we get that dimension of the reduced space is
\begin{equation}
	\label{eq:dimE7}
	\dim \mathcal{X}_{\text{red}}=
	2\operatorname{Area}(N) - 2\cdot 3 = 10
\end{equation}
which is consistent with having $8$ Casimir root variables $a_0,...,a_7$ and two-dimensional phase space with coordinates $\lambda_d, \mu_d$ for Painlev\'e \(E_7^{(1)}\) system. 

The reduction conditions come from patches between parallel zigzag paths on the bipartite graph, as it was discussed in Section \ref{ssec:zigzag mut 1}. Followon Lemma~\ref{lem:J implies I} we define 
\begin{equation}
	\label{Gaired}
	\begin{gathered}
		C_1 = x_1 x_3 x_5 x_7, ~~~
		C_2 = x_9 x_{11} x_{13} x_{15},
		\\
		H_1=\mathsf{S}(x_1, x_7, x_5), ~~~
		H_2=\mathsf{S}(x_{11}, x_{13}, x_{15}),
	\end{gathered}
\end{equation}
and define the reduction ideal by
\begin{equation}
	\label{JE7}
	J=(C_1-1,C_2-1,H_1+1,H_2+1).
\end{equation}

The ideal has simpler form in the seed related by a sequence of mutations $\boldsymbol{\mu}=\mu_{13}\mu_{15}\mu_{5}\mu_{7}$. The transformed quiver is drawn on the Fig.~\ref{fig:E7quivers}, right. The following proposition follows from a direct computation

\begin{Proposition}	

(a) In a seed \(\mathsf{t}=\boldsymbol{\mu}(\mathsf{s})\) the ideal (\ref{JE7}) is generated by
\begin{equation}
	\label{eq:JE7 in y}
	J=(y_1 y_3 - 1, y_9 y_{11} - 1, y_1 + 1, y_{11} + 1),
\end{equation}
where $y_i = \boldsymbol{\mu}(x_i)$.

(b) The ideal \(J\) is closed under the Poisson bracket.
\end{Proposition}

The fact that ideal in a seed \(\boldsymbol{\mu}\mathsf{s}\) if given by binomial conditions should be viewed as a particular case of Conjecture~\ref{conj:reductioncord} (a). In the seed \(\boldsymbol{\mu}\mathsf{s}\) the equations that determines ideal \(J\) can be easily solved 
\begin{equation}
	\label{eq:E6constr}
	y_1 = y_3 = y_9 = y_{11} = -1.
\end{equation}
In accordance with Conjecture~\ref{conj:reductioncord} (b), introduce functions \(\mathbf{w}\) by
\begin{equation}
	\label{glu_var}
	\begin{gathered}
	    w_1 = y_7, ~~~ w_2 = y_{10} y_5 y_{13}, ~~~ w_3=y_2 y_5 y_{13}, ~~~ w_4 = y_{14}, ~~~ w_5 = y_6, \\
	    w_6 = y_{15}, ~~~ w_7 = y_{12}, ~~~ w_8 = y_4, ~~~ w_9 = y_{16}, ~~~ w_{10} = y_8.
	\end{gathered}  
\end{equation}
It is straightforward to check that
\begin{Proposition}
	a) The variables \(\mathbf{w}\) are local coordinates on the Hamiltonian reduction with respect to ideal \(J\).
 
	b) The Poisson bracket between \(\{w_i,w_j\}=b^{\mathrm{red}}_{ij}w_iw_j\) is logarithmically constant,  where the matrix \(b^{\mathrm{red}}\) is the adjacency matrix for a quiver in Fig.~\ref{fi:E7quiver}.
 \end{Proposition}
Moreover, we conjecture that  \(\mathbf{w}\) are cluster coordinates on \(\mathcal{X}_{\text{red}}\), which means that any mutations in \(\mathbf{w}\) variables can be lifted to sequences of mutations of \(\mathbf{y}\) variables.

Notice also that the quiver on Fig.~\ref{fi:E7quiver} coincides with the zigzag quiver of pointed Newton polygon from Fig.~\ref{fig:E7polygon} in accordance with Theorem \ref{th:selfduality}.

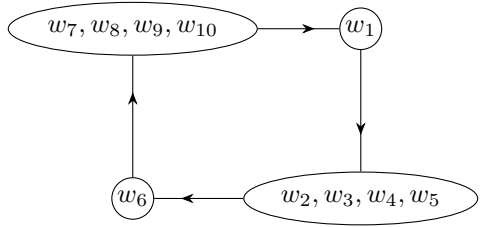
\begin{figure}[h]
	\begin{center}
        \begin{tikzpicture}[scale=1.5, font = \small]
            
            \node[styleNode, ellipse, minimum height=0.7cm] (yright) at(2,0){$w_2, w_3, w_4, w_5$};

            \node[styleNode, ellipse, minimum height=0.7cm] (yleft) at(0,1.5){$w_7, w_8, w_9, w_{10}$};

            \node[styleNode] (ytop) at(2,1.5){$w_1$};

            \node[styleNode] (ybottom) at(0,0){$w_6$};
            
            \draw[styleArrow](ytop) to[] (yright); 		

            \draw[styleArrow](yright) to[] (ybottom); 
            
            \draw[styleArrow](ybottom) to[] (yleft); 		
            
            \draw[styleArrow](yleft) to[] (ytop); 		
		
        \end{tikzpicture}
		\caption{Quiver defining Poisson brackets for \(\mathbf{w}\) variables}
		\label{fi:E7quiver}
	\end{center}
\end{figure}

The independent Casimir functions for the bracket can be chosen as
\begin{equation}
	\label{CasE7}
	\begin{gathered}		
		a_1=\frac{w_2}{w_3},\ \ \ a_2=\frac{w_3}{w_4},\ \ \ a_3=\frac{w_4}{w_5},\ \ \
        a_4=w_5 w_7,
		\\
		a_5=\frac{w_8}{w_7},\ \ \ a_6=\frac{w_9}{w_8},\ \ \ a_7=\frac{w_{10}}{w_9}, \ \ \
        a_0=w_1 w_6.
	\end{gathered}
\end{equation}
Recall that we have a distingwished Casimir function \(q=\prod_{f_i \in F(\Gamma)} x_i\) such that integrable system live on sub-variety \(q=1\). For \(q\neq 1\) this parameter plays a role of shift of difference (\(q\)--Painlev\'e) equation. In variables \(\mathbf{w}\) and \(\mathbf{a}\) it has a form 
%
%
\begin{equation}
	\label{qE7}
    q = \prod_{i=1}^{16} x_i = w_1^2  w_2  w_3  w_4  w_5  w_6^2  w_7  w_8  w_9 w_{10} = a_1 a_2^2 a_3^3  a_4^4 a_5^3 a _6^2 a_7 a_0^2.
\end{equation}
Note that face variables \(\mathbf{w}\) appears in this formula in non-unit powers contrary to definition of \(q\) for seeds construction from consistent dimer models withour reduction. 

The variables $a_i$ will play a role of multiplicative root variables of $E_7^{(1)}$. Namely, the action of Weyl group defined in \eqref{WE7} on \(\mathbf{a}\) agrees with the formula~\eqref{eq:actiononroot}. 

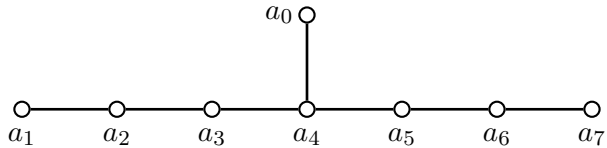
\begin{figure}[h]
	\begin{center}
		\begin{tikzpicture}[elt/.style={circle,draw=black!100,thick, inner sep=0pt,minimum size=2mm},scale=1.25]
			\path 	(-2,0) 	node 	(a1) [elt] {}
			(-1,0) 	node 	(a2) [elt] {}
			( 0,0) node  	(a3) [elt] {}
			( 1,0) 	node  	(a4) [elt] {}
			( 2,0) 	node 	(a5) [elt] {}
			( 3,0)	node 	(a6) [elt] {}
			( 4,0)	node 	(a7) [elt] {}
			( 1,1)	node 	(a0) [elt] {};
			\draw [black,line width=1pt ] (a1) -- (a2) -- (a3) -- (a4) -- (a5) --  (a6) -- (a7) (a4) -- (a0);
			\node at ($(a1.south) + (0,-0.2)$) 	{$a_{1}$};
			\node at ($(a2.south) + (0,-0.2)$)  {$a_{2}$};
			\node at ($(a3.south) + (0,-0.2)$)  {$a_{3}$};
			\node at ($(a4.south) + (0,-0.2)$)  {$a_{4}$};	
			\node at ($(a5.south) + (0,-0.2)$)  {$a_{5}$};		
			\node at ($(a6.south) + (0,-0.2)$) 	{$a_{6}$};	
			\node at ($(a7.south) + (0,-0.2)$) 	{$a_{7}$};	
			\node at ($(a0.west) + (-0.2,0)$) 	{$a_{0}$};		
		\end{tikzpicture}
		\caption{Dynkin diagram of $E_7^{(1)}$}
		\label{fi:DynE7}
	\end{center}
\end{figure}

\begin{Proposition}
	The generators \(s_0,\dots,s_7\) given below satisfy relations of \(W^{\mathrm{ae}}(E_7)\)	
    \begin{equation}
    	\label{WE7}
    	\begin{gathered}	
    	s_1=(2,3),\ \ \
        s_2=(3,4), \ \ \
        s_3=(4,5), \ \ \
        s_4 = \mu_5 (5,7) \mu_5,
    	\\
    	s_5=(7,8), \ \ \
        s_6=(8,9), \ \ \
        s_7=(9,10), \ \ \
        s_0= \mu_6 (1,6) \mu_6.
    	\end{gathered}
    \end{equation}
\end{Proposition}

These formulas for the action are standard (see \cite{Bershtein:2018cluster, Mizuno,Masuda:2023birational}).

\subsection{Hamiltonian and spectral curve}

The dimer partition function before the reduction can be obtained as determinant of Kasteleyn operator with the weights given on Fig.~\ref{fig:E7Kast}. Here we used shorthand notation \(x_{i_1i_2\dots i_k}=x_{i_1}\cdot x_{i_2}\cdot\dots \cdot x_{i_k}\) as in Sec.\ref{sec:examples} , while we separate indices \(i_j\) by commas when they correspond to two-digit numbers.

\begin{figure}[h]
    \centering
    \includegraphics{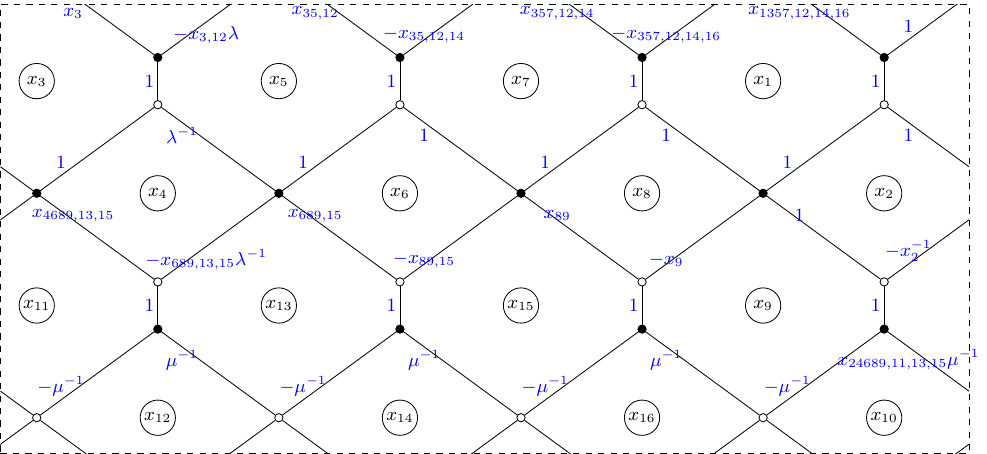}
	\caption{\label{fig:E7Kast}
		Kasteleyn operator for the Newton polygon drawn on Fig.~\ref{fig:E7polygon}.}
\end{figure}

\begin{Remark}
	One can also compute polynomial determining the same spectral curve using approach of \cite{Fock:2016}, see Remark~\ref{rem:FM}. For example, we can take \(M=6\) and consider the loop group element
	\begin{multline}
		L(\lambda)=
		H_4(x_{13})F_4H_4(x_{14})E_4H_3(x_{15})F_3H_3(x_{16})E_3H_2(x_9)F_2H_2(x_{10})E_2 H_1(x_{11})F_1H_1(x_{12})E_1
		\\
		\Lambda(\lambda x_{12345678}^{-1})
		H_4(x_5)E_4H_4(x_6)F_4H_3(x_7)E_3H_3(x_8)F_3H_2(x_1)E_2H_2(x_2)F_2 H_1(x_3)E_1H_1(x_4)F_1\Lambda(\lambda),
	\end{multline}
	which corresponds to the word \(\left(\bar{s}_4s_4\bar{s}_3s_3\bar{s}_2s_2\bar{s}_1s_1\Lambda\right)\left(s_4\bar{s}_4s_3\bar{s}_3s_2\bar{s}_2s_1\bar{s}_1\Lambda\right)\in W^{\mathrm{ae}} (A_5 \times A_5)\). Here \(\Lambda \in W^{\mathrm{ae}} (A_5 \times A_5)\) is an automorhism that cyclically permutes both \(s_1,\dots, s_M\) and \(\bar{s}_1,\dots,\bar{s}_M\), and 
	\(\Lambda(\lambda)\in \widehat{PGL(M)}\) denotes its lift to the coextended loop group given by \cite[eq. (20)]{Fock:2016}. Then \(\mathcal{Z}(\lambda,\mu)\sim\det\left(L(\lambda/\mu)+\mu\right)\). 
\end{Remark}

Now let us impose the reduction conditions given by ideal \(J\) defined in \eqref{JE7}, note also that that \(q=1\). Then it is straightforward to check that the function $\mathcal{Z}_{\mathrm{red}}(\lambda,\mu)=\left.\mathcal{Z}(\lambda,\mu)\right|_{\mathbf{V}(J)}$ satisfies reduction conditions \eqref{eq:polyn mut cond} for two horizontal sides in Newton polygon on Fig.~\ref{fig:E7polygon}.
This is a particular case of Lemma \ref{lem:J implies I}.
More explicitly this means that polynomials corresponding to the top and bottom sides of the polygon on Fig.~\ref{fig:E7polygon} turn into full squares, and polynomials corresponding to the the next-to-top most and next-to-bottom most horizontal lines are divisible by same multipliers.
%
%
%
Geometrically it means that the curve acquires two nodal singularities, and its (smooth) genus drops from $g=3$ (as for generic curve with Newton polygon on Fig.~\ref{fig:E7polygon}) to $g=1$, as it should be for a Painlev\'e system.

Moreover, one can check that the function \(\mathcal{Z}_{\mathrm{red}}(\lambda,\mu)\) after reduction is a Laurent polynomial of the spectral variables \(\lambda,\mu\) and coordinates \emph{after reduction} \(\mathbf{w}\). The explicit formula for the this polynomial is rather involved. However, according to Lemma \ref{lem:Painleve by boundary}, the polynomial corresponding to (reduced) pointed polygon can be recovered, up to common multiplier and constant term, from the knowledge of the roots of its restrictions to the ``boundaries'' of Newton polygon.  The corresponding factors for \(\mathcal{Z}_{\mathrm{red}}(\lambda,\mu)\) are shown on Fig.~\ref{fi:4*2 SC roots}, where we use shorthand notation
\begin{equation}
    \label{eq:avarshorthand}
	\mathbf{a}^{\mathbf{r}}=\prod\nolimits_{i=0}^7 a_i^{r_i}.
\end{equation}
for the products of root variables.

\begin{figure}[h]
	\begin{center}
		\begin{tikzpicture}[scale=0.5, font = \small]
			
			
			\node[left] at (0,0.5) {$\mu +\mathbf{a}^{\mathbf{r}_{10}}$};
			\node[left] at (0,1.5) {$\mu +\mathbf{a}^{\mathbf{r}_9}$};
			\node[left] at (0,2.5) {$\mu +\mathbf{a}^{\mathbf{r}_8}$};
			\node[left] at (0,3.5) {$\mu +\mathbf{a}^{\mathbf{r}_7}$};

			\node[below] at (1,0) {$(\lambda + \mathbf{a}^{\mathbf{r}_1})^2$};

			\node[right] at (2,0.5) {$\mu +\mathbf{a}^{\mathbf{r}_{2}}$};
			\node[right] at (2,1.5) {$\mu +\mathbf{a}^{\mathbf{r}_3}$};
			\node[right] at (2,2.5) {$\mu +\mathbf{a}^{\mathbf{r}_4}$};
			\node[right] at (2,3.5) {$\mu +\mathbf{a}^{\mathbf{r}_5}$};

			\node[above] at (1,4) {$(\lambda + \mathbf{a}^{\mathbf{r}_6})^2$};
   
			\draw[fill] (0,0) circle (1pt) -- (1,0) circle (1pt) -- (2,0) circle (1pt) -- (2,1) circle (1pt) -- (2,2) circle (1pt) -- (2,3) circle (1pt) -- (2,4) circle (1pt) -- (1,4) circle (1pt) -- (0,4) circle (1pt) -- (0,3) circle (1pt) -- (0,2) circle (1pt) -- (0,1) circle (1pt) -- (0,0);
			\draw[fill] (1,2) circle (2pt);
			\draw (1,1) circle (2pt);
			\draw (1,3) circle (2pt);
		
		\node at (15,2.5)	
		{$\begin{tabular}{c|cccccccc}
			& $a_0$ & $a_1$& $a_2$& $a_3$& $a_4$& $a_5$ & $a_6$ & $a_7$ 
			\\
			\hline
			$\mathbf{r}_1$ & 0 & 0 & 0 & 0 & 0 & 0 & 0 & 0 
			\\
			\hline
            $\mathbf{r}_2$ & 0 & 0 & 0 & 0 & 0& 0& 0 & 0 
			\\
			$\mathbf{r}_3$ & 1 & 0& 0& 1& 1 &1 & 0 & 0 			
			\\
			$\mathbf{r}_4$ & 1 & 0& 1& 2& 2 &1 & 0 & 0 			
			\\
			$\mathbf{r}_5$ & 2 & 1& 2& 3& 3 &2 & 1 & 1 			
			\\
			\hline
            $\mathbf{r}_6$ & 1 & 0 & 1 & 1 & 2 & 2 & 1 & 1 
			\\
			\hline
			$\mathbf{r}_7$ & 0 & 0& 0& 1& 0 &0 & 0 & 0 			
			\\
			$\mathbf{r}_8$ & 1 & 0& 0& 1& 1 &0 & 0 & 0
			\\
			$\mathbf{r}_9$ & 1 & 1& 1& 2& 2& 1& 0 & 0 			
			\\
			$\mathbf{r}_{10}$ & 2 & 1& 2& 3& 3& 2& 1 & 0
		\end{tabular}$};
		\end{tikzpicture}		
		\caption{Factors of the restriction of spectral curve polynomial \(\mathcal{Z}_{\mathrm{red}}(\lambda,\mu)\) on the sides  of \(N\)}
		\label{fi:4*2 SC roots}
	\end{center}
\end{figure}

\begin{figure}[h]
	\begin{center}
		\begin{tikzpicture}[scale=0.5, font = \small]
			
			\node[left] at (0,0.5) {$\mu_d +\mathbf{a}^{\mathbf{d}_{10}}$};
			\node[left] at (0,1.5) {$\mu_d +\mathbf{a}^{\mathbf{d}_9}$};
			\node[left] at (0,2.5) {$\mu_d +\mathbf{a}^{\mathbf{d}_8}$};
			\node[left] at (0,3.5) {$\mu_d +\mathbf{a}^{\mathbf{d}_7}$};

			\node[below] at (1,0) {$(\lambda_d + \mathbf{a}^{\mathbf{d}_1})^2$};

			\node[right] at (2,0.5) {$\mu_d +\mathbf{a}^{\mathbf{d}_{2}}$};
			\node[right] at (2,1.5) {$\mu_d +\mathbf{a}^{\mathbf{d}_3}$};
			\node[right] at (2,2.5) {$\mu_d +\mathbf{a}^{\mathbf{d}_4}$};
			\node[right] at (2,3.5) {$\mu_d +\mathbf{a}^{\mathbf{d}_5}$};

			\node[above] at (1,4) {$(\lambda_d + \mathbf{a}^{\mathbf{d}_6})^2$};
   
			\draw[fill] (0,0) circle (1pt) -- (1,0) circle (1pt) -- (2,0) circle (1pt) -- (2,1) circle (1pt) -- (2,2) circle (1pt) -- (2,3) circle (1pt) -- (2,4) circle (1pt) -- (1,4) circle (1pt) -- (0,4) circle (1pt) -- (0,3) circle (1pt) -- (0,2) circle (1pt) -- (0,1) circle (1pt) -- (0,0);
			\draw[fill] (1,2) circle (2pt);
			\draw (1,1) circle (2pt);
			\draw (1,3) circle (2pt);
		
		\node at (13.5,2.5)	
		{$\begin{tabular}{c|cccccccc}
			& $a_0$ & $a_1$& $a_2$& $a_3$& $a_4$& $a_5$ & $a_6$ & $a_7$ 
			\\
			\hline
			$\mathbf{d}_1$ & 0 & 0& 0& 0& 0& 0 &0&0 
			\\
			\hline
            $\mathbf{d}_2$ & 0 & 0 & 0 & 0 & 0 & 0 & 0 & 0 
			\\
			$\mathbf{d}_3$ & 0 & 1 & 0 & 0 & 0 & 0 & 0 & 0 
			\\
			$\mathbf{d}_4$ & 0 & 1 & 1 & 0 & 0 & 0 & 0 & 0
			\\
			$\mathbf{d}_5$ & 0 & 1 & 1 & 1 & 0 & 0 & 0 & 0	
			\\
			\hline
            $\mathbf{d}_6$ & 1 & 0 & 0 & 0 & 0 & 0 & 0 & 0
			\\
			\hline
			$\mathbf{d}_7$ & 0 & 1 & 1 & 1 & 1 & 0 & 0 & 0
			\\
			$\mathbf{d}_8$ & 0 & 1 & 1 & 1 & 1 & 1 & 0 & 0
			\\
			$\mathbf{d}_9$ & 0 & 1 & 1 & 1 & 1 & 1 & 1 & 0
			\\
			$\mathbf{d}_{10}$ & 0 & 1 & 1 & 1 & 1 & 1 & 1 & 1	
		\end{tabular}$};

		\begin{scope}[shift={(25.3,0)}]

			\node[left] at (0,0.5) {$1 +w_{10}$};
			\node[left] at (0,1.5) {$1 +w_9$};
			\node[left] at (0,2.5) {$1 +w_8$};
			\node[left] at (0,3.5) {$1 +w_7$};

			\node[below] at (1,0) {$(1 + w_1)^2$};

			\node[right] at (2,0.5) {$1 +w_2$};
			\node[right] at (2,1.5) {$1 +w_3$};
			\node[right] at (2,2.5) {$1 +w_4$};
			\node[right] at (2,3.5) {$1 +w_5$};

			\node[above] at (1,4) {$(1 + w_6)^2$};
   
			\draw[fill] (0,0) circle (1pt) -- (1,0) circle (1pt) -- (2,0) circle (1pt) -- (2,1) circle (1pt) -- (2,2) circle (1pt) -- (2,3) circle (1pt) -- (2,4) circle (1pt) -- (1,4) circle (1pt) -- (0,4) circle (1pt) -- (0,3) circle (1pt) -- (0,2) circle (1pt) -- (0,1) circle (1pt) -- (0,0);
			\draw[fill] (1,2) circle (2pt);
			\draw (1,1) circle (2pt);
			\draw (1,3) circle (2pt);

        \end{scope}
        
		\end{tikzpicture}		
		\caption{Factors of the restriction of Hamiltonian \(f(\mathrm{w}(\mathbf{a})|\lambda_d,\mu_d)\) on the sides  of \(N\)}
		\label{fi:4*2 Ham roots}
	\end{center}
\end{figure}

The constant term in the dimer partition function \(\mathcal{Z}_{\mathrm{red}}(\lambda,\mu)\) is a Hamiltonian $f_d(\mathbf{w})$ of $E_7^{(1)}$ Painlev\'e system. Picking dual spectral variables to be
\begin{equation}
    \lambda_d = w_1, \; \mu_d = w_2,
\end{equation}
it is straightforward computation to show that in consistency with Theorem \ref{th:selfduality} the Hamiltonian $f(w(a), \lambda_d, \mu_d)$ has a Newton polygon coinciding with the one for the \(\mathcal{Z}_{\mathrm{red}}(\lambda,\mu)\). Moreover, it satisfies the same reduction constraints, so it can be recovered from its restrictions to the boundary sides of Newton polygon up to multiplier as it was proved in Lemma \ref{lem:Painleve by boundary}. The corresponding factors are given on the Fig.~\ref{fi:4*2 Ham roots}, right.

According to Theorem \ref{th:selfduality} there should exist element $\mathrm{w} \in W(E_7)$ acting on root variables, mapping spectral curve to dual spectral curve, i.e. satisfying \(\mathrm{w}^{-1}(\mathbf{a}^{\mathbf{d_i}}) = \mathbf{a}^{r_{i}}\) for any $0\leq i \leq 9$. It is straightforward to check that an element
\begin{equation}
    \mathrm{w} = s_{03243546510243245104653240340532460532435120}
\end{equation}
satisfies this requirement.

The Hamiltonian can be written in a compact way using \(\mathbf{w}\) variables (in proper normalization, see Definition~\ref{def:proper normalization})
\begin{multline}\label{eq:E7 Hamiltonian}
	f(\mathrm{w}(\mathbf{a})|\lambda_d,\mu_d)=
	w_6^{1/2} w_{1}^{3/2} \prod\nolimits_{i=5}^8 w_i^{1/2} \Big(
	\big(\prod\nolimits_{i=2}^5 (1+w_i)-1\big)
	+ \prod\nolimits_{i=2}^5 w_i \, \big((1+w_6)^2-1\big)
	\\
	+w_6^2 \prod\nolimits_{i=2}^5 w_i \big(\prod\nolimits_{i=7}^{10} (1+w_i)-1\big) 
	+w_6^2 \prod\nolimits_{i=1}^4 w_i\prod\nolimits_{i=7}^{10} w_i \, \big((1+w_{1})^2-1\big) 
	\\
	+w_6 \prod\nolimits_{i=1}^4 w_i\, \big(\sum\nolimits_{i=2}^5 w_i^{-1}+\sum\nolimits_{i=7}^{10} w_i\big)+w_{1}^{-1}  \big(\sum _{i=7}^{10} w_i^{-1} +\sum\nolimits _{i=2}^5w_i\big)\Big).
\end{multline}
This formula can be deduced from the simple expressions for the factors corresponding to the boundary sides of Newton polygon in \(\mathbf{w}\) variables given on the Fig.~\ref{fi:4*2 Ham roots}, right. Notice that this is consistent with the general duality idea that the dual zigzag variables (i.e. face variables) have to be the roots of restrictions of Hamiltonian to the boundary sides of its Newton polygon. This is another evidence that \(\mathbf{w}\) variables are the right analog of the face variables for the reduced cluster variety. According to Corollary \ref{cor:polyinvatiance} the Hamiltonian \eqref{eq:E7 Hamiltonian} (at $q=1$) is invariant under the action of the cluster mapping class group $\mathcal{G}_\mathcal{Q}$ for the quiver from Fig.~\ref{fi:E7quiver}.

\begin{Remark}
	Using polynomial mutation one can find other pointed Painlev\'e polygons corresponding to $E_7^{(1)}$. Three of them are given on Fig.~\ref{fi:triE7}. For example, using the formula~\eqref{eq:dim X_red} we have for the left, central, and right polygons correspondingly
	\begin{subequations}
		\begin{align}
			&\dim \mathcal{X}_{\text{red}}=
			2\operatorname{Area}(N_{\text{left}}) - 2\cdot 3 = 10,
			\\
			&\dim \mathcal{X}_{\text{red}}=
			2\operatorname{Area}(N_{\text{central}}) - 8-2\cdot 3 = 10,			
			\\
			&\dim \mathcal{X}_{\text{red}}=
			2\operatorname{Area}(N_{\text{right}}) - 8 = 10.			
		\end{align}
	\end{subequations}
	
%

\begin{figure}[h]
	\begin{center}
		\begin{tikzpicture}[scale=0.75, font = \small]
			
			\begin{scope}[shift={(0.5,0)}]	
				
				\draw[fill] (0,0) circle (1pt) -- (0,1) circle (1pt) -- (0,2) circle (1pt) -- (0,3) circle (1pt) -- (0,4) circle (1pt) -- (1,3) circle (1pt) -- (2,2) circle (1pt) -- (3,1) circle (1pt) -- (4,0) circle (1pt) -- (3,0) circle (1pt) -- (2,0) circle (1pt) -- (1,0) circle (1pt) -- (0,0);
				\draw[fill] (2,1) circle (2pt);
				\draw (1,1) circle (2pt);
				\draw (1,2) circle (2pt);

				
			\end{scope}
			\begin{scope}[shift={(7.5,0)}]	
				
				\draw[fill] (0,0) circle (1pt) -- (0,1) circle (1pt) -- (0,2) circle (1pt) -- (0,3) circle (1pt) -- (0,4) circle (1pt) -- (1,4) circle (1pt) -- (2,4) circle (1pt) -- (3,4) circle (1pt) -- (3,3) circle (1pt) -- (3,2) circle (1pt)-- (3,1) circle (1pt) -- (3,0) circle (1pt) -- (2,0) circle (1pt) -- (1,0) circle (1pt) -- (0,0);

				\draw[fill] (2,1) circle (2pt);
				\draw (1,1) circle (2pt);
				\draw (1,2) circle (2pt);
				\draw (1,3) circle (2pt);
				\draw (2,2) circle (2pt);
				\draw (2,3) circle (2pt);
				
%
%
			\end{scope}

			\begin{scope}[shift={(18,1)}]
				
				\draw[fill] (-4,2) circle (1pt) -- (-2,1) circle (1pt) -- (0,0) circle (1pt) -- (2,-1) circle (1pt) -- (1,-1) circle (1pt) -- (0,-1) circle (1pt) -- (-1,-1) circle (1pt) -- (-2,-1) circle (1pt) 
				-- (-3,-1) circle (1pt) -- (-4,-1) circle (1pt)  -- (-4,0) circle (1pt) -- (-4,1) circle (1pt) -- (-4,2);
				\draw[fill] (-1,0) circle (2pt);
				\draw (-2,0) circle (2pt);
				\draw (-3,0) circle (2pt);
				\draw (-3,1) circle (2pt);

				
			\end{scope}

		\end{tikzpicture}	
%
%
%
%
%
		\caption{Three another pointed polygons for Painlev\'e $E_7^{(1)}$ \label{fi:triE7}}
		
	\end{center}
\end{figure}
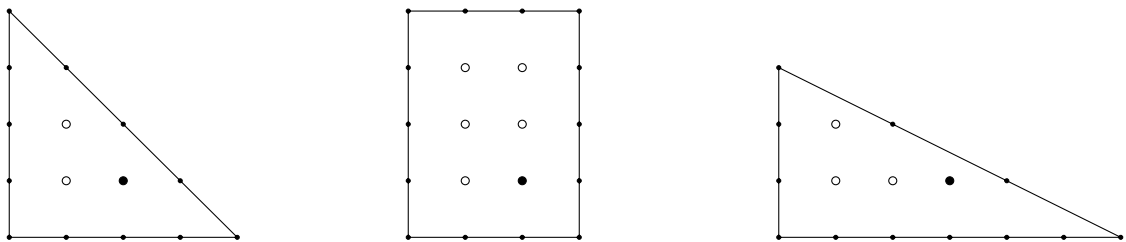
\end{Remark}

\section{\(q\)-Painlev\'e \(E_8^{(1)}\) } \label{sec:E8}
\subsection{Reduction of cluster variety and \(W^{\mathrm{ae}}(E_8)\)}
Consider \(N\) to be a rectangular triangle with the cathetuses 3 and 9, as presented on the left of Fig.~\ref{fi:3*9graph}.  
%
%
%
%

\begin{figure}[h]
	\begin{center}
		\begin{tikzpicture}[scale=0.5, font = \small]
			
			\node[below] at (2,-4) {$B$};  
			\node[below] at (-1,-4) {$C$};                  
			\node[right] at (-1,5) {$A$};                  
			
			\draw[fill] (2,-4) circle (1pt) -- (1,-1) circle (1pt) -- (0,2) circle (1pt) -- (-1,5) circle (1pt) -- (-1,4) circle (1pt) -- (-1,3) circle (1pt) -- (-1,2) circle (1pt) -- 
			(-1,1) circle (1pt) -- (-1,0) circle (1pt) -- (-1,-1) circle (1pt) -- (-1,-2) circle (1pt) 
			-- (-1,-3) circle (1pt) -- (-1,-4) circle (1pt)  -- (0,-4) circle (1pt) -- (1,-4) circle (1pt) -- (2,-4);
			\draw[fill] (0,-1) circle (2pt);
			\draw (0,-2) circle (2pt);
			\draw (0,-3) circle (2pt);
			\draw (1,-3) circle (2pt);
			
			\draw (0,0)  circle (2pt);
			\draw (0,1) circle (2pt);
			\draw (1,-2) circle (2pt);
			
		\node at (17,0) {\includegraphics{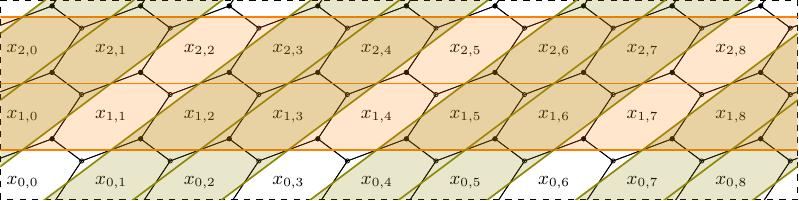}};		
		\end{tikzpicture}				
			\caption{\label{fi:3*9graph}
				On the left: Newton polygon $N$ for $E_8^{(1)}$. On the right: bipartite graph corresponding to \(N\). The graph assumed to be drawn on a torus. The two sets of parallel zigzags are drown in orange and olive, the patches for the reduction are filled.			
			}	
	\end{center}
\end{figure}
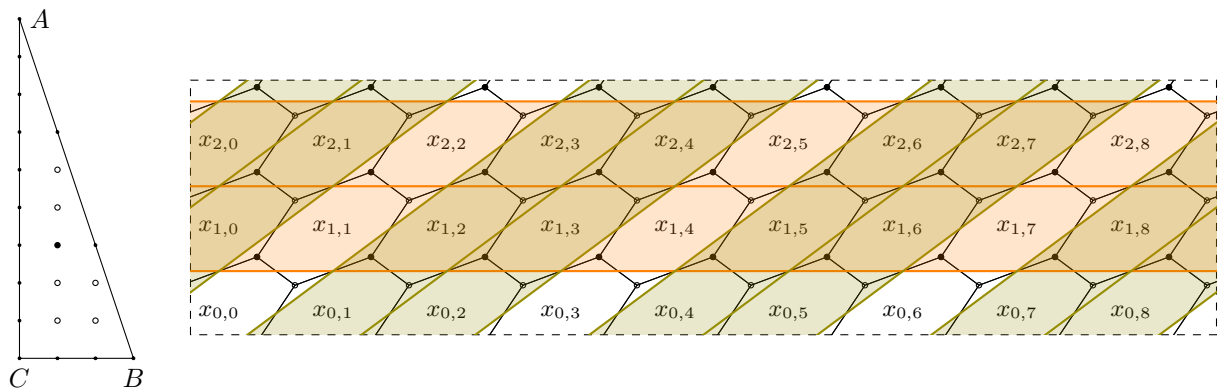

The corresponding bipartite graph \(\Gamma\) is a hexagonal grid  shown on Fig.~\ref{fi:3*9graph} on the right and the corresponding quiver \(\mathcal{Q}\) is shown on the Fig.~\ref{fi:3*9quiver}. We denote the cluster variables by \(x_{i,j}\) and the whole seed by \(\mathsf{s}\).

\begin{figure}[h]
	\begin{center}
		\includegraphics{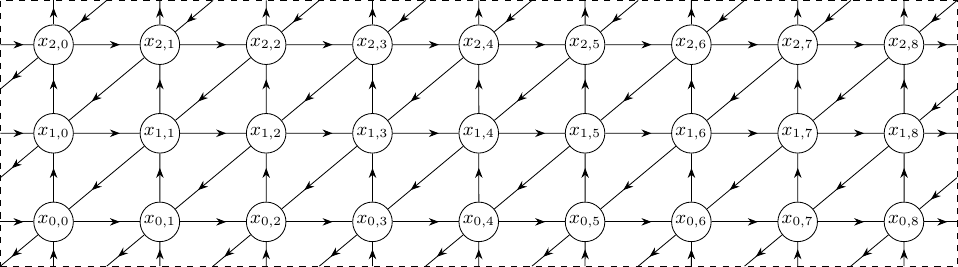}				
		\caption{Quiver for bipartite graph given in Fig. \ref{fi:3*9graph}. The quiver is assumed to be drawn on a torus, the dashed rectangle represents fundamental domain.}
		\label{fi:3*9quiver}
	\end{center}
\end{figure}

We impose reduction conditions with \(h=3\) corresponding to the sides $BC$ and $AB$ of the triangle, both of integer length 3. Hence, by the formula \eqref{eq:dim X_red} we get that expected dimension of the cluster variety is 11, while the expected dimension of the phase space of integrable system is 10. This is agreed with Painlev\'e \(E_8^{(1)}\) expectations, namely we expect to have 9 root variables \(a_0,\dots,a_8\) subject of constrain that \(q=1\) and two dynamical variables \(\lambda_d,\mu_d\). 

One can impose reduction along the patches shown on the right of Fig.~\ref{fi:3*9graph}. These patches are of the form \(\Pi_{9,3}\) and the generators of the corresponding ideals have the form (c.f. formulas~\eqref{eq:C,H},~\eqref{eq:Sevostyanov})
\begin{align*}
	&C_1 = x_{1,0} x_{1,1} x_{1,2} x_{1,3} x_{1,4} x_{1,5} x_{1,6} x_{1,7} x_{1,8},&
	&C_3 = x_{0,2} x_{0,5} x_{0,8} x_{1,0} x_{1,3} x_{1,6} x_{2,1} x_{2,4} x_{2,7}, \\
	&C_2 = x_{2,0} x_{2,1} x_{2,2} x_{2,3} x_{2,4} x_{2,5} x_{2,6} x_{2,7} x_{2,8},& 
	&C_4 = x_{0,1} x_{0,4} x_{0,7} x_{1,2} x_{1,5} x_{1,8} x_{2,0} x_{2,3} x_{2,6}, \\
	&H_1 = \mathsf{S}(x_{1,2},x_{1,3},x_{1,4},x_{1,5},x_{1,6},x_{1,7},x_{1,8},x_{1,0}),& 
	&H_3 = \mathsf{S}(x_{2,1},x_{1,0},x_{0,8},x_{2,7},x_{1,6},x_{0,5},x_{2,4},x_{1,3}),  \\
	&H_2 =\mathsf{S}(x_{2,2},x_{2,3},x_{2,4},x_{2,5},x_{2,6},x_{2,7},x_{2,8},x_{2,0}),& 
	&H_4 = \mathsf{S}(x_{2,0},x_{1,8},x_{0,7},x_{2,6},x_{1,5},x_{0,4},x_{2,3},x_{1,2}) \\
	&H_{12} = \{H_1, H_2\},&\quad &H_{34} = \{H_3, H_4\}.
\end{align*}
and the reduction ideal is generated by 
\begin{equation}
	J=\big(C_1-1,C_2-1,H_1+1,H_2+1,H_{12},C_3-1,C_4-1,H_3+1,H_4+1,H_{34}\big).
\end{equation}

Now we can perform the following sequence of mutations 
\begin{multline}
	\label{muE8}
	\boldsymbol{\mu}=
	\mu_{{1,8}} \mu_{{1,7}} \mu_{{2,6}} \mu_{{2,0}} \mu_{{1,0}} \mu_{{2,8}} \mu_{{1,7}} \mu_{{1,1}} \mu_{{0,8}} \mu_{{2,8}} \mu_{{0,1}} \mu_{{0,7}} \mu_{{1,0}} \mu_{{2,8}} \mu_{{1,6}}  \mu_{{2,7}} \mu_{{0,8}} 
	\\
	\mu_{{1,0}} \mu_{{2,1}} \mu_{{0,4}} \mu_{{1,5}} \mu_{{2,6}} \mu_{{0,7}} \mu_{{1,8}} \mu_{{2,0}} \mu_{{0,1}} \mu_{{2,4}} \mu_{{0,5}} \mu_{{1,6}} \mu_{{2,7}} \mu_{{0,8}} \mu_{{1,0}} \mu_{{2,1}}
\end{multline}
where \(\mu_{i,j}\) stands for mutation in vertex corresponding to \(x_{i,j}\) in the Fig. \ref{fi:3*9quiver}. Let us  denote variables after \(
\boldsymbol{\mu}\) by \(y_{i,j}\). It is straightforward to compute quiver after 
\(\boldsymbol{\mu}\), but it looks to be not very illuminating so we omit it.

\begin{Proposition}	
	(a) In coordinates \(y_{i,j}\) the ideal \(J\) is generated by 
	\begin{multline}\label{eq:E8 J in y}
		J=\big(1+y_{0,2}, 1+y_{1,3}, 1+y_{2,3}, 1+y_{0,5}, 1+y_{2,5},1+y_{2,1}, \\ 1-y_{1,2} y_{2,4}, 1-y_{0,7} y_{2,7}, 1-y_{1,4} y_{1,6}, 1+y_{0,4} y_{1,6} y_{2,2}\big)
	\end{multline}
	
	(b) The ideal \(J\) is closed under the Poisson bracket.
\end{Proposition} 
This proposition follows from straightforward computation. For part (b) one can use both original seed \(\mathsf{s}\) or \(\boldsymbol{\mu}(\mathsf{s})\), in the latter case it is easily seen, that generators~\eqref{eq:E8 J in y} do Poisson commute.

The part (a) of the proposition explains the purpose of the transformation \(\boldsymbol{\mu}\), namely in new variables the ideal (\ref{eq:E8 J in y}) is generated by the Poisson commuting binomials, (c.f. Lemma~\ref{lem:patch cluster chart}).

Introduce now new variables
\begin{multline}
	\label{wE8}
	w_1=-y_{0,3} y_{0,4} y_{1,5} y_{2,2} y_{2,4} y_{2,6},\;\; w_2=-y_{1,5} y_{2,4} y_{2,6} \;\;	w_3=y_{0,0}, \;\; w_4=y_{0,1}, \;\;  w_5=y_{0,6},
	\\ 
	 w_6=y_{1,7},\;\; w_7=y_{1,8}\;\; w_{8}=y_{1,1},\;\; w_{9}=y_{2,6},\;\; w_{10}=y_{1,0},,\;\; w_{11}=-y_{0,8} y_{2,0},.
\end{multline}

\begin{Proposition}
	(a) The variables \(w_1,\dots,w_{11}\) are local coordinates on the Hamiltonian reduction with respect to ideal \(J\).
	
	(b) The Poisson bracket between \(\{w_i,w_j\}=b^{\mathrm{red}}_{ij}w_iw_j\) is logarithmically constant,  where the matrix \(b^{\mathrm{red}}\) is the adjacency matrix for a quiver in Fig.~\ref{fi:E8 reduction quiver}.
	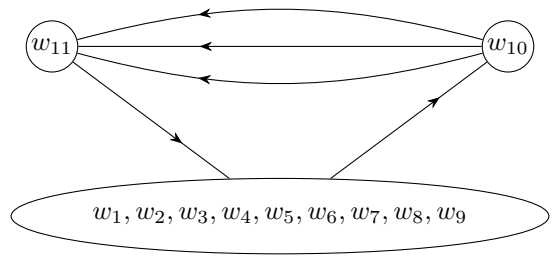
\begin{figure}[h]
		\begin{center}
			\begin{tikzpicture}[scale=1.5, font = \small]
				
				\node[styleNode, ellipse, minimum height=1cm] (ybottom) at(0,0){$w_1, w_2, w_3, w_4, w_5, w_6, w_7, w_{8}, w_{9}$};

				\node[styleNode] (yright) at(2,1.5){$w_{10}$};

				\node[styleNode] (yleft) at(-2,1.5){$w_{11}$};
				
				\draw[styleArrow](yleft) to[] (ybottom); 		
				\draw[styleArrow](ybottom) to[] (yright); 		
				\draw[styleArrow](yright) to[] (yleft); 		
				\draw[styleArrow](yright) to[bend right=15] (yleft); 		
				\draw[styleArrow](yright) to[bend left=15] (yleft); 		
			\end{tikzpicture}
		\caption{Quiver for the Poisson bracket after reduction in variables \eqref{eq:E8 J in y}}
		\label{fi:E8 reduction quiver}	
		\end{center}		
	\end{figure}	
\end{Proposition}
To prove (a) one have to check, first, that variables \(\mathbf{w}\) Poisson commute with generators of the ideal \eqref{eq:E8 J in y}. Then we note that clearly \(w_1,\dots, w_{11}\) are algebraically independent and recall that \(\dim \mathcal{X}_{\text{red}}=11\).

We conjecture that \(w_1,\dots, w_{11}\) are natural cluster coordinates on \(\mathcal{X}_{\text{red}}\). This means in particular, that mutations with respect to these variables should be a descent of some mutation sequences on $\mathcal{X}$. For variables \(w_4,\dots,w_{11}\) this is clear since they are cluster variables for seed \(\boldsymbol{\mu}(\mathsf{s})\). For variables \(w_1,w_2,w_{11}\) it becomes quite nontrivial, we checked this for \(w_{11}\) using another sequence \(\boldsymbol{\mu}'\). Additional argument for our choice of cluster variables comes from the study of the Hamiltonian below.

Note also that the quiver on Fig.~\ref{fi:E8 reduction quiver} coincides with zigzag quiver constructed from the polygon on Fig.~\ref{fi:3*9graph}, in agreement with Theorem~\ref{th:selfduality}.

The Poisson bracket on \(\mathcal{X}_{\text{red}}\) defined by the quiver on Fig.~\ref{fi:E8 reduction quiver} has rank 2. Therefore it has 9 linearly independent Casimir functions. We choose them as follows 
\begin{multline}\label{eq:E8 root}
	a_0=w_{5} w_{6} w_{8} w_{10} w_{11},\; a_1=\frac{w_5}{w_{8}},\; a_2=\frac{w_6}{w_5},\; a_3=\frac{w_7}{w_6},\\ a_4=\frac{w_2}{w_7},\; a_5=\frac{w_1}{w_2},\; a_6=\frac{w_{9}}{w_1},\; a_7=\frac{w_4}{w_{9}},\; a_8=\frac{w_3}{w_4}.
\end{multline}
In such definition the variables \(\mathbf{a}\) serve as root variables for the root system \(E_8^{(1)}\), see Fig.~\ref{fi:E8 numbering} for the labeling of the nodes on the Dynkin diagram.

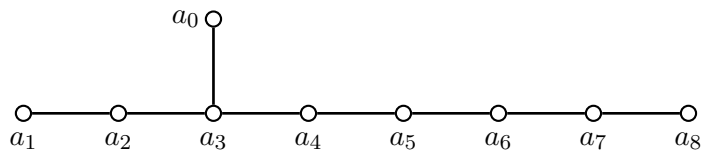
\begin{figure}[h]
	\begin{center}
		\begin{tikzpicture}[elt/.style={circle,draw=black!100,thick, inner sep=0pt,minimum size=2mm},scale=1.25]
			\path 	(-2,0) 	node 	(a1) [elt] {}
			(-1,0) 	node 	(a2) [elt] {}
			( 0,0) node  	(a3) [elt] {}
			( 1,0) 	node  	(a4) [elt] {}
			( 2,0) 	node 	(a5) [elt] {}
			( 3,0)	node 	(a6) [elt] {}
			( 4,0)	node 	(a7) [elt] {}
			( 5,0)	node 	(a8) [elt] {}
			( 0,1)	node 	(a0) [elt] {};
			\draw [black,line width=1pt ] (a1) -- (a2) -- (a3) -- (a4) -- (a5) --  (a6) -- (a7) --(a8) (a3) -- (a0);
			\node at ($(a1.south) + (0,-0.2)$) 	{$a_{1}$};
			\node at ($(a2.south) + (0,-0.2)$)  {$a_{2}$};
			\node at ($(a3.south) + (0,-0.2)$)  {$a_{3}$};
			\node at ($(a4.south) + (0,-0.2)$)  {$a_{4}$};	
			\node at ($(a5.south) + (0,-0.2)$)  {$a_{5}$};		
			\node at ($(a6.south) + (0,-0.2)$) 	{$a_{6}$};	
			\node at ($(a7.south) + (0,-0.2)$) 	{$a_{7}$};	
			\node at ($(a8.south) + (0,-0.2)$) 	{$a_{8}$};	
			\node at ($(a0.west) + (-0.2,0)$) 	{$a_{0}$};		
		\end{tikzpicture}
		\caption{Root numbering for $E_8^{(1)}$ Dynkin diagram}
		\label{fi:E8 numbering}
	\end{center}
\end{figure}

The Casimir function \(q\) (which corresponds to imaginary root and is responsible for deatonomization) has the following expressions
\begin{equation}\label{eq:q E8}
	q=\prod_{i=0}^2\prod_{j=0}^{8}x_{i,j} = w_1 w_2 w_3 w_4 w_5 w_6 w_7 w_8 w_9 w_{10}^3  w_{11}^3=a_0^3 a_1^2 a_2^4 a_3^6 a_4^5 a_5^4 a_6^3 a_7^2 a_8^1.
\end{equation}
The exponents for \(a_i\) in \eqref{eq:q E8} coincides with exponents of roots for \(E_8^{(1)}\). 

The Weyl group \(W^{\mathrm{ae}}(E_8)\) can be realized using mutations and transpositions (the formulas are standard, see \cite{Bershtein:2018cluster, Mizuno,Masuda:2023birational})

\begin{Proposition}
	The generators \(s_0,\dots,s_8\) given below satisfy relations of \(W^{\mathrm{ae}}(E_8)\)	
	\begin{equation}
		\begin{aligned}
			s_1=(5,8),\;\; s_2=(5,6),\;\; s_3=(6,7),\;\;  s_4=(2,7),\;\; s_5=(1,2),\;\; s_6=(1,9),
			\\
			s_7=(4,9),\;\; s_8=(3,4),\;\; s_0=(10,11)\mu_5\mu_{6}\mu_8\mu_{10}\mu_{11}\mu_5\mu_{6}\mu_8.
		\end{aligned}
	\end{equation}	
\end{Proposition}

Let us comment to formula for \(s_0\). It is easy to see from Fig.~\ref{fi:E8 reduction quiver} that the mutation sequence \(\mu_5\mu_{6}\mu_8\) removes three edges between vertices 10 and 11. The new variables \(w_{10}'\) and \(w_{11}'\) commute and their product is Casimir. Hence \((10,11)\mu_{10}\mu_{11}\) is an automorphism of the obtained quiver.

\subsection{Hamiltonian and spectral curve}

In order to construct spectral curve and Hamiltonian we will use approach of \cite{Fock:2016}, see Remark~\ref{rem:FM}. Let as take \(M=9\) and consider the element
\begin{multline}
	L(\lambda)=H_0(x_{0,0}) E_0  H_3(x_{1,3}) E_3 H_6(x_{2,6})E_6 H_2(x_{1,2}) E_2 H_5(x_{2,5})E_5 H_8(x_{0,8}) E_8 	
	\\
	H_1(x_{1,1}) E_1 H_4(x_{2,4}) E_4 H_7(x_{0,7}) E_7 H_0(x_{1,0}) E_0  H_3(x_{2,3})E_3 H_6(x_{0,6})E_6 H_2(x_{2,2})E_2 	
	\\ 
	H_5(x_{0,5}) E_5 H_8(x_{1,8}) E_8 H_1(x_{2,1}) E_1  H_4(x_{0,4})E_4 H_7(x_{1,7}) E_7 H_0(x_{0,2})E_0 H_3(x_{0,3})E_3  
	\\ H_6(x_{1.6}) E_6 H_2(x_{0,2}) E_2  H_5(x_{1,5}) E_5 H_8(x_{2,8}) E_8 H_1(x_{0,1}) E_1 H_4(x_{1,4}) E_4 H_7(x_{2,7}) E_7, 
\end{multline}
which correspond to the word \((s_{036258147})^3\in W^{\mathrm{ae}} (A_8 \times A_8)\). Then the partition function can be defined as
\begin{equation}
	\mathcal{Z}(\lambda,\mu)=\lambda^{-1}\mu^{-1}\Big(\det\big(g(\lambda)+\mu\big)\Big|_{\displaystyle \lambda \mapsto   x_{0,0}^{-1} x_{1,4} x_{1,5} x_{1,6} x_{1,7} x_{1,8} x_{2,7} x_{2,8}\lambda \mu^{-1}  }\Big).
\end{equation}
Note that here we made \(SL(2,\mathbb{Z})\) transformation on \(\lambda, \mu\) in order to obtain the Newton polygon given in Fig.~\ref{fi:3*9graph}. and also some rescaling in order to make \(\lambda,\mu\) be Casimir functions. 

Now impose reduction conditions given by ideal \(J\), note also that that \(q=1\).
Then it is straightforward to check (and also follows from Lemma~\ref{lem:J implies I}) that the function $\mathcal{Z}_{\mathrm{red}}(\lambda,\mu)=\mathcal{Z}(\lambda,\mu)|_{\mathbf{V}(J)}$ would satisfy conditions \eqref{eq:polyn mut cond} for the sides $AB$ and $BC$. Then up to constant term and overall monomial factor the Laurent polynomial $\mathcal{Z}_{\mathrm{red}}(\lambda,\mu)$ is determined by \(\mathcal{Z}_{\mathrm{red}}|_E\) for the sides \(E\) (c.f. Lemma~\ref{lem:Painleve by boundary}). In the Fig.~\ref{fi:3*9 SC roots} we presented these polynomials in factorized form. Note that the factors are expressed in terms of root variables \(\mathbf{a}\) defined in formula \eqref{eq:E8 root}, where we also use shorthand notation similar to \ref{eq:avarshorthand}. 

\begin{figure}[h]
	\begin{center}
		\begin{tikzpicture}[scale=0.5, font = \small]
			
			
			\node[left] at (-1,-3.5) {$\mu +\mathbf{a}^{\mathbf{r}_9}$};
			\node[left] at (-1,-2.5) {$\mu +\mathbf{a}^{\mathbf{r}_8}$};
			\node[left] at (-1,-1.5) {$\mu +\mathbf{a}^{\mathbf{r}_7}$};
			\node[left] at (-1,-0.5) {$\mu +\mathbf{a}^{\mathbf{r}_6}$};
			\node[left] at (-1,0.5) {$\mu +\mathbf{a}^{\mathbf{r}_5}$};
			\node[left] at (-1,1.5) {$\mu +\mathbf{a}^{\mathbf{r}_4}$};
			\node[left] at (-1,2.5) {$\mu +\mathbf{a}^{\mathbf{r}_3}$};
			\node[left] at (-1,3.5) {$\mu +\mathbf{a}^{\mathbf{r}_2}$};
			\node[left] at (-1,4.5) {$\mu +\mathbf{a}^{\mathbf{r}_1}$};
			
			\node[below] at (0.5,-4) {$(1-\lambda)^3$};
			
			\node[above right] at (0,2) {$(\mu^3-\mathbf{a}^{\mathbf{r}_{0}} \lambda)^3$};
			
			\draw[fill] (2,-4) circle (1pt) -- (1,-1) circle (1pt) -- (0,2) circle (1pt) -- (-1,5) circle (1pt) -- (-1,4) circle (1pt) -- (-1,3) circle (1pt) -- (-1,2) circle (1pt) -- 
			(-1,1) circle (1pt) -- (-1,0) circle (1pt) -- (-1,-1) circle (1pt) -- (-1,-2) circle (1pt) 
			-- (-1,-3) circle (1pt) -- (-1,-4) circle (1pt)  -- (0,-4) circle (1pt) -- (1,-4) circle (1pt) -- (2,-4);
			\draw[fill] (0,-1) circle (2pt);
			\draw (0,-2) circle (2pt);
			\draw (0,-3) circle (2pt);
			\draw (1,-3) circle (2pt);
			
			\draw (0,0)  circle (2pt);
			\draw (0,1) circle (2pt);
			\draw (1,-2) circle (2pt);

		\node at (15,0.5)	
		{$\begin{tabular}{c|ccccccccc}
			& $a_0$ & $a_1$& $a_2$& $a_3$& $a_4$& $a_5$ & $a_6$ & $a_7$ & $a_8$ 
			\\
			\hline
			$\mathbf{r}_1$ & 0 & 0& 0& 0& 0& 0 &0&0&0 
			\\
			$\mathbf{r}_2$ & 1 & 0 & 0 & 1 & 1 & 1 & 1 & 1 & 0 
			\\
			$\mathbf{r}_3$ & 1 & 0& 1& 1& 1& 1 &1 & 1 & 0 			
			\\
			$\mathbf{r}_4$ & 1 & 0& 1& 2& 1& 1 &1 & 1 & 0 			
			\\
			$\mathbf{r}_5$ & 2 & 1& 2& 3& 2& 1 &1 & 1 & 0 			
			\\
			$\mathbf{r}_6$ & 2 & 1& 2& 3& 2& 2 &1 & 1 & 0 			
			\\
			$\mathbf{r}_7$ & 2 & 1& 2& 3& 2& 2 &2 & 1 & 0 			
			\\
			$\mathbf{r}_8$ & 3 & 1& 3& 5& 4& 3 &2 & 1 & 0 			
			\\
			$\mathbf{r}_9$ & 3 & 2& 4& 6& 5& 4 &3 & 2 & 0 			
			\\
			\hline
			$\mathbf{r}_{0}$ & 5 & 2& 5& 8& 6& 5 &4 & 3 & 0 			
		\end{tabular}$};
		\end{tikzpicture}		
		\caption{Factors of the restriction of spectral curve polynomial \(\mathcal{Z}_{\mathrm{red}}(\lambda,\mu)\) on the sides  of \(N\)}
		\label{fi:3*9 SC roots}
	\end{center}
\end{figure}
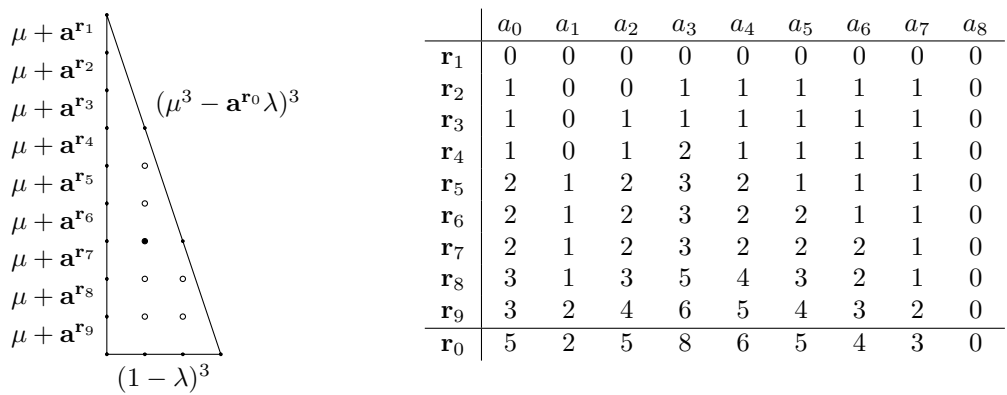

The constant term of \(\mathcal{Z}_{\mathrm{red}}(\lambda,\mu)\) is a Hamiltonian \(f_d(\mathbf{w})\) Let us introduce dual spectral variables as 
\begin{equation}
	\lambda_d=-w_{10},\; \mu_d=w_7^{-1}.
\end{equation}
Then the Hamiltonian is a function on dual spectral parameters and root variables. We denote it by \(f(\mathrm{w}(a)|\lambda_d,\mu_d)\) in view of Theorem~\ref{th:selfduality}. It is checked by straightforward (but rather cumbersome) computation that the Newton polygon of the Hamiltonian coincides with \(N\) given in Fig.~\ref{fi:3*9graph}, Moreover, the Hamiltonian \(f(\mathrm{w}(a)|\lambda_d,\mu_d)\) also satisfies reduction conditions along the sides \(AB\) and \(BC\). Hence it is also determined by the restriction of  \(f(\mathrm{w}(a)|\lambda_d,\mu_d)\) to the sides of \(N\), we give them on the Fig.~\ref{fi:3*9 Ham roots}. 

\begin{figure}[h]
	\begin{center}
		\begin{tikzpicture}[scale=0.5, font = \small]
			
			
			\node[left] at (-1,-3.5) {$\mu_d +\mathbf{a}^{\mathbf{d}_9}$};
			\node[left] at (-1,-2.5) {$\mu_d +\mathbf{a}^{\mathbf{d}_8}$};
			\node[left] at (-1,-1.5) {$\mu_d +\mathbf{a}^{\mathbf{d}_7}$};
			\node[left] at (-1,-0.5) {$\mu_d +\mathbf{a}^{\mathbf{d}_6}$};
			\node[left] at (-1,0.5) {$\mu_d +\mathbf{a}^{\mathbf{d}_5}$};
			\node[left] at (-1,1.5) {$\mu_d +\mathbf{a}^{\mathbf{d}_4}$};
			\node[left] at (-1,2.5) {$\mu_d +\mathbf{a}^{\mathbf{d}_3}$};
			\node[left] at (-1,3.5) {$\mu_d +\mathbf{a}^{\mathbf{d}_2}$};
			\node[left] at (-1,4.5) {$\mu_d +\mathbf{a}^{\mathbf{d}_1}$};
			
			\node[below] at (0.5,-4) {$(1-\lambda_d)^3$};
			
			\node[above right] at (0,2) {$(\mu_d^3-\mathbf{a}^{\mathbf{d}_{0}} \lambda_d)^3$};
			
			\draw[fill] (2,-4) circle (1pt) -- (1,-1) circle (1pt) -- (0,2) circle (1pt) -- (-1,5) circle (1pt) -- (-1,4) circle (1pt) -- (-1,3) circle (1pt) -- (-1,2) circle (1pt) -- 
			(-1,1) circle (1pt) -- (-1,0) circle (1pt) -- (-1,-1) circle (1pt) -- (-1,-2) circle (1pt) 
			-- (-1,-3) circle (1pt) -- (-1,-4) circle (1pt)  -- (0,-4) circle (1pt) -- (1,-4) circle (1pt) -- (2,-4);
			\draw[fill] (0,-1) circle (2pt);
			\draw (0,-2) circle (2pt);
			\draw (0,-3) circle (2pt);
			\draw (1,-3) circle (2pt);
			
			\draw (0,0)  circle (2pt);
			\draw (0,1) circle (2pt);
			\draw (1,-2) circle (2pt);

			\node at (13,0.5)	
			{
				$\begin{tabular}{c|ccccccccc}
					& $a_0$ & $a_1$& $a_2$& $a_3$& $a_4$& $a_5$ & $a_6$ & $a_7$ & $a_8$ 
					\\
					\hline
					$\mathbf{d}_1$ & 0& -1& -1& -1& 0& 0 &0&0&0 
					\\
					$\mathbf{d}_2$ & 0 & 0 & -1 & -1 & 0 & 0 & 0 & 0 & 0 
					\\
					$\mathbf{d}_3$ & 0 & 0& 0& -1 & 0& 0 &0 & 0 & 0 			
					\\
					$\mathbf{d}_4$ & 0 & 0& 0& 0& 0& 0 &0 & 0 & 0 			
					\\
					$\mathbf{d}_5$ & 0 & 0& 0& 0& 1& 0 &0 & 0 & 0 			
					\\
					$\mathbf{d}_6$ & 0 & 0& 0& 0& 1& 1 &0 & 0 & 0 			
					\\
					$\mathbf{d}_7$ & 0 & 0& 0& 0& 1& 1 &1 & 0 & 0 			
					\\
					$\mathbf{d}_8$ & 0 & 0& 0& 0& 1& 1 &1 & 1 & 0 			
					\\
					$\mathbf{d}_9$ & 0 & 0& 0& 0& 1& 1 &1 & 1 & 1 			
					\\
					\hline
					$\mathbf{d}_{0}$ & -1 & -1& -2& -3& 0& 0 &0 & 0 & 0 			
				\end{tabular}$
			};
			\begin{scope}[shift={(25.7,0)}]
				
				
				\node[left] at (-1,-3.5) {$1+w_1$};
				\node[left] at (-1,-2.5) {$1+w_2$};
				\node[left] at (-1,-1.5) {$1+w_3$};
				\node[left] at (-1,-0.5) {$1+w_4$};
				\node[left] at (-1,0.5) {$1+w_5$};
				\node[left] at (-1,1.5) {$1+w_6$};
				\node[left] at (-1,2.5) {$1+w_7$};
				\node[left] at (-1,3.5) {$1+w_{8}$};
				\node[left] at (-1,4.5) {$1+w_{9}$};
				
				\node[below] at (0.5,-4) {$(1+w_{10})^3$};
				
				\node[above right] at (0,2) {$(1+w_{11})^3$};
				
				\draw[fill] (2,-4) circle (1pt) -- (1,-1) circle (1pt) -- (0,2) circle (1pt) -- (-1,5) circle (1pt) -- (-1,4) circle (1pt) -- (-1,3) circle (1pt) -- (-1,2) circle (1pt) -- 
				(-1,1) circle (1pt) -- (-1,0) circle (1pt) -- (-1,-1) circle (1pt) -- (-1,-2) circle (1pt) 
				-- (-1,-3) circle (1pt) -- (-1,-4) circle (1pt)  -- (0,-4) circle (1pt) -- (1,-4) circle (1pt) -- (2,-4);
				\draw[fill] (0,-1) circle (2pt);
				\draw (0,-2) circle (2pt);
				\draw (0,-3) circle (2pt);
				\draw (1,-3) circle (2pt);
				
				\draw (0,0)  circle (2pt);
				\draw (0,1) circle (2pt);
				\draw (1,-2) circle (2pt);
				
		\end{scope}
			
		\end{tikzpicture}		
		\caption{Factors of the restriction of Hamiltonian \(f(\mathrm{w}(\mathbf{a})|\lambda_d,\mu_d)\) on the sides  of \(N\)}
		\label{fi:3*9 Ham roots}
		
	\end{center}
\end{figure}
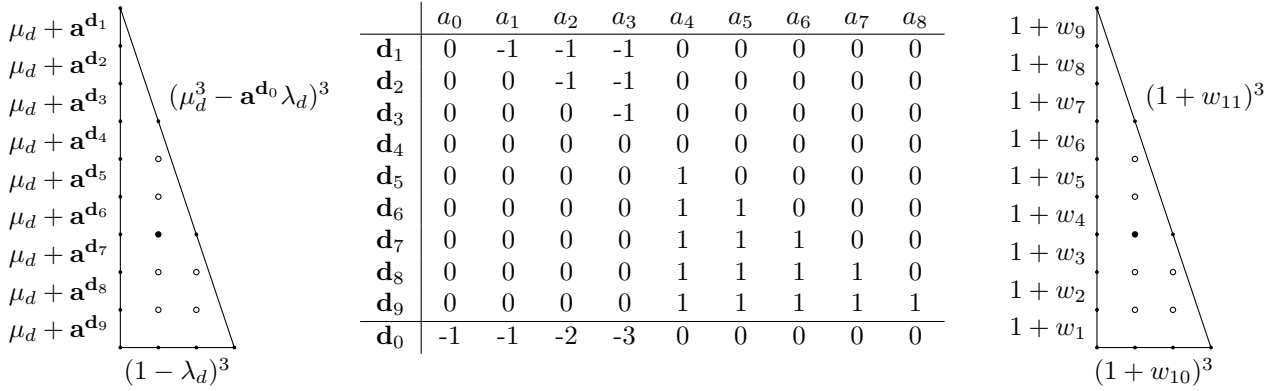
Note also that the factors of \(f(\mathrm{w}(a)|\lambda_d,\mu_d)\) have very simple form in terms of variables on the reduction \(\mathbf{w}\), they are given on the right of Fig.~\ref{fi:3*9 Ham roots}. Recall that in the case without reduction such factors should correspond to zigzags of the dual graph \(\Gamma^D\), i.e. face variables. Taking into account Theorem~\ref{th:zigzags boundary} above and self-duality Theorem~\ref{th:selfduality} this is an additional argument in favor our choice cluster coordinates \(\mathbf{w}\) on the reduction.

As was shown in Lemma~\ref{lem:Painleve by boundary} the function \(f_d(\mathbf{a}|\lambda_d,\mu_d)\) is determined by information above (restrictions on the boundary of \(N\) and reduction conditions). Explicitly it has the form 
\begin{multline}\label{eq:E8 Hamiltonian}
	f(\mathrm{w}(a)|\lambda_d,\mu_d)=w_{10} w_{11}^2 \Big( \big(\prod\nolimits_{i=1}^9 (1+w_i)-1\big)  +  \big((1+w_{10})^3-1\big)\prod\nolimits_{i=1}^9 w_i 
	+   \big((1+w_{11})^3-1\big)w_{10}^3\prod\nolimits_{i=1}^9 w_i
	\\ 
	+(w_{11}^{-2}+2w_{11}^{-1}) \sum\nolimits_{i=1}^9 w_i+
	(w_{10}^2+2 w_{10}) \prod\nolimits_{i=1}^9 w_i   \sum\nolimits_{i=1}^9 {w_i}^{-1}
	+w_{11}^{-1}\sum_{1\leq i<j\le 9} w_i w_j
	\\ 
	+w_{10} \prod\nolimits_{i=1}^9 w_i \sum_{1\leq i<j\le 9} w_i^{-1} w_j^{-1}
	+w_{11} w_{10}^2 \prod\nolimits_{i=1}^9 w_i\sum\nolimits_{i=1}^9 w_i^{-1}+w_{10}^{-1}w_{11}^{-2}\sum\nolimits_{i=1}^9 w_i \Big). 
\end{multline}

In order to show self-duality Theorem~\ref{th:selfduality} in this case it remains to find an element \( \mathrm{w}^{-1} \in W(E_8)\) which transforms \(f(\mathrm{w}(\mathbf{a})|\lambda_d,\mu_d)\) to \(f(\mathbf{a}|\lambda,\mu)=(\mathcal{Z}_{\mathrm{red}}(\lambda,\mu)-\text{constant term})\). Since these polynomials are determined by their behavior on the boundary of \(N\) it is sufficient to find element \( \mathrm{w} \in W(E_8)\) such that  \(\mathrm{w}^{-1}(\mathbf{a}^{\mathbf{d}_i})= \mathbf{a}^{\mathbf{r}_{\sigma(i)}}\) for any \(0 \leq i \leq 9\) and \(\sigma\) be a permutation of \(\{1,\dots,9\}\). One of the possible solutions has the form 
\begin{equation}
	\mathrm{w}^{-1}=s_{1 2 3 4 5 6 7 6 5 4 3 2 0 3 4 5 6 1 2 3 4 0 5 3 4 2 3 1 2 0 3 4 5 6 7 3 4 5 6 2 3 4 5 1 2 0 3 0 4 3 2 3 1 0 5 4 3 2 4 3 0 6 5 4 3 0 2 5 4 3 6 5 4 3 1 2 0 3 0}.		
\end{equation}
We do not have an interpretation of this particular element. Only its existence is important for the self-duality Theorem~\ref{th:selfduality} and construction of the double affine Weyl group in Theorem~\ref{th:extended group Painleve}. This existence can be also shown in non constructive manner by comparing the scalar product of the vectors \(\mathbf{r}_0,\dots, \mathbf{r}_9\) and \(\mathbf{d}_0,\dots, \mathbf{d}_9\) in a root lattice.

\begin{Remark} \label{rem:E8 two polygons}
	One can also construct integrable system corresponding to \(q\)-Painlev\'e \(E_8^{(1)}\) (namely the phase space, Hamiltonian and spectral curve) starting from the pointed Painlev\'e polygon given by rectangle \(3\times 6\), see  Fig.~\ref{fi:E8 rectangle} left. The corresponding spectral curve agrees with one given in \cite[Fig. 2]{Moriyama:2021quantum}. 
	
	Another option is to start from triangle with side of length \(6\), see Fig.~\ref{fi:E8 rectangle} right. This agrees with the \cite[Fig 12]{Benini:2009webs}

	It is easy to see (and also follows from Theorem~\ref{th:KNP}) that these cases are related to Newton polygon on Fig.~\ref{fi:3*9graph} by mutations of pointed polygons.
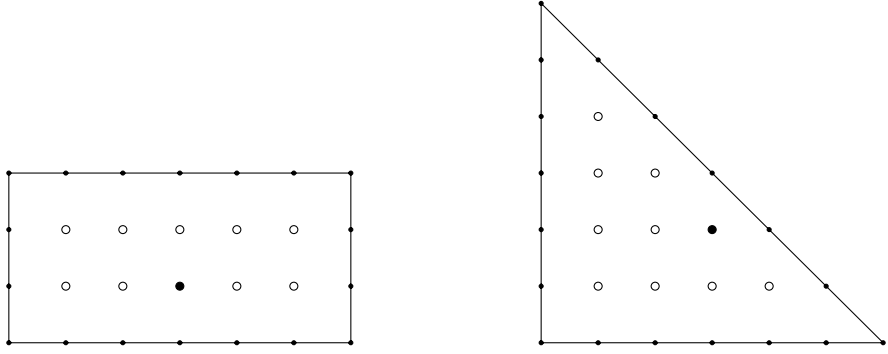
\begin{figure}[h]
	\begin{center}
		\begin{tikzpicture}[scale=0.75, font = \small]
			
			\draw[fill] (-3,2) circle (1pt) -- (-2,2) circle (1pt) -- (-1,2) circle (1pt) -- (0,2) circle (1pt) -- (1,2) circle (1pt) -- (2,2) circle (1pt) -- (3,2) circle (1pt) -- (3,1) circle (1pt) 
			-- (3,0) circle (1pt) -- (3,-1) circle (1pt)  -- (2,-1) circle (1pt) -- (1,-1) circle (1pt) -- (0,-1) circle (1pt) --(-1,-1) circle (1pt) --(-2,-1) circle (1pt) --(-3,-1) circle (1pt) --(-3,0) circle (1pt) --(-3,1) circle (1pt) -- (-3,2);
			\draw[fill] (0,0) circle (2pt);
			\draw (1,0) circle (2pt);
			\draw (1,1) circle (2pt);
			\draw (2,0) circle (2pt);
			\draw (2,1) circle (2pt);
			\draw (0,1) circle (2pt);
			\draw (-1,0) circle (2pt);
			\draw (-1,1) circle (2pt);
			\draw (-2,0) circle (2pt);
			\draw (-2,1) circle (2pt);
		\end{tikzpicture}
		\qquad \qquad\qquad
		\begin{tikzpicture}[scale=0.75, font = \small]
			
			\draw[fill] (-3,2) circle (1pt) -- (-3,3) circle (1pt) -- (-3,4) circle (1pt) -- (-3,5) circle (1pt) -- (-2,4) circle (1pt) -- (-1,3) circle (1pt) -- (0,2) circle (1pt) -- (1,1) circle (1pt) 
			-- (2,0) circle (1pt) -- (3,-1) circle (1pt)  -- (2,-1) circle (1pt) -- (1,-1) circle (1pt) -- (0,-1) circle (1pt) --(-1,-1) circle (1pt) --(-2,-1) circle (1pt) --(-3,-1) circle (1pt) --(-3,0) circle (1pt) --(-3,1) circle (1pt) -- (-3,2);
			\draw[fill] (0,1) circle (2pt);
			\draw (1,0) circle (2pt);
			\draw (0,0) circle (2pt);
			\draw (-1,0) circle (2pt);
			\draw (-1,1) circle (2pt);
			\draw (-1,2) circle (2pt);
			\draw (-2,0) circle (2pt);
			\draw (-2,1) circle (2pt);
			\draw (-2,2) circle (2pt);
			\draw (-2,3) circle (2pt);
		\end{tikzpicture}
		
		\caption{Two another pointed polygons of Painlev\'e $E_8^{(1)}$
		\label{fi:E8 rectangle}}
	\end{center}
\end{figure}

\end{Remark}

\bibliographystyle{alpha}
\addcontentsline{toc}{section}{\refname}  
\bibliography{Reduction}

\begin{thebibliography}{ACGK12}

\bibitem[ACGK12]{Akhtar:2012}
Mohammad Akhtar, Tom Coates, Sergey Galkin, and Alexander~M. Kasprzyk.
\newblock Minkowski polynomials and mutations.
\newblock {\em SIGMA Symmetry Integrability Geom. Methods Appl.}, 8:Paper 094,
  17, 2012.
\newblock {[\href{https://arxiv.org/abs/1212.1785}{arXiv:1212.1785}]}.

\bibitem[BBT09]{Benini:2009webs}
Francesco Benini, Sergio Benvenuti, and Yuji Tachikawa.
\newblock Webs of five-branes and $\mathcal{N}=2$ superconformal field
  theories.
\newblock {\em Journal of High Energy Physics}, 2009(09):052, 2009.
\newblock \href{https://arxiv.org/abs/0906.0359}{arXiv:0906.0359}.

\bibitem[BGM18]{Bershtein:2018cluster}
Mikhail Bershtein, Pavlo Gavrylenko, and Andrei Marshakov.
\newblock {Cluster integrable systems, q-Painlev{\'e} equations and their
  quantization}.
\newblock {\em Journal of High Energy Physics}, 2018(2), 2018.
\newblock [\href{http://arxiv.org/abs/1711.02063}{\texttt{arXiv:1711.02063}}].

\bibitem[Boc16]{Bocklandt:2016dimer}
Raf Bocklandt.
\newblock {A dimer abc}.
\newblock {\em Bulletin of the London Mathematical Society}, 48(3):387--451,
  2016.
\newblock [{\href{https://arxiv.org/abs/1510.04242}
  {\texttt{arXiv:1510.04242}}]}.

\bibitem[Bro12]{Broomhead:2010}
Nathan Broomhead.
\newblock Dimer models and {C}alabi-{Y}au algebras.
\newblock {\em Mem. Amer. Math. Soc.}, 215(1011):viii+86, 2012.
\newblock [\href{http://arxiv.org/abs/0901.4662}{\texttt{arXiv:0901.4662}}].

\bibitem[BS]{Inprog}
Mikhail Bershtein and Mykola Semenyakin.
\newblock {Cluster Hamiltonian reduction}.
\newblock in preparation.

\bibitem[CB04]{Crawley:2004indecomposable}
William Crawley-Boevey.
\newblock Indecomposable parabolic bundles: and the existence of matrices in
  prescribed conjugacy class closures with product equal to the identity.
\newblock {\em Publications Math{\'e}matiques de l'IH{\'E}S}, 100(1):171--207,
  2004.
\newblock
  [\href{http://arxiv.org/abs/math/0307246}{\texttt{arXiv:math/0307246}}].

\bibitem[CR07]{Cimasoni:2007dimers}
David Cimasoni and Nicolai Reshetikhin.
\newblock {Dimers on surface graphs and spin structures. I}.
\newblock {\em Communications in Mathematical Physics}, 275:187--208, 2007.
\newblock
  [\href{http://arxiv.org/abs/math-ph/0608070}{\texttt{arXiv:math-ph/0608070}}].

\bibitem[CS24]{Cremonesi:2023zig}
Stefano Cremonesi and Jos{\'e} S{\'a}.
\newblock Zig-zag deformations of toric quiver gauge theories. part i:
  reflexive polytopes.
\newblock {\em Journal of High Energy Physics}, 2024(5):1--62, 2024.
\newblock \href{https://arxiv.org/abs/2312.13909}{arXiv:2312.13909}.

\bibitem[CW22]{Casals:2022microlocal}
Roger Casals and Daping Weng.
\newblock {Microlocal theory of Legendrian links and cluster algebras}.
\newblock 2022.
\newblock
  {[\href{http://arxiv.org/abs/arXiv:2204.13244}{\texttt{arXiv:arXiv:2204.13244}}]}.

\bibitem[FG06a]{Fock:2006cluster}
Vladimir Fock and Alexander Goncharov.
\newblock {Cluster $\mathcal{X}$-varieties, amalgamation, and Poisson—Lie
  groups}.
\newblock {\em Algebraic Geometry and Number Theory: In Honor of Vladimir
  Drinfeld’s 50th Birthday}, pages 27--68, 2006.
\newblock
  {[\href{http://arxiv.org/abs/math/0508408}{\texttt{arXiv:math/0508408}}]}.

\bibitem[FG06b]{Fock:2006moduli}
Vladimir Fock and Alexander Goncharov.
\newblock {Moduli spaces of local systems and higher Teichm{\"u}ller theory}.
\newblock {\em Publications Math{\'e}matiques de l'IH{\'E}S}, 103:1--211, 2006.
\newblock
  {[\href{http://arxiv.org/abs/math/0311149}{\texttt{arXiv:math/0311149}}]}.

\bibitem[FHKV08]{Feng:2008dimer}
Bo~Feng, Yang-Hui He, Kristian~D. Kennaway, and Cumrun Vafa.
\newblock Dimer models from mirror symmetry and quivering amoebae.
\newblock {\em Adv. Theor. Math. Phys.}, 12(3):489--545, 2008.
\newblock [{\href{https://arxiv.org/abs/hep-th/0511287}
  {\texttt{arXiv:hep-th/0511287}}]}.

\bibitem[Flo14]{Floater:2014wachspress}
Michael~S Floater.
\newblock Wachspress and mean value coordinates.
\newblock In {\em Approximation Theory XIV: San Antonio 2013}, pages 81--102.
  Springer, 2014.

\bibitem[FM16]{Fock:2016}
Vladimir Fock and Andrei Marshakov.
\newblock Loop groups, clusters, dimers and integrable systems.
\newblock In {\em Geometry and quantization of moduli spaces}, Adv. Courses
  Math. CRM Barcelona, pages 1--66. Birkh\"{a}user/Springer, Cham, 2016.
\newblock [\href{http://arxiv.org/abs/1401.1606}{\texttt{arXiv:1401.1606}}].

\bibitem[FRG24]{Franco:2023quiver}
Sebastian Franco and Diego Rodr{\'\i}guez-G{\'o}mez.
\newblock Quiver tails and brane webs.
\newblock {\em Journal of High Energy Physics}, 2024(10):1--38, 2024.
\newblock \href{https://arxiv.org/abs/2310.10724}{arXiv:2310.10724}.

\bibitem[FS23]{Franco:2023twin}
Sebasti{\'a}n Franco and Rak-Kyeong Seong.
\newblock {Twin theories, polytope mutations and quivers for GTPs}.
\newblock {\em Journal of High Energy Physics}, 2023(7):1--44, 2023.
\newblock \href{https://arxiv.org/abs/2302.10951}{arXiv:2302.10951}.

\bibitem[GG24]{Galashin:2022move}
Pavel Galashin and Terrence George.
\newblock Move-reduced graphs on a torus.
\newblock {\em Transactions of the American Mathematical Society},
  377(06):4055--4099, 2024.
\newblock
  {[\href{http://arxiv.org/abs/2212.12962}{\texttt{arXiv:2212.12962}}]}.

\bibitem[GGK23]{George2022inverse}
Terrence George, Alexander Goncharov, and Richard Kenyon.
\newblock The inverse spectral map for dimers.
\newblock {\em Mathematical Physics, Analysis and Geometry}, 26(3):24, 2023.
\newblock [\href{http://arxiv.org/abs/2207.10146}{\texttt{arXiv:2207.10146}}].

\bibitem[GHK15]{Gross:2013birational}
Mark Gross, Paul Hacking, and Sean Keel.
\newblock Birational geometry of cluster algebras.
\newblock {\em Algebr. Geom.}, 2(2):137--175, 2015.
\newblock [{\href{https://arxiv.org/abs/1309.2573}
  {\texttt{arXiv:1309.2573}}]}.

\bibitem[GI24]{George:2019cluster}
Terrence George and Giovanni Inchiostro.
\newblock The cluster modular group of the dimer model.
\newblock {\em Annales de l'Institut Henri Poincar{\'e} D}, 11(1), 2024.
\newblock
  {[\href{http://arxiv.org/abs/1909.12896}{\texttt{arXiv:1909.12896}}]}.

\bibitem[GK13]{Goncharov:2013}
Alexander~B. Goncharov and Richard Kenyon.
\newblock Dimers and cluster integrable systems.
\newblock {\em Ann. Sci. \'{E}c. Norm. Sup\'{e}r. (4)}, 46(5):747--813, 2013.
\newblock [\href{http://arxiv.org/abs/1107.5588}{\texttt{arXiv:1107.5588}}].

\bibitem[GR23]{George:2022discrete}
Terrence George and Sanjay Ramassamy.
\newblock {Discrete dynamics in cluster integrable systems from geometric
  $R$-matrix transformations}.
\newblock {\em Comb. Theory}, 3(2):Paper No. 12, 29, 2023.
\newblock
  {[\href{http://arxiv.org/abs/2208.10306}{\texttt{arXiv:2208.10306}}]}.

\bibitem[GS18]{Goncharov:2018donaldson}
Alexander Goncharov and Linhui Shen.
\newblock {Donaldson--Thomas transformations of moduli spaces of G-local
  systems}.
\newblock {\em Advances in Mathematics}, 327:225--348, 2018.
\newblock [{\href{https://arxiv.org/abs/1602.06479}
  {\texttt{arXiv:1602.06479}}]}.

\bibitem[GU10]{Galkin:2010mutations}
Sergey Galkin and Alexandr Usnich.
\newblock Mutations of potentials, 2010.
\newblock preprint IPMU 10-0100.

\bibitem[Gul08]{Gulotta:2008properly}
Daniel~R Gulotta.
\newblock {Properly ordered dimers, $R$-charges, and an efficient inverse
  algorithm}.
\newblock {\em Journal of High Energy Physics}, 2008(10):014, 2008.
\newblock {[\href{http://arxiv.org/abs/0807.3012}{\texttt{arXiv:0807.3012}}]}.

\bibitem[HN22]{Higashitani:2022}
Akihiro Higashitani and Yusuke Nakajima.
\newblock Deformations of dimer models.
\newblock {\em SIGMA Symmetry Integrability Geom. Methods Appl.}, 18:Paper No.
  030, 53, 2022.
\newblock {[\href{https://arxiv.org/abs/1903.01636}{arXiv:1903.01636}]}.

\bibitem[HS12]{Hanany:2012brane}
Amihay Hanany and R-K Seong.
\newblock Brane tilings and reflexive polygons.
\newblock {\em Fortschritte der Physik}, 60(6):695--803, 2012.
\newblock [{\href{https://arxiv.org/abs/1201.2614}
  {\texttt{arXiv:1201.2614}}]}.

\bibitem[ILP16]{Inoue:2016toric}
Rei Inoue, Thomas Lam, and Pavlo Pylyavskyy.
\newblock {Toric networks, geometric R-matrices and generalized discrete Toda
  lattices}.
\newblock {\em Communications in Mathematical Physics}, 347:799--855, 2016.
\newblock
  {[\href{http://arxiv.org/abs/1504.03448}{\texttt{arXiv:1504.03448}}]}.

\bibitem[ILP19]{Inoue:2019cluster}
Rei Inoue, Thomas Lam, and Pavlo Pylyavskyy.
\newblock {On the cluster nature and quantization of geometric $R$-matrices}.
\newblock {\em Publications of the Research Institute for Mathematical
  Sciences}, 55(1), 2019.
\newblock
  {[\href{http://arxiv.org/abs/1607.00722}{\texttt{arXiv:1607.00722}}]}.

\bibitem[IS06]{Ion:2006Triple}
Bogdan Ion and Siddhartha Sahi.
\newblock Triple groups and {C}herednik algebras.
\newblock In {\em Jack, {H}all-{L}ittlewood and {M}acdonald polynomials},
  volume 417 of {\em Contemp. Math.}, pages 183--206. Amer. Math. Soc.,
  Providence, RI, 2006.
\newblock
  {[\href{http://arxiv.org/abs/math/0304186}{\texttt{arXiv:math/0304186}}]}.

\bibitem[IU11]{Ishii:2011note}
Akira Ishii and Kazushi Ueda.
\newblock A note on consistency conditions on dimer models.
\newblock In {\em Higher dimensional algebraic geometry}, volume B24 of {\em
  RIMS K\^{o}ky\^{u}roku Bessatsu}, pages 143--164. Res. Inst. Math. Sci.
  (RIMS), Kyoto, 2011.
\newblock [\href{http://arxiv.org/abs/1012.5449}{\texttt{arXiv:1012.5449}}].

\bibitem[IU15]{Ishii:2009:2}
Akira Ishii and Kazushi Ueda.
\newblock Dimer models and the special {M}c{K}ay correspondence.
\newblock {\em Geom. Topol.}, 19(6):3405--3466, 2015.
\newblock [\href{http://arxiv.org/abs/0905.0059}{\texttt{arXiv:0905.0059}}].

\bibitem[Kas63]{Kasteleyn:1963}
P.~W. Kasteleyn.
\newblock Dimer statistics and phase transitions.
\newblock {\em J. Mathematical Phys.}, 4:287--293, 1963.

\bibitem[Kho78]{Khovanskii:1978}
Askold Khovanski\u{\i}.
\newblock Newton polyhedra, and the genus of complete intersections.
\newblock {\em Funktsional. Anal. i Prilozhen.}, 12(1):51--61, 1978.

\bibitem[KNP17]{Kasprzyk:2017minimality}
Alexander Kasprzyk, Benjamin Nill, and Thomas Prince.
\newblock Minimality and mutation-equivalence of polygons.
\newblock In {\em Forum of mathematics, Sigma}, volume~5, page e18. Cambridge
  University Press, 2017.
\newblock
  {[\href{http://arxiv.org/abs/1501.05335}{\texttt{arXiv:1501.05335}}]}.

\bibitem[KNY17]{KNY15}
Kenji Kajiwara, Masatoshi Noumi, and Yasuhiko Yamada.
\newblock Geometric aspects of {P}ainlev\'{e} equations.
\newblock {\em Journal of Physics A: Mathematical and Theoretical},
  50(7):073001, Jan 2017.
\newblock [{\href{https://arxiv.org/abs/1509.08186}
  {\texttt{arXiv:1509.08186}}]}.

\bibitem[KO06]{Kenyon:2006planar}
Richard Kenyon and Andrei Okounkov.
\newblock Planar dimers and {H}arnack curves.
\newblock {\em Duke Math. J.}, 131(3):499--524, 2006.
\newblock \href{https://arxiv.org/abs/math/0311062}{arXiv:math/0311062}.

\bibitem[Mar13]{Marshakov:2013}
A.~Marshakov.
\newblock Lie groups, cluster variables and integrable systems.
\newblock {\em Journal of Geometry and Physics}, 67:16–36, 2013.
\newblock [\href{http://arxiv.org/abs/1207.1869}{\texttt{arXiv:1207.1869}}].

\bibitem[Miz24]{Mizuno}
Yuma Mizuno.
\newblock {$q$-Painlevé equations on cluster Poisson varieties via toric
  geometry}.
\newblock {\em Selecta Mathematica}, 30(2):19, 2024.
\newblock [\href{http://arxiv.org/abs/2008.11219}{\texttt{arXiv:2008.11219}}].

\bibitem[MOT23]{Masuda:2023birational}
Tetsu Masuda, Naoto Okubo, and Teruhisa Tsuda.
\newblock {Birational Weyl group actions via mutation combinatorics in cluster
  algebras}.
\newblock 2023.
\newblock [\href{http://arxiv.org/abs/2303.06704}{\texttt{arXiv:2303.06704}}].

\bibitem[MS19]{MS:2019}
Andrei Marshakov and Mykola Semenyakin.
\newblock {Cluster integrable systems and spin chains}.
\newblock {\em Journal of High Energy Physics}, 2019(100), 2019.
\newblock [\href{http://arxiv.org/abs/1905.09921}{\texttt{arXiv:1905.09921}}].

\bibitem[MY21]{Moriyama:2021quantum}
Sanefumi Moriyama and Yasuhiko Yamada.
\newblock {Quantum representation of affine Weyl groups and associated quantum
  curves}.
\newblock {\em SIGMA. Symmetry, Integrability and Geometry: Methods and
  Applications}, 17:076, 2021.
\newblock [\href{http://arxiv.org/abs/2104.06661}{\texttt{arXiv:2104.06661}}].

\bibitem[Oku15]{Okubo:2015bilinear}
Naoto Okubo.
\newblock {Bilinear equations and q-discrete Painlev{\'e} equations satisfied
  by variables and coefficients in cluster algebras}.
\newblock {\em Journal of Physics A: Mathematical and Theoretical},
  48(35):355201, 2015.
\newblock {[\href{http://arxiv.org/abs/1505.03067}{{\tt arXiv:1505.03067}}]}.

\bibitem[Oku17]{Okubo:2017coprimeness}
Naoto Okubo.
\newblock {Co-primeness preserving higher dimensional extension of q-discrete
  Painleve I, II equations}.
\newblock 2017.
\newblock {[\href{http://arxiv.org/abs/1704.05403}{{\tt arXiv:1704.05403}}]}.

\bibitem[Rui90]{Ruijsenaars:1990Toda}
S.~N.~M. Ruijsenaars.
\newblock Relativistic {T}oda systems.
\newblock {\em Comm. Math. Phys.}, 133(2):217--247, 1990.

\bibitem[Sak01]{Sakai:2001}
Hidetaka Sakai.
\newblock Rational surfaces associated with affine root systems and geometry of
  the {P}ainlev\'{e} equations.
\newblock {\em Comm. Math. Phys.}, 220(1):165--229, 2001.

\bibitem[Sev99]{Sevostyanov:1999}
Alexey Sevostyanov.
\newblock Regular nilpotent elements and quantum groups.
\newblock {\em Comm. Math. Phys.}, 204(1):1--16, 1999.
\newblock
  {[\href{http://arxiv.org/abs/math/9812107}{\texttt{arXiv:math/9812107}}]}.

\bibitem[Sim91]{Simpson:1992products}
Carlos~T. Simpson.
\newblock Products of matrices.
\newblock In {\em Differential geometry, global analysis, and topology
  ({H}alifax, {NS}, 1990)}, volume~12 of {\em CMS Conf. Proc.}, pages 157--185.
  Amer. Math. Soc., Providence, RI, 1991.

\bibitem[SO20]{Suzuki:2020}
Takao Suzuki and Naoto Okubo.
\newblock Cluster algebra and {$q$}-{P}ainlev\'{e} equations: higher order
  generalization and degeneration structure.
\newblock In {\em Mathematical structures of integrable systems and their
  applications}, volume B78 of {\em RIMS K\^{o}ky\^{u}roku Bessatsu}, pages
  53--75. Res. Inst. Math. Sci. (RIMS), Kyoto, 2020.

\bibitem[ST97]{Saito:1997}
Kyoji Saito and Tadayoshi Takebayashi.
\newblock Extended affine root systems. {III}. {E}lliptic {W}eyl groups.
\newblock {\em Publ. Res. Inst. Math. Sci.}, 33(2):301--329, 1997.

\bibitem[Wac75]{Wachspress:1975}
Eugene~L. Wachspress.
\newblock {\em A rational finite element basis}, volume Vol. 114 of {\em
  Mathematics in Science and Engineering}.
\newblock Academic Press, Inc. [Harcourt Brace Jovanovich, Publishers], New
  York-London, 1975.

\bibitem[Wal04]{Wall:2004singular}
Charles Terence~Clegg Wall.
\newblock {\em Singular points of plane curves}.
\newblock Number~63. Cambridge University Press, 2004.

\end{thebibliography}

\end{document}